\documentclass[10pt,usenames,dvipsnames]{amsart}
\usepackage{epsf,latexsym,amsmath,amssymb,amscd,datetime}
\usepackage{amsmath,amsthm,amssymb,enumerate,eucal,url,calligra,mathrsfs,bbm}

\usepackage{blkarray} 

\usepackage{tikz,scalerel}
\usetikzlibrary{arrows}

\usepackage{imakeidx}     

\usepackage[outdir=./]{epstopdf}

\usepackage{graphicx}
\usepackage{color}
\newenvironment{jfnote}{ \bgroup \color{blue} }{\egroup}

\newcommand{\oldStuff}[1]{}

\DeclareMathOperator{\SHom}{\mathscr{H}\text{\kern -3pt {\calligra\large om}}\,}

\IfFileExists{my_xrefs}{\input my_xrefs}{}

\newcommand{\naturals}{{\mathbb N}}

\newcommand{\mec}[1]{{\bf #1}}	

\usepackage{mathrsfs}
\usepackage{amssymb}
\usepackage{dsfont}
\usepackage{verbatim}
\usepackage{url}

\theoremstyle{plain}
\ifundef\chapter
  {\newtheorem{theorem}{Theorem}[section]}
  {\newtheorem{theorem}{Theorem}[section]}
\ifundef\chapter
  { }
  {}

\ifundef\chapter
  { }
  {\numberwithin{equation}{section}
   
  }

\newtheorem{lemma}[theorem]{Lemma}
\newtheorem{proposition}[theorem]{Proposition}
\newtheorem{corollary}[theorem]{Corollary}

\newtheorem*{theorem*}{Theorem}  
\newtheorem*{lemma*}{Lemma}
\newtheorem*{corollary*}{Corollary}
\newtheorem*{proposition*}{Proposition}

\theoremstyle{definition}
\newtheorem{definition}[theorem]{Definition}

\newtheorem{xca}{Exercise}[section]

\newtheorem{example}[theorem]{Example}

\newtheorem{notation}[theorem]{Notation}

\newtheorem{remark}[theorem]{Remark}

\newcommand{\subspace}{\subset}

\newcommand{\ignore}[1]{}

\newcommand{\field}{{\mathbb F}}

\newcommand{\reals}{{\mathbb R}}

\newcommand{\integers}{{\mathbb Z}}

\DeclareMathAlphabet{\mathcal}{OMS}{cmsy}{m}{n}

\newcommand\cA{\mathcal{A}}
\newcommand\cB{\mathcal{B}}
\newcommand\cC{\mathcal{C}}
\newcommand\cD{\mathcal{D}}
\newcommand\cE{\mathcal{E}}

\newcommand\cL{\mathcal{L}}

\newcommand\cS{\mathcal{S}}

\newcommand\cU{\mathcal{U}}

\newcommand\cW{\mathcal{W}}
\newcommand\cX{\mathcal{X}}
\newcommand\cY{\mathcal{Y}}

\def\from{\colon}

\def\eqdef{\overset{\text{def}}{=}}

\DeclareMathOperator{\Span}{Span}

\def\implies{\Rightarrow}

\DeclareRobustCommand
  \rddots{\mathinner{\mkern1mu\raise\p@
    \vbox{\kern7\p@\hbox{.}}\mkern2mu
    \raise4\p@\hbox{.}\mkern2mu\raise7\p@\hbox{.}\mkern1mu}}

\newcommand\xhookrightarrow[2][]{\ext@arrow 0062{\hookrightarrowfill@}{#1}{#2}}
\def\hookrightarrowfill@{\arrowfill@\lhook\relbar\rightarrow}

\usepackage{relsize}
\usepackage{tikz}
\usetikzlibrary{matrix,arrows,decorations.pathmorphing}
\usepackage{tikz-cd}
\usetikzlibrary{cd}

\usepackage[pdftex,colorlinks,linkcolor=blue,citecolor=purple]{hyperref}

\newcommand{\mygreen}{}  

\newcommand{\Awithout}[1]{A_{\,\widehat{#1}}}

\begin{document}

\title[Coordination and Discoordination] 
{Coordination and Discoordination in Linear Algebra,
Linear Information Theory, and
Coded Caching}

\author{Joel Friedman}
\address{Department of Computer Science, 
        University of British Columbia, Vancouver, BC\ \ V6T 1Z4, CANADA}
\email{{\tt jf@cs.ubc.ca}}
\thanks{Research supported in part by an NSERC grant.}

\author{Amir Tootooni}
\address{Department of Computer Science,
        University of British Columbia, Vancouver, BC\ \ V6T 1Z4, CANADA}
\email{{\tt tootooniamirhossein@gmail.com}}
\thanks{Research supported in part by an NSERC grant.}




\begin{abstract}
In the first part of this
paper we develop some theorems in linear algebra
applicable to information theory 
when all random variables involved are linear functions
of the individual bits of a source of independent bits.

We say that a collection of subspaces of a vector space are
{\em coordinated} if the vector space has a basis such that
each subspace is spanned by its intersection with the basis.
We measure the failure of a collection of subspaces to be coordinated by an
invariant that we call the {\em discoordination} of the family.
We develop some foundational results regarding discoordination.
In particular, these results give a number of new formulas
involving three subspaces of a vector space. 

We then apply a number of our results, along with a method of
Tian in \cite{tian_2018_computer_aided},
to obtain some new lower bounds
in a special case of the basic coded caching problem.
In terms of the usual notation for these problems,
we show that for $N=3$ documents and
$K=3$ caches,
we have $6M+5R\ge 11$ for 
a scheme that achieves the memory-rate pair $(M,R)$,
assuming the scheme is linear.
\end{abstract}

\maketitle
\setcounter{tocdepth}{3}
\tableofcontents

\newcommand{\sePrelimProofs}{17}

\newcommand{\Univ}{{\cU}}

\newenvironment{claim}[1]{\par\noindent{Claim.}\space#1}{}

\section{Introduction}
\label{se_Joel_intro}

In this article we develop some tools in what might call
{\em linear information theory}, by which we mean information theory that
assumes that all random variables under discussion
are linear functions of a source.
There are a number of reasons to restrict a problem in information
theory to the special case of linear random
variables: first, in algorithms, it is often much simpler and
more practical to work with linear functions than non-linear ones.
Second, linear functions are often optimal or nearly optimal in terms
of the objectives of a problem.
Third, if we cannot completely solve a problem in information theory,
a good starting point would be to solve it under the assumption of
linearity, and then address the general case.

Theoretically, questions in information about linear random variables can be 
stated in terms of unknown matrices; however, in many applications,
the usual tools of matrix analysis and linear algebra do not suffice.
In this article we develop a new set of tools in linear algebra
regarding what we call {\em coordinated} subspaces and the
{\em discoordination} of a family of subspaces;
we then give an application to linear information theory.

Some of the main tools we develop in linear algebra concern formulas
involving three subspaces of a vector space, and their
{\em discoordination}, an invariant that allows us to write
many new formulas regarding the dimensions of subspaces obtained
by taking the three subspaces and repeatedly taking sums and intersections.
More generally, we develop a number of theorems regarding the
{\em discoordination} of any number of subspaces of a vector space.
We then apply these theorems to certain collections of three random
variables to obtain some
partial results on one instance of the problem of ``coded caching,''
a problem in information theory initiated 
by Maddah-Ali and Niesen \cite{MA_niesen_2014_seminal}
that has received a lot
of attention (see \cite{MA_avestimehr_2019_survey,
tian_2018_computer_aided,saberali_thesis,tootooni_msc_thesis}
and the many references therein).

The reader primarily interested in information theory can understand
our linear algebra theorems in the following way:
information theory often exploits the concept of the {\em mutual information}
of two random variables, due to the many useful properties it satisfies.
By contrast, the mutual information of three random variables is
seldomly used to produce bounds in applications,
due to the fact that it is much worse behaved: for example,
it can be {\mygreen positive, zero, or} negative.
However, if three random variables are linear
functions of a source, then we will show that there is a simple
formula for their mutual information, namely as the dimension
of their intersection minus
their discoordination alluded to above.
Our application to coded caching will not use this particular formula,
but exploits 
related formulas involving three subspaces and their discoordination,
along with some of the general theory of coordination and discoordination
that we develop in this article.

We have written this article assuming a minimal background in
linear algebra and information theory, to be readable to a wider
audience.
We hope to interest information theorists in the mathematical tools
we introduce, which may have other applications. 
Also, we have mildly simplified the
usual {\em coded caching} problem, so that it requires less background
to formally state;
we believe the coded caching problems deserve a wide mathematical audience and
likely have applications beyond caching per se.

We emphasize that the linear algebra required to read this paper is
no more than a typical one-term introductory ``honors''
(i.e., abstract) linear algebra course, as in
\cite{janich,axler}.
However, we will briefly review this background, as well as briefly
review information theory;
most of these ideas are common in the literature, although terminology
and notation differ.
The second author's MSc.~thesis \cite{tootooni_msc_thesis} contains
some additional details and references.

We next describe some of the main results in this paper in rough
terms; the formal mathematical definitions will appear in
Section~\ref{se_Joel_basics}, the main results in linear
algebra will be stated in Section~\ref{se_Joel_coord_discoord_defs},
and we give a more precise statement of coded caching in
Section~\ref{se_Joel_coded_intro}.

We remark that the focus of this article is on three subspaces of a
vector space, i.e., three linear random variables, and there are many
more open questions regarding this situation and that of four or more
subspaces or of random variables.
Hence we believe that the study of
coordination and discoordination will likely have more
applications and merits further study.

{\mygreen
During the 
revision of this article, Chao Tian pointed out to us:
(1) the memory-rate tradeoff
$(1/2,5/3)$ of Chapter~10 has appeared in \cite{gomez_vilardebo}
(Corollary~1.1 there, page~4490, with $N=K=3$ and $q=2$),
and (2) optimal memory-rate tradeoffs for the
linear problem of coded-caching $N=K=3$,
which we study in Sections~7--13, have been determined
by Cao and Xu, using computer-aided methods, in a preprint
\cite{cao_xu_n_k_three}.
}

\subsection{Main Results in Linear Algebra}

The linear algebra we develop generalizes what is often called the
``dimension formula,'' that states that for vector subspaces $A_1,A_2$
of some finite-dimensional vector space,
$\cU$, we have
$$
\dim(A_1\cap A_2) = I(A_1;A_2),
$$
where
$$
I(A_1;A_2) \eqdef \dim(A_1)+\dim(A_2) - \dim(A_1+A_2),
$$
where $A_1+A_2$ denotes the sum (or the span) of $A_1$ and $A_2$
in the ambient vector space $\cU$.  Of course, the vector spaces
$A_1\cap A_2$ and $A_1+A_2$ are not intrinsic to the isomorphism
classes of $A_1$ and $A_2$,
but depends on the way they are related to each other
in the ambient vector space, $\cU$.

The reader familiar with information theory will recognize
$I(A_1;A_2)$ as the {\em mutual information} of $A_1,A_2$, when
viewing them as random variables of a source that is
the dual space of $\cU$.

One of our main results in linear algebra concerns three subspaces
$A_1,A_2,A_3\subset \cU$, and quantity
\begin{align}
\label{eq_three_way_mutual_eq1}
I(A_1;A_2;A_3) & =\dim(A_1+A_2+A_3) - \dim(A_1+A_2) - \dim(A_1+A_3)
\\
\label{eq_three_way_mutual_eq2}
& 
- \dim(A_2+A_3)  
+ \dim(A_1) + \dim(A_2) + \dim(A_3) 
\end{align}
which, using the dimension formula, can also be written as 
$$
\dim(A_1 \cap A_2) + \dim(A_1 \cap A_3) - \dim\bigl(A_1 \cap (A_2 + A_3)\bigr) 
$$
(typically $I(A_1;A_2;A_3)$
is called the {\em (three-way) mutual information} of $A_1,A_2,A_3$
in information theory).
It is well known that in contrast to the dimension formula,
$I(A_1;A_2;A_3)$ does not generally equal $\dim(A_1\cap A_2\cap A_3)$.
The equation does hold if the $A_1,A_2,A_3$ are {\em coordinated}
in the sense that they have a {\em coordinating basis}, meaning
a basis, $X$, of $\cU$, such that
for $i=1,2,3$ the vectors $A_i\cap X$ span $A_i$;
in this case $I(A_1;A_2;A_3)$ equals
$\dim(A_1,A_2,A_3)$, which is naively what 
``mutual information'' is trying to capture.
A simple example where the three subspace analog of the dimension
formula fails to hold, i.e., where
$$
I(A_1;A_2;A_3) \ne \dim(A_1\cap A_2\cap A_3) ,
$$
is
for $\cU=\field^2$ for {\mygreen an} arbitrary field, $\field$, and
\begin{equation}\label{eq_fundamental_counterexample}
A_1={\rm Span}(e_1),
\ A_2={\rm Span}(e_2),
\ A_3={\rm Span}(e_1+e_2),
\end{equation} 
where $e_1,e_2$ are the standard basis vectors; in this case
$A_1\cap A_2\cap A_3 = \{0\}$ but
$I(A_1;A_2;A_3)=-1$.
One fundamental result in this article is that 
\eqref{eq_fundamental_counterexample} is essentially {\em the only} example
where this formula fails:
more precisely, if $A_1,A_2,A_3\subset\cU$ are three subspaces of
a finite dimensional $\field$-vector space, $\cU$, then we may
decompose $\cU$ as a direct sum of subspaces $\cU_1$ and $\cU_2$,
through which $A_1,A_2,A_3$ factor
(``factor'' here is analogous to how a linear operator on a vector space
factors through its generalized eigenspaces),
such that 
\begin{enumerate}
\item 
$A_1\cap\cU_1,A_2\cap\cU_1,A_3\cap\cU_1$ are coordinated, 
and 
\item 
there is an
isomorphism $\iota\from\cU_2\to \field^2\otimes\field^m$ for some
$m\ge 0$, under which $\iota$ applied to the restriction of the
$A_1,A_2,A_3$ equals 
\begin{equation}\label{eq_basic_counterexample}
\iota\bigl( A_1\cap\cU_2 \bigr) 
=\{e_1\}\otimes\field^m,
\ \iota\bigl( A_2\cap\cU_2 \bigr) 
=\{e_2\}\otimes\field^m,
\ \iota\bigl( A_3\cap\cU_2 \bigr) 
=\{e_1+e_2\}\otimes\field^m.
\end{equation} 
\end{enumerate}
The integer $m$ is uniquely determined, and we will prove that it equals
the {\em discoordination} of $A_1,A_2,A_3$, which we define for any
number of subspaces $A_1,\ldots,A_m$ as
$$
{\rm DisCoord}(A_1,\ldots,A_m) 
{\mygreen \eqdef} \min_{X\in {\rm Ind}(\cU)} 
\sum_{i=1}^m
\bigl( \dim(A_i) - |X\cap A_i| \bigr),
$$
where the minimum is taken over all $X$ that are linearly independent subsets
of $\cU$, and we use ${\rm Ind}(\cU)$ to denote the set of all such $X$.
We easily see that 
$$
{\rm DisCoord}(A_1,\ldots,A_m) \ge 0,
$$
with equality iff
$A_1,\ldots,A_m$ are coordinated.

The above theory will imply that
$$
I(A_1;A_2;A_3) = \dim(A_1\cap A_2\cap A_3) - {\rm DisCoord}(A_1,A_2,A_3)
\le \dim(A_1\cap A_2\cap A_3),
$$
and hence $I(A_1;A_2;A_3) = \dim(A_1\cap A_2\cap A_3)$ iff $A_1,A_2,A_3$
are coordinated.  

In this paper we will give some fundamental theorems regarding
coordination and discoordination, many of which we use to study
the discoordination of three subspaces.

We remark that discoordination is at the heart of the failure
of a number of ``would be'' 
desirable properties in linear algebra, and it likely arises
in many places in mathematics: for example, the first
author has encountered this in the study of
``2-independence,'' 
a set of questions in linear algebra (that is equivalent
to a question about sheaves on a graph with two vertices);
it turns out that if all vector spaces involved are coordinated subspaces
of some ambient space, then the questions regarding
$2$-independence are easy to answer;
see \cite{alice_thesis}.
See also \cite{lafforgue,lafforgue_erratum} as another place where the
discoordination of three subspaces arises.

{\mygreen
\subsection{Additional
Historical Context of Discoordination of Three Subspaces}

Let us indicate the connection between the discoordination of three
subspaces to
``information inequalities'' in information theory
and to ``representable matroids.''

We remark that certain well-known facts imply
that for any subspaces $A,B,C$ of a finite dimensional vector space, $\cU$,
the value of ${\rm DisCoord}(A,B,C)$
has no bearing on what are commonly called 
``information inequalities'' or
(following Nick Pippenger)
the ``laws of information theory'' for $A,B,C$ (see
Chapter~12 of the textbook \cite{yeung_textbook} for definitions and
references):
to elaborate, these terms refer to any linear inequalities involving the
dimensions of
$$
A,\ B,\ C, 
\ A+B,\ A+C,\ B+C,\ A+B+C .
$$
It is well-known \cite{hammerEtAl2000}, Theorem~3, page~453,
that all such inequalities are generated by ``nine basic inequalities''
(equations~(9) in \cite{hammerEtAl2000})
that can be deduced by considering
coordinated subspaces alone
(see Figure~1, page~454, of \cite{hammerEtAl2000}).
By contrast, there are connections
between ${\rm DisCoord}(A,B,C)$ and these ``nine basic inequalities:''
namely, one these nine inequalities is
\begin{equation}\label{eq_three_of_nine_basic}
\dim(A) + \dim(A + B + C) \le \dim(A+B) + \dim(A+C) 
\end{equation} 
(and two others are obtained by permuting $A,B,C$).
We will show, in 
Corollary~\ref{co_various_formulas_for_discoordination_of_three}
(item~(6) there),
that
\begin{equation}\label{eq_three_of_nine_basic_as_discoordination}
\dim(A+B) + \dim(A+C) - \dim(A) - \dim(A + B + C)
=
\dim\bigl( (B\cap C)/A \bigr) +
{\rm DisCoord}(A,B,C) 
\end{equation} 
We note that the left-hand-side of
\eqref{eq_three_of_nine_basic_as_discoordination}
is commonly written as $I(B;C|A)$ in the context of information theory
(see, e.g.,
Section~1 of both \cite{dougherty_five,dougherty_six}).
Hence ${\rm DisCoord}(A,B,C)$ is 
related to ``information inequalities.''
Of course, the dimension formula implies that
$$
\dim(A+B) + \dim(A+C) - \dim(A) - \dim(A + B + C) 
= \dim\bigl( (A+B)\cap (A+C ) \bigr) - \dim(A) ,
$$
and hence 
the right- and left-hand-sides of 
\eqref{eq_three_of_nine_basic_as_discoordination}
can be written as
\begin{equation}\label{eq_two_quantities_to_compare_one_with_discoord}
\dim\bigl( (B\cap C)/A \bigr) +
{\rm DisCoord}(A,B,C)
=
\dim\bigl( (A+B)\cap (A+C ) \bigr) - \dim(A) 
\end{equation} 
The right-hand-side of 
\eqref{eq_two_quantities_to_compare_one_with_discoord}
has the advantage of looking simpler than the left-hand-side.
The advantage of the left-hand-side of 
\eqref{eq_two_quantities_to_compare_one_with_discoord},
is that 
it is the sum of two more ``elemental'' terms,
each of which is non-negative and has a simple
meaning with the ideas we develop in this paper:
${\rm DisCoord}(A,B,C)$ will be a focus of much of this paper,
and $\dim((B+C)/A)$ corresponds to one piece in the
``Venn diagram'' of $I,J,K$ in the
case where $A=e_I$, $B=e_J$, $C=e_K$ are coordinated, 
namely the $(J\cap K)\setminus I$ piece, and so
$$
\dim\bigl( (B+C)/A \bigr)
=\dim^{\cU/A}\bigl( [B\cap C]_A ) = \bigl|(J\cap K)\setminus I \bigr| .
$$
See the discussion around 
Corollary~\ref{co_various_formulas_for_discoordination_of_three}
and Figure~\ref{fi_venn_diagram_balanced}.

Discoordination has a similar---but perhaps more direct---connection to
{\em Ingleton's inequality} 
\cite{hammerEtAl2000,yeung_textbook,matus1999III,
kinser2011,nelsonEtAl2018,dougherty_five,dougherty_six}:
Ingleton \cite{ingleton1971} in Section~4 proves a theorem
for {\em representable
matroids}\footnote{\mygreen
  Roughly speaking, a
  {\em representable matroid} \cite{ingleton1971} refers to a matroid
  that can be realized by elements of a right-vector space
  over a division algebra (i.e., a skew field) $\Delta$.
  A lot of recent literature, e.g., \cite{dougherty_six,nelsonEtAl2018} and
  the references therein, focuses on the
  case where $\Delta$ is a field.
  Note that Example~1, page~153 of \cite{ingleton1971}, i.e., the 
  ``non-Pappas
  matroid,'' shows that some matroids are representable, but only
  when
  $\Delta$ is not a field (i.e., $\Delta$ is not commutative).
  }; 
the proof in \cite{ingleton1971} uses the fact (see equations~(2)
on page 159) that
$$
\dim\bigl( (A+C)\cap (B+C) \bigr) \ge \dim\bigl( (A\cap B) + C \bigr);
$$
in our context, item~(5) of our
Corollary~\ref{co_various_formulas_for_discoordination_of_three}
shows that
$$
{\rm DisCoord}(A,B,C) = 
\dim\bigl( (A+C)\cap (B+C) \bigr) - \dim\bigl( (A\cap B) + C \bigr).
$$
Hence
${\rm DisCoord}(A,B,C)$ is directly related to Ingleton's derivation
of his inequality \cite{ingleton1971}.

We would be interested to know if there are further connections
between discoordination (of three or more subspaces) and its properties,
to the fields of
``information inequalities'' and
``representable matroids.''
}

\subsection{Rough Description of Problems Regarded ``Coded Caching''}

The second part of the paper shows how one can apply the discoordination
of three subspaces and to new obtain results regarding a
special case of 
the class of problems in information theory known collectively
as {\em coded caching}.
What makes the problems in this field so 
intriguing is that not only does it have practical applications,
but it is 
quite an elegant mathematical
puzzle that remains open in some very simple cases, despite
an impressive number of different mathematical approaches to studying
this problem (see the references mentioned above).
Furthermore, we think
coded caching will likely have
applications beyond the original caching setting
in the seminal work of Maddah-Ali and Niesen \cite{MA_niesen_2014_seminal}.
Although we cannot do justice to the wide array of results in this
field here (but see 
\cite{MA_avestimehr_2019_survey,tian_2018_computer_aided,tootooni_msc_thesis}), 
the problem we study in this article requires only
the original setting of \cite{MA_niesen_2014_seminal} and
the results in \cite{chen_journal,tian_2018_computer_aided}.
Let us describe the problem and our results in rough terms.

A central server has access to $N$ documents,
each consisting of $F$ bits
of information.
The server broadcasts information to $K$ users (i.e., send the same
message to all users).
There are two phases where the server can broadcast information;
during the second phase, for all $i=1,\ldots,K$, the $i$-th user will need
to know the contents
of exactly one document, say document number $d_i=1,\ldots,N$, but
the value of the vector of requests
$\mec d=(d_1,\ldots,d_k)$, an arbitrary element of $\{1,\ldots,N\}^k$,
is unknown during the first phase.
The first phase is a time of ``low network usage,''
where the server can broadcast all $NF$ bits of information
to all users;
each user has their own ``cache'' that can store up to $MF$ bits of
information, i.e. some function of these $NF$ bits, 
where $M$ is a rational number less than $N$.
Hence each user can store some---but not all---of the information contained
in the
$N$ documents.
Just before the second phase, the server and each user become aware
of the value of $\mec d$.
At this point---a time of ``high network usage''---the server
can broadcast at most $RF$ bits of information where $R$ is some
rational number less than $N$.
We are interested to know for which values of $M,R,F,N,K$
there is a caching scheme,
i.e., choice of values of the contents of the $K$ caches
(of at most $MF$ bits each),
such that for any $\mec d\in\{1,\ldots,N\}^K$, the server
can broadcast a message of at most $RF$ bits that allows each
user $i$ to reconstruct document $d_i$.
Specifically, we call $(M,R)$ the {\em memory-rate pair} and
say that it is {\em achievable} for a given $N,K$ if for some $F$
there exists a caching scheme and broadcasting scheme with the 
above parameters.  

The reader who has never seen this problem before is encouraged
to think about the case $N=K=2$, settled in the seminal
work of Maddah-Ali and Niesen \cite{MA_niesen_2014_seminal},
where their solution involves a clever technique of combining
information to obtain their bound $M+R\ge 3/2$ in their appendix,
page~2866 there.
The bounds we get for $N=K=3$ are based on a variant of their 
technique.

At present the optimal value of $R$ for the original 
coded caching problem is known to within a factor of roughly $2$
for all $N,K,M$; see \cite{MA_avestimehr_2019_survey}.
There are a large number of variants of the original problem
(see \cite{saberali_thesis}).

In this article we address only the case
$N=K=3$: this case was previously settled for $M\ge 1$ in the seminal
work of Maddah-Ali and Niesen \cite{MA_niesen_2014_seminal},
for $M\le 1/3$ by \cite{chen_journal}; for $1/3 <  M< 1$
the best lower bounds on $R$ to date were given
by Tian \cite{tian_2018_computer_aided} of
$$
M+R \ge 2, \ 2M+R \ge 8/3,
$$
who also gave a rather ingenious argument to show that the meeting
point of these two lower bounds, $(2/3,4/3)$, cannot be achieved
by caching schemes and broadcasting schemes that are all linear
functions of the source of $NF$ bits\footnote{
  Later in this article we will explain, as does Tian, that Tian's
  result on $(2/3,4/3)$ assumes the accuracy of the results of
  a computer-aided floating point computation.
  }.  
Most of our work on coded caching is to
use our theory of discoordination
and generalize Tian's argument to show that any linear coded caching
scheme must satisfy
$$
6M+5R \ge 11 .
$$
We will also show
that $(1/2,5/3)$ is achievable.  It therefore follows that:
\begin{enumerate}
\item 
Tian's inequality $2M+R\ge 8/3$ is tight for $1/3\le M\le 1/2$,
and the problem remains open for $1/2 < M < 1$;
\item
without the assumption of linearity, the best lower bounds are still
Tian's bounds $2M+R\ge 8/3$ for $1/2<M\le 2/3$, and
$M+R\ge 2$ for $2/3\le M \le 1$
(which meet at the point $(2/3,4/3)$;
and
\item
assuming linearity, the best bounds are (Tian's)
$2M+R\ge 8/3$ for $1/2\le M\le 7/12$,
and our bound
$6M+5R\ge 11$ for $7/12\le M\le  1$,
which meet at the point $(M,R)=(7/12,3/2)$.
\end{enumerate}

In a bit more detail, we first review Tian's method and show
how it gives a lower bound of $3M+2R\ge 5$
to the $(2/3,4/3)$ caching scheme that Tian studied, under the assumption
that $Z_1,Z_2,Z_3$ are {\em separated} in the sense that
they dependent on independent parts of each of $W_1,W_2,W_3$.
We next show that any $Z_i$ of a linear caching scheme
can be decomposed into a sum of copies of four basic schemes,
and we show that one of these schemes leads to the achievability
of $(M,R)=(1/2,5/3)$.
We next give a hybrid of Tian's method and a direct computation
to get a lower bound of $3M+2R$ based on how many copies of each
scheme is involved, assuming the scheme is separated.
We get a second such lower bound, curiously also on $3M+2R$,
using our discoordination methods, which does not assume
separation.
(Visibly, each lower bound fails to give an optimal bound for 
one of the four schemes, and hence any blend of these lower bounds
cannot be optimal.)
Combining these bounds yields the lower bound
$$
6M+5R \ge 11
$$
for separated linear schemes.
We will also explain that we conjecture $4M+3R\ge 7$, and explain
a bit about {\mygreen what} a scheme with $4M+3R < 7$ would have to involve.
Finally, in the last section we show how to give a (slightly weaker)
version of our hybrid bound that holds without the assumption of
{\mygreen separability;} interestingly, we will make a different use of our
decomposition theorem for three subspaces here.
Although this bound that does not assume separability is slightly weaker
(in the coefficient of part of one of the four schemes), we still get
the bound $6M+5R\ge 11$.

See also Tian's work \cite{tian_2018_computer_aided} for results
on a number of other results with small values of $N$ and $K$,
and \cite{MA_avestimehr_2019_survey} for a number of other recent results
in this article and prior works.

\subsection{Organization of the Rest of this Article}

The first part of this article is devoted to proving theorems in
linear algebra regarding coordination and discoordination of
a collection of subspaces of an ambient vector space.
In Section~\ref{se_Joel_basics}
we review some notation used throughout this article,
especially that involving linear algebra.
In Section~\ref{se_Joel_coord_discoord_defs}
we give the basic definitions of coordination and discoordination,
state all the main theorems we prove regarding these notions, and
make some preliminary remarks about them.
In Section~\ref{se_Joel_coordination_method}
we give our main method to prove that a collection of subspaces
is coordinated, 
{\mygreen and apply it to a number of such collections.}
In Section~\ref{se_Joel_discoord_formula}
we give a number of fundamental results regarding coordination
and discoordination that we will need later.
In Section~\ref{se_Joel_three_subspaces_main}
we use the previous subsection to prove our main theorem
about three general subspaces of a vector space, and to prove
some related results.

The second part of this article is devoted to studying a special
case of coded caching.
In Section~\ref{se_Joel_coded_intro}
we review the problem of coded caching, and generalize
a result of Tian \cite{tian_2018_computer_aided}
in the case $N=K=3$.
In Section~\ref{se_Joel_symmetrization}
we explain the technique of averaging, and---what is essentially
equivalent---reducing the study of coded caching schemes to
{\em symmetric} schemes.
In Section~\ref{se_Joel_scheme_decomp_lemma}
we prove that any cache in a coded caching scheme that is linear can
be decomposed into a small number of special schemes.
In Section~\ref{se_Joel_new_point} we take one of the special
schemes in
Section~\ref{se_Joel_scheme_decomp_lemma} and use it to show that
the point $(1/2,5/3)$ is achievable for $N=K=3$.
In Section~\ref{se_Joel_N_K_3_discoord}
we prove two main lower bounds regarding the memory-rate trade-offs
for linear coded caching schemes in the case $N=K=3$ that involve
the discoordination of a collection of three vector subspaces.
In Section~\ref{se_Joel_hybrid}
we give another such lower bound by adapting
an extremely clever method of Tian
\cite{tian_2018_computer_aided} as part of the $N=K=3$ analysis there,
which we combine with a discoordination bound
in Section~\ref{se_Joel_N_K_3_discoord} to obtain
and the inequality $6M+5R\ge 11$ under the condition 
that the three caches are ``separated,'' in a sense we will define.
In Section~\ref{se_Joel_hybrid_not_sep} we show that
$6M+5R\ge 11$ holds under the assumption of linearity alone;
this involves a very different application of
our main theorems on three subspaces to caching schemes that are 
not necessarily
separated.
We make some concluding remarks in
Section~\ref{se_Joel_conclusion}.

\subsection{Acknowledgements}


The first author wishes to thank 
Seyed Ali Sabareli and Lutz Lampe for inviting him
to Sabareli's thesis defense and introducing him
to their work and the literature on coded caching.
He also thanks Alice Izsak and Lior Silberman for work on
gapped sheaves that first brought to his attention
the importance of coordination versus
discoordination
of subspaces.
We also thank Luc Illusie for a biographical remark on
discoordination of three subspaces, and to
Sathish Gopalakrishnan for conversations.

{\mygreen
We thank Chao Tian for bibliographical remarks added in the 
revision, including bringing to our attention the works of
\cite{gomez_vilardebo,cao_xu_n_k_three}
and the suggestion to address the relationship
of our methods to Ingleton's inequality.
}

\section{Basic Notation and Conventions Regarding Linear Algebra and
Information Theory}
\label{se_Joel_basics}

In this section we will give some definitions and notions
in linear algebra that we will need,
and give a brief review of information theory.
We refer to \cite{janich,axler} for basic notions in
linear algebra, quotient vector spaces, etc.
We will briefly review these as needed.

\subsection{Set Theoretic Notation}

We use $\integers,\reals$ to respectively denote the integers and the
real numbers.
We use $\naturals$ to denote the natural numbers $\{1,2,\ldots\}$,
and for $n\in\naturals$ we use $[n]$ to denote $\{1,\ldots,n\}$.
However, we alert the reader that square brackets will at times be
used in the following notation:
(1) $[a,b]$ denoting the closed real interval between
$a$ and $b$, where $a,b\in\reals$ with $a\le b$, and (2)
$[u]$ and $[Y]$ as shorthand for
$[u]_W=u+W$ and $[Y]_W=Y+W$, the image in the quotient vector space
$\cU/W$ of a vector $u\in \cU$ 
or a subset $Y\subset \cU$
(see Subsection~\ref{su_quotient_space}).

If $A,B$ are sets we use $A\setminus B$ to denote the set difference of 
$A$ and $B$
$$
A\setminus B = \{ a\in A \ | \ a\notin B\}.
$$
The notation $A\subset B$ is used when $A$ is a subset of a set, $B$,
or a subspace of a vector space, $B$, 
which, {\mygreen in both cases}, allows for $A=B$.

\subsection{Sum, Direct Sum, and $\oplus$}
\label{su_sum_direct_sum_oplus}

In mathematics, $\oplus$ usually denotes the direct sum, for example
of vector spaces.  However, in the coded caching literature,
$\oplus$ is usually used for addition in $\field^n$ for
the field $\field=\integers/2\integers$.
In order to avoid confusion, we will use
$\underline{\oplus}$ for the direct sum of vector spaces: e.g.,
if $U_1$ and $U_2$ are vector spaces over a field $\field$,
then $U_1\underline{\oplus} U_2$ is their direct sum
(also equal to their product), whose underlying set is
the Cartesian product $U_1\times U_2$ and where addition is given
by component-wise addition.

The following convention will be very useful to discuss coded
caching.  Let $\field$ be an arbitrary field,
and let $U_1,U_2$ be finite-dimensional
subspaces of an $\field$-vector space $\cU$
with $\dim(U_1)=\dim(U_2)$,
and let $\nu\from U_1\to U_2$ be an isomorphism.
Then we use $U_1\oplus_{\nu} U_2$ to denote the subspace of $\cU$
consisting of all vectors $u_1+\nu(u_1)$ with $u_1\in U_1$.
Often $\nu$ will be understood (or unimportant), 
in which case we just write $U_1\oplus U_2$.
Hence $U_1\oplus U_2$ always connotes that there is an
understood isomorphism $U_1\to U_2$.
Our isomorphisms, such as those in
Lemma~\ref{le_Z_scheme_decomposition}, will be built 
by
choosing an ordered basis $b_1^1,\ldots,b_1^k$ for $U_1$ and
another
$b_2^1,\ldots,b_2^k$ for $U_2$, and taking
$\nu\from U_1\to U_2$ to be the unique linear map $\nu(b_1^i)=b_2^i$.

Similarly, if $U_1,U_2,U_3$ are isomorphic subspaces of some $\field$-vector
space, $\cU$, and $\nu_i\from U_1\to U_i$ are isomorphisms for $i=2,3$,
then $U_1\oplus_{\nu_2}U_2\oplus_{\nu_3}U_3$ denotes the subspace
of $\cU$ given by all vectors $u_1+\nu_2(u_1)+\nu_3(u_1)$ with
$u_1\in U_1$.  We will use this notation in
Lemma~\ref{le_Z_scheme_decomposition} and the discussion that
follows.

\subsection{Inequality Summation Principle}
\label{su_inequality_summation_principle}

The following proposition is immediate; however, it will
be helpful to signal the reader when we use it, so
we name it
the ``Inequality Summation Principle.''

\begin{proposition}[Inequality Summation Principle]
\label{pr_inequality_summation_principle}
For $m\in\naturals$, say that for real numbers 
$s_1,\ldots,s_m$ and $t_1,\ldots,t_m$ we have
\begin{equation}\label{eq_individual_inequalities}
s_1\le t_1,\ \cdots,\ s_m\le t_m .
\end{equation} 
Then
\begin{equation}\label{eq_summed_inequalities}
s_1+\cdots+s_m \le t_1+\cdots+t_m,
\end{equation} 
and equality holds in \eqref{eq_summed_inequalities} iff equality
holds in all the inequalities in 
\eqref{eq_individual_inequalities}.
\end{proposition}

[Although this principle may seem trivial,
it is the well-known idea behind
{\em complementary
slackness} in linear programming.]

\subsection{Basic Notation: $\field$-Universes, Sum, and Span}

\begin{definition}
Let $\field$ be an arbitrary field.
By an {\em $\field$-universe}, $\Univ$, we mean
we mean a finite-dimensional
$\field$-vector space.
By the term ``$\field$-universe,'' without mention of $\field$, we 
understand that
$\field$ is an arbitrary field.
\end{definition}
In this article the field $\field$ and $\field$-universe, $\cU$, 
are generally 
fixed or, at least, understood in context.  Generally, $\field$ is
an arbitrary field for our theorems in linear algebra, and
$\field=\integers/2\integers$ in applications to information
theory.
At times we will work in more than one ambient universe $\cU$,
in which case for $S\subset \cU$ a subspace 
we write $\dim^\cU(S)$ to emphasize $\cU$.

At times we use results that hold when the ambient vector space 
$\cU$ can be infinite-dimensional; in this case we use the
term ``$\field$-vector space,'' and, similarly, we understand $\field$
to be an arbitrary field unless explicitly mentioned otherwise. 
However, in this article we mostly limit ourselves to
ambient vector
spaces, $\cU$, that are finite-dimensional.

The following notation is standard:
if $A,B$ are subsets of an $\field$-vector space, $\Univ$, the
{\em sum of $A$ and $B$} refers to the set
\begin{equation}\label{eq_definition_of_sum}
A+B = \{ a+b \ | \ a\in A,b\in B \}  ;
\end{equation} 
if $A,B\subset\Univ$ are subspaces, then we easily see that $A+B$ is
also a subspace;
we similarly define $A_1+\cdots+A_m$ for any subsets $A_1,\ldots,A_m$ of
$\Univ$.
If $S_1,\ldots,S_m$
are subsets of $\Univ$, we use
$$
{\rm Span}(S_1,\ldots,S_m)
$$
to denote the span of $S_1,\ldots,S_m$; if $S_1,\ldots,S_m$ are subspaces,
then this span equals $S_1+\cdots+S_m$.

\subsection{The Dimension Formula and Its Proof}
\label{su_dimension_formula_proof}

Let us recall the {\em dimension formula} and its proof; for details,
see, for example, \cite{janich}, Theorem 3 in Section~3.2 (page 49)
(called there
the ``Dimension formula for subspaces'').
We will make frequent {\mygreen use} of this formula.
Furthermore, the proof of this formula
illustrates our main technique 
(of {\em quasi-increasing subspaces},
see Section~\ref{se_Joel_coordination_method})
to show that certain
subspaces of a universe are {\em coordinated} (we formally define
this notion in the next section).

The ``dimension formula'' says that if $U_1,U_2\subset \cU$
are subspaces of an $\field$-universe, $\cU$, then
\begin{equation}\label{eq_dimension_theorem}
\dim(U_1\cap U_2) + \dim(U_1+U_2)=\dim(U_1)+\dim(U_2).
\end{equation} 
The usual proof (see, e.g., \cite{janich}) begins as follows:
\begin{enumerate}
\item 
let $B_0$ be a basis for $U_1\cap U_2$;
\item
extend $B_0$ (in an arbitrary fashion) to a basis $B_0\cup B_1$
of $U_1$;
\item
extend $B_0$ (in an arbitrary fashion) to a basis $B_0\cup B_2$
of $U_2$;
\item
we then verify that $B_0,B_1,B_2$ are disjoint and 
$B=B_0\cup B_1\cup B_2$ is a linearly independent set;
we will see that
this verification is equivalent to verifying that the sequence
$U_1\cap U_2,U_1,U_2$ is {\em quasi-increasing} in the language
of our Section~\ref{se_Joel_coordination_method}.
\end{enumerate}

%
%
%

To finish the proof of the dimension theorem, we note that
$$
\dim(U_1+U_2)=|B_0|+|B_1|+|B_2|, 
\ \dim(U_1\cap U_2)=|B_0|,
\ \dim(U_i)=|B_0|+|B_i|
$$
for $i=1,2$.


\subsection{Conventions Regarding Quotient Spaces and Relative
Bases}
\label{su_quotient_space}

In this subsection, we recall the usual notion of a {\em quotient space} of 
vector spaces
(see \cite{axler,janich} for details); then we define the notion 
of a {\em relative basis},
which is a commonly used idea (but for which we know 
of no standard terminology);
relative bases will feature prominently throughout this article.

Let $\field$ be a field, $U$ an $\field$-vector space, and $W\subset U$
a subspace.
By a {\em $W$-coset of $U$} we mean any set of the form
$u+W$ where $u\in U$ and $+$ is as in
\eqref{eq_definition_of_sum}; it is convenient to denote $u+W$ by
$[u]_W$, and we use $U/W$ to denote the set of all $W$-cosets.  
(We easily see that $u_1+W=u_2+W$ iff $u_1-u_2\in W$, so that
one can also view $U/W$ as the set of equivalence classes under the
equivalence $u_1\sim u_2$ iff $u_1-u_2\in W$.)
We easily check that the vector space structure on $W$
gives rise to one on $U/W$, and that
$$
\dim(U/W) = \dim(U)-\dim(W) 
$$
if $U$ is finite-dimensional.
If $Y\subset U$ is any subset of $U$, we use the notation $[Y]_W$ to 
denote the set of $W$-cosets $Y+W$, viewed as a subset of $U/W$;
we call $[Y]_W$ the {\em image of $Y$ in $U/W$};
hence if $u\in \cU$, then $[\{u\}]_W$ is just the one element set
$[u]_W\in U/W$.
At times we write $[u]$ and $[Y]$ for 
$[u]_W$ and $[Y]_W$ when $W$ is understood.

{\mygreen
If $W\subset U$, then a {\em complement of $W$ in $U$} refers to
any subspace $W'\subset U$ such that each vector in $U$ can be
written uniquely as a sum of an element in $W$ plus one in $W'$.
We easily see that this is equivalent to saying that the map
$U\to U/W$ restricted to $W'$ gives an isomorphism from $W'\to U/W$.
}

The following terminology is not standard but will be very useful in
this article.
\begin{definition}\label{de_relative_basis}
Let $W\subset U$ be vectors spaces over some field, of finite dimensions
$m,n$ respectively.  
We say that a subset, $Y=\{y_1,\ldots,y_{n-m}\}$, of $U$ is
a {\em basis of $U$ relative to $W$} if
the image, $[Y]_W=\{y_1+W,\ldots, y_{n-m}+W\}$ is a basis of 
{\mygreen $U/W$}.
\end{definition}
In the above, we easily see that if $X=\{x_1,\ldots,x_m\}$ is any
basis for 
{\mygreen $W$}, 
then $Y=\{y_1,\ldots,y_{n-m}\}$ is
a basis of 
{\mygreen $U$ relative to $W$}
iff 
$X\cup Y$ is a basis for $U$
(which also implies that $X$ and $Y$ are disjoint).
So while we may think of $Y$ as what we add to $X$ to complete
the basis, the above definition shows that our choice of $Y$
depends only on $W={\rm Span}(X)$.
{\mygreen
Moreover, we easily see that for each $Y$ 
in Definition~\ref{de_relative_basis},
i.e., for each $Y$ that is a basis of $U$ relative to $W$,
we have that $W'={\rm Span}(Y)$ is a complement of $W$ in $U$,
and, conversely, if $W'$ is any complement of $W$ in $U$,
and $Y$ is any basis of $W'$, then
$Y$ is a basis for $U$ relative to $W$.
}

At times, if $U,W\subset \cU$ are two subspaces of a vector space $\cU$,
we alternatively
use $U/W$ to denote the subspace $[U]_W=U+W$ of $\cU/W$.

If $\cU$ is an $\field$-universe, 
and $W\subset\cU$ a subspace,
part of our methods examines what happens to certain vector
subspaces of $\cU$ when we consider their image in $\cU/W$.
{\mygreen We will often use the following remark:}
if $A\subset\cU$ is another subspace, then
$[A]_W$ is a subspace of $\cU/W$, but (we easily check that)
$[A]_W$, as a subspace of $\cU/W$,
is isomorphic to the $\field$-vector space $A/(A\cap W)$.
Hence
$$
\dim^{\cU/W}([A]_W) = \dim^\cU(A)-\dim^\cU(A\cap W).
$$

\subsection{Independent Subspaces and Decompositions}

The notion of the linear independence of subspaces of a vector subspace
is not a standard notion although likely occurs implicitly in many
places in the literature,
such as the decomposition of a vector space, $V$, into the generalized
eigenspaces of an operator $V\to V$.

Consider
subspaces 
$A_1,\ldots,A_m$ of an $\field$-universe, 
and for each $i\in[m]=\{1,\ldots,m\}$,
let $X_i$ be a basis for $A_i$.  Then
each vector in $A_1+\cdots+A_m$ can be written as $a_1+\cdots+a_m$,
and hence lies in the span of
$X_1\cup\cdots\cup X_m$.  Therefore
$$
\dim(A_1+\cdots+A_m) \le | X_1\cup\cdots\cup X_m | \le
|X_1|+\cdots+|X_m|,
$$
and hence
\begin{equation}\label{eq_sum_vector_spaces_dim_inequality}
\dim(A_1+\cdots+A_m) \le \dim(A_1)+\cdots+\dim(A_m);
\end{equation} 
furthermore strict inequality holds in one of two cases:
(1) the $X_1,\ldots,X_m$ are not
distinct, or (2) some proper subset of $X_1\cup\cdots\cup X_m$ also spans
$A_1+\cdots+A_m$.
Both cases imply that for some
$i\in[m]$, 
some $x\in X_i$
can be expressed as a linear combination in the vectors in
$X_i\setminus\{x\}$ and the remaining $X_j$ (i.e., such that $j\ne i$).
Since the vectors in each of the bases are linearly independent, in this
expression leads to an equation
$$
a_1+\cdots+a_m=0\quad\mbox{where $a_j\in A_j$ for all $j$ and $a_i\ne 0$}.
$$
Conversely, if equality holds in 
\eqref{eq_sum_vector_spaces_dim_inequality}, then $X_1,\ldots,X_m$
are necessarily distinct and their union is a linearly independent
set that spans $A_1+\cdots+A_m$;
hence this union comprises a basis for $A_1+\cdots+A_m$.
There are a number of equivalent ways of stating the condition
of equality holding in \eqref{eq_sum_vector_spaces_dim_inequality},
which are minor variants of these two conditions
and which we state below (left as an easy exercise for the reader).

\begin{definition}\label{de_linearly_independent_subspaces}
Let $A_1,\ldots,A_m$ be subspaces of an $\field$-universe $\cU$.
We say that $A_1,\ldots,A_m$ are {\em linearly independent} if
any of the following conditions hold:
\begin{enumerate}
\item 
for all $a_1,\ldots,a_m$ with $a_i\in A_i$ for all $i\in[m]$,
the equation
$$
a_1+\cdots+a_m=0 
$$
implies that $a_1=\cdots=a_m=0$;
\item
any $a\in A_1+\cdots+A_m$ has a unique representation as a sum
$a=a_1+\cdots+a_m$ with $a_i\in A_i$ for all $i\in[m]$;
\item 
for any bases $X_1,\ldots,X_m$ of $A_1,\ldots,A_m$ respectively,
the $X_1,\ldots,X_m$ are pairwise disjoint and
$X_1\cup\cdots\cup X_m$ is a basis for $A_1+\cdots+A_m\subset \cU$;
\item 
there exist bases $X_1,\ldots,X_m$ of $A_1,\ldots,A_m$ respectively,
such that 
the $X_1,\ldots,X_m$ are pairwise disjoint and
$X_1\cup\cdots\cup X_m$ is a basis for $A_1+\cdots+A_m\subset \cU$;
and
\item
\begin{equation}\label{eq_linear_independence_dim_sum_criterion}
\dim(A_1)+\cdots+\dim(A_m) = \dim(A_1+\cdots+A_m).
\end{equation} 
\end{enumerate}
\end{definition}
We note that condition~(5) makes use of the fact that $\cU$ is finite
dimensional, whereas (1)--(4) above are equivalent when $\cU$ is any
$\field$-vector space such that any subspace
of $\cU$ has a basis\footnote{
  The existence of a basis for any $\field$-vector space is typically 
  assumed in
  linear algebra, although depending on the vector spaces, this 
  assumption may
  require a set theoretic assumption such as the validity
  of transfinite induction.
  }.

\begin{example}
If $u_1,\ldots,u_m$ are vectors in some vector space,
we easily see that vectors are linearly independent iff
all these vectors are nonzero and
${\rm Span}(u_1),\ldots,{\rm Span}(u_m)$ are linearly
independent subspaces.  Hence the classical notion of linearly
independent vectors can be described in terms of the linear
independence of one-dimensional subspaces.
\end{example}

\begin{example}
If $B_1,\ldots,B_m$ is any partition of a set of 
linearly independent vectors in any vector space,
then their spans are linearly independent subspaces.
\end{example}

Just as in Definition~\ref{de_linearly_independent_subspaces}, we easily
check that the three conditions in the following definition
are equivalent.

\begin{definition}
By a {\em decomposition} of a subspace $U\subset\cU$ of an
$\field$-universe, $\cU$, we mean subspaces $U_1,\ldots,U_m\subset U$
such that any of the equivalent conditions hold:
\begin{enumerate}
\item 
each $u\in U$ can be written uniquely as $u_1+\cdots+u_m$ where
$u_i\in U_i$ for all $i\in[m]$;
\item 
the subspaces $U_1,\ldots,U_m$ are independent and their span is all of $U$;
\item 
the map $U_1\underline\oplus\cdots \underline\oplus U_m\to U$ 
taking $(u_1,\ldots,u_m)$
to $u_1+\cdots+u_m$ is an isomorphism.
\end{enumerate}
\end{definition}

Next we want to 
{\mygreen define what it means for 
a subspace $A\subset U$ to factor through}
a decomposition $U_1,\ldots,U_m$ of $U$.  
{\mygreen To do so, note that the}
first condition in Definition~\ref{de_linearly_independent_subspaces}
implies that 
if $U_1,\ldots,U_m$ are any linearly independent subspaces, then
so are
$A\cap U_1,\ldots,A\cap U_m$.
Hence
\begin{equation}\label{eq_tilde_dimension_inequality_under_decomposition}
\sum_{i=1}^m \dim(A\cap U_i) \le \dim(A) .
\end{equation}
We easily verify the conditions in the definition below are
equivalent (and, again, leave these to the reader).

\begin{definition}\label{de_subspaces_factors_through_a_decomposition}
Let $U_1,\ldots,U_m$ be a decomposition of a subspace $U$ of some universe.
We say that subspace 
$A\subset U$ {\em factors through} this decomposition
if any of the equivalent conditions hold:
\begin{enumerate}
\item
$A\cap U_1,\ldots,A\cap U_m$ is a decomposition of $A$;
\item
any vector in $A$ can be written as a sum of vectors in
$A\cap U_1,\ldots,A\cap U_m$ (which is necessarily unique);
\item
the span of 
$A\cap U_1,\ldots,A\cap U_m$ is all of $A$;
\item
we have
\begin{equation}\label{eq_A_factors_equality}
\sum_{i=1}^m \dim(A\cap U_i) = \dim(A)
\end{equation} 
i.e., equality holds in
\eqref{eq_tilde_dimension_inequality_under_decomposition}.
\end{enumerate}
If so, we refer to each of $A\cap\cU_1,\ldots,A\cap\cU_m$ as 
a {\em factor} of $A$ (in this decomposition).
\end{definition}

The following proposition likely occurs in a number of places in
the literature.
\begin{proposition}\label{pr_intersect_and_span_factor_through_decomps}
If $A,B\subset U$ factor through a decomposition $U_1,\ldots,U_m$ of a
subspace, $U$, of some universe,
then $A+B$, $A\cap B$ also factor through this decomposition.
\end{proposition}
\begin{proof}
For any $i\in[m]$, the dimension
formula 
{\mygreen applied to $A\cap U_i$ and $B\cap U_i$} 
implies that
$$
\dim(A\cap U_i) + \dim(B\cap U_i) =
\dim(A\cap B\cap U_i) + \dim\bigl( (A\cap U_i)+(B\cap U_i) \bigr),
$$
{\mygreen which is}
$$
\le 
\dim(A\cap B\cap U_i) + \dim\bigl( (A+B)\cap U_i) \bigr)
$$
since $A\cap U_i,B\cap U_i$ are both subspaces of $(A+B)\cap U_i$.
Summing
{\mygreen
$$
\dim(A\cap U_i) + \dim(B\cap U_i) 
\le 
\dim(A\cap B\cap U_i) + \dim\bigl( (A+B)\cap U_i) \bigr)
$$}
over all $i$, the left-hand-side is just $\dim(A)+\dim(B)$,
and so the dimension formula implies that
\begin{equation}\label{eq_backwards_A_cap_B_sum_B_ineq}
\dim(A\cap B) + \dim(A+B) \le
\sum_{i=1}^m 
\dim(A\cap B\cap U_i) + 
\sum_{i=1}^m \dim\bigl( (A+B)\cap U_i) \bigr).
\end{equation} 
But \eqref{eq_tilde_dimension_inequality_under_decomposition} implies that
\begin{align}
\label{eq_decomposition_inequality_for_the_intersection}
\sum_{i=1}^m \dim\bigl( A\cap B)\cap \cU_i \bigr) 
\le & \dim(A\cap B) , \\
\label{eq_decomposition_inequality_for_the_sum}
\sum_{i=1}^m \dim\bigl( (A+ B)\cap \cU_i \bigr) 
\le & \dim(A+ B) ;
\end{align}
summing these inequalities and comparing with
\eqref{eq_tilde_dimension_inequality_under_decomposition} shows that 
\eqref{eq_decomposition_inequality_for_the_intersection} and
\eqref{eq_decomposition_inequality_for_the_sum}
must hold with equality.
Hence $A+B$ and $A\cap B$ factor through the decomposition
$U_1,\ldots,U_m$.
\end{proof}

We will have occasion to use the following almost immediate consequence.

\begin{theorem}\label{th_factorization_under_cap_sum}
Let $U_1,\ldots,U_m$ be a decomposition of a subspace $U$ of some universe.
Say that each of the subspaces
$A_1,\ldots,A_s\subset U$ factors through this decomposition.
Then any subspace that can be written as an expression
involving $+$ and $\cap$ and the $A_1,\ldots,A_s$
(and parenthesis) factors through this decomposition as well.
Similarly, ``$\dim$'' applied to any such expression can be computed
by the sum over $i\in[m]$ of ``$\dim$'' applied the same expression
on each factor, i.e., where
each $A_j$ replaced by $A_j\cap \cU_i$; furthermore, the same holds 
for $\dim^{\cU/B}([A]_B)$ where $A,B$ are each such expressions.
\end{theorem}
\begin{proof}
The proof that each expression in $+,\cap,A_1,\ldots,A_m$ (and
parenthesis) factors through the decomposition
follows immediately
from Proposition~\ref{pr_intersect_and_span_factor_through_decomps},
using 
induction on the {\em size} of
the expression, meaning the number of $\cap$'s and $+$'s in it.
The fact that $\dim$ applied to such an expression is the sum of the
same expression applied to each factor follows from
\eqref{eq_A_factors_equality}.
Finally,
if $A,B$ factor through such a decomposition, then 
{\mygreen we have}
$$
\dim^{\cU/B}([A]_B) = \dim^\cU\bigl( A/(A\cap B) \bigr) 
=\dim^\cU(A)-\dim^\cU(A\cap B)
$$
and $\dim^\cU(A)$ and $\dim^\cU(A\cap B)$ can be computed as
the sum of these dimensions over each factor of $A$ and $B$;
{\mygreen hence $\dim^{\cU/B}([A]_B)=\dim^\cU\bigl( A/(A\cap B) \bigr)$
factors through the decomposition.}
(Alternatively one can write the expression displayed 
above as $\dim^\cU(A+B)-\dim^\cU(B)$ and reach the same conclusion.)
\end{proof}

\subsection{Basis Exchange and Independent Subspaces}

In this article, we will use a number of variants of the basis exchange
and basis extension principles.
Let us state a few that we will need; they are easy exercises left to
the reader in view of the usual Basis Exchange Lemma
(e.g., Section~3.4 of \cite{janich}).

\begin{proposition}\label{pr_basis_exchange}
Let $U$ be a subspace of any $\field$-universe.  Then
\begin{enumerate}
\item if $U={\rm Span}(S)$ for some subset, $S$, of $\cU$, then
some subset $S'\subset S$ is a basis for $U$;
\item
{\mygreen
if $W\subset U$ is a subspace, and $S\subset U$ is a set such that
$U={\rm Span}(W,S)$, then 
some subset $S'\subset S$ is a basis for $U$ relative to $W$;
}
\item 
if $X$ is a basis of $U$,
$X_0\subset X$ a subset, and 
$Y {\mygreen \subset U}$ 
{\mygreen such that $Y$ is a linearly independent set, and}
${\rm Span}(X_0)$ and ${\rm Span}(Y)$ 
{\mygreen are}
linearly independent,
then there exists a basis for $U$ consisting of $X_0\cup Y$ plus
a subset of vectors from $X\setminus X_0$; and
\item 
if $X$ is a basis for $U$ and $Y$ a subset of linearly independent
vectors in $U$, then there is a basis of $U$ of the form $Y\cup X_0$
with $X_0\subset X$
{\mygreen and $Y,X_0$ disjoint}
(this is the standard basis exchange principle).
\end{enumerate}
\end{proposition}

\subsection{A Review of Information
Theory and the Definition of a Linear Random Variable}
\label{su_review_info_theory}

In this subsection, we will review the notions in information theory
that are most essential to this paper, such as the entropy of a 
random variable, and alert the reader to some particular assumptions
and notation that we use.
A more complete discussion of information theory can be found in
a number of basic textbooks, such as \cite{cover_thomas}.
We then motivate and discuss linear random variables, and our view of them as
subspaces of the dual space of the source.

Throughout this subsection, $\field=\integers/2\integers$ is the finite
field of two elements.
In this subsection we review the usual notion of entropy and
explain what we mean by a {\em linear random variable}
of an $\field$-vector space, $S$.

Let us first summarize this subsection,
for the sake of the experts (who can 
likely read Definition~\ref{de_linear_random_vars} and
skip most of the rest of this subsection).
Classically, a random
variable on a source (meaning, in this article,
a finite probability space), $S$,
is a map $Y\from S\to \cY$.  
{\mygreen We assume that $S$ is an $\field$-vector space
and---as a probability space---is endowed with the uniform distribution.}
We say that $Y$ is
{\em classical linear random variable} if $\cY$ can be given the structure
of an $\field$-vector space so that $Y$ is a linear transformation.
In this case 
{\mygreen we will easily prove that}
$Y$ is equivalent to a surjective map $S\to\field^m$, where
$m=\dim({\rm Image}(Y))$, 
which allows us to identify $Y$ with an $m$-dimensional
subspace, $V$, of the dual space, $S^*$, of $S$.
{\mygreen Hence for fixed $S$ we get a map
$$
\{ \mbox{classical linear random variables $Y\from S\to \cY$} \} \to 
\{ \mbox{linear random variables on $S$} \},
$$
where
$$
\{ \mbox{linear random variables on $S$} \}
\eqdef
\{ \mbox{subspaces, $V$, of $\cU\eqdef S^*$} \},
$$
where equivalent classical random variables are mapped to the same
subspace of $\cU$; for this and numerous other reasons,
it is far simpler to work with subspaces of $\cU$.
Here are some further relations between a classical linear random 
variable, $Y$, and its associated $V\subset\cU=S^*$:
for one,} 
$H_2(Y)$, the usual {\em entropy} of $Y$, is just 
$m=\dim(V)$.
Finally, we will show that if $Y_1,Y_2$ are two classical linear
random variables, and $V_1,V_2\subset S^*$ the associated subsets
of $S^*$ then (1) $Y_1$ and $Y_2$ are equivalent iff $V_1=V_2$,
and (2) the subspace of $S^*$ associated to the random variable $(Y_1,Y_2)$
is just $V_1+V_2$ (i.e., the span of $V_1$ and $V_2$);
it follows that $I(Y_1;Y_2)$, the classical {\em mutual information
of $Y_1$ and $Y_2$}, equals $\dim(V_1\cap V_2)$.
We remark that these ideas are implicit in a lot of the information
theory literature, in particular in the way
Maddah-Ali and Niesen
\cite{MA_niesen_2014_seminal} and other papers on coded caching
describe their coded caching schemes,
all of which are linear.
Hence this subsection simply gives a review of some parts of information
theory in common use.

{\mygreen For the rest of this subsection we spell out the details for
the statements in the previous paragraph.
We remark that \cite{hammerEtAl2000}
also ties together linear algebra, Shannon entropy, 
and---in addition---Kolmogorov complexity.}

Before reviewing classical information theory, let us give an example.

\begin{example}\label{ex_simple_example_linear_random_variables}
Let $\field=\integers/2\integers$, 
and consider $S=\field^4$, with $\oplus$ denoting addition
(i.e., of components, modulo $2$).
If $x=(x_1,x_2,x_3,x_4)\in S=\field^4$,
then $Y_1=x_1\oplus x_2$ is an example of what we will call a
``classical linear random variable.''  Technically $x_i$ are really
maps $S\to\field$, i.e., elements of the dual space $S^*$, so 
$Y_1\in S^*$.
If $Y_2=x_3$, then the random variable $Y_3=(Y_1,Y_2)$ is the map
taking $S=\field^4$ to $\field^2$ that takes $x$ to $(x_1\oplus x_2,x_3)$.
Similarly the random variable $Y_4=(x_2,x_3)$ is a two-dimensional
random variable, and $(Y_3,Y_4)$, which is literally the random
variable $S\to\field^4$ taking $s$ to 
$$
\bigl( x_1\oplus x_2, x_3,x_2,x_3 \bigr) =
\bigl( x_1(s)\oplus x_2(s), x_3(s),x_2(s),x_3(s) \bigr)
$$
is equivalent to the random variable
$Y_5=(x_1,x_2,x_3)$, and the entropy of $Y_5$ is $H_2(Y_5)=3$
(not $4$, since the map $S\to\field^4$ above has a $3$-dimensional image).
Since the $x_i$ are really elements of the dual space, $S^*$, of $S$,
one can view the random variables here as a vector-valued random variable
whose entries are elements of $S^*$; the components of this vector 
span a subspace, $V$, of $S^*$,
and the entropy of these vector-valued random variables
is simply $\dim(V)$.
\end{example}

Let us now review some notions of classical information theory
(see \cite{cover_thomas} for more details) and some assumptions we
make, after which we define
linear random variables.

In classical information theory, a {\em source}, $S$, is
a finite set with a probability measure $P\from S\to\reals$
whose values are positive and sum to one.
(Hence we do not allow $P(s)=0$ for an $s\in S$.)
A random variable is defined as a map $Y\from S\to \cY$ where
$\cY$ is a set.  For each $y\in \cY$, we define
\begin{equation}\label{eq_p_of_y_equals_sum}
p_y = \sum_{Y(s)=y} P(s),
\end{equation} 
and we define its (base $2$) entropy to be
\begin{equation}\label{eq_that_defines_entropy}
H(Y)=H_2(Y) = \sum_{y\in\cY} p_y \log_2(1/p_y),
\end{equation} 
where $p_y \log_2(1/p_y)$ is taken to be $0$ if $p_y=0$.
We note that since (in this article)
$P(s)>0$ for all $s\in S$, for $y\in\cY$ we have $p_y$ in
\eqref{eq_p_of_y_equals_sum} is positive iff $y$ lies in
${\rm Image}(Y)$, the image of $Y$.

If $Y$ is {\em uniformly distributed} 
in the sense that $p_y$ is independent of $y$,
it easily follows that $p_y = 1/|\cY|$, and so
\begin{equation}\label{eq_entropy_of_uniform_random_variable}
H(Y) = \log_2( |\cY| ).
\end{equation} 

Each random variable $Y\from S\to \cY$ induces a partition of $S$, namely
$$
S = \bigcup_{y\in \cY} Y^{-1}(y).
$$
We say that another random variable $Y'\from S\to \cY'$ is {\em equivalent to
$Y$}
(respectively, {\em a refinement of $Y$})
if the partition that $Y'$ induces on $S$ is the same as
(respectively, a refinement of) that induced by $Y$;
we easily see that this holds iff 
there is an isomorphism (respectively, morphism)
$\mu\from{\rm Image}(Y')\to{\rm Image}(Y)$
such that $Y=\mu \circ Y'$.
We use the shorthand $Y'\implies Y$ (or say $Y'$ implies $Y$)
when $Y'$ is a refinement of $Y$.
We easily see if $Y'\implies Y$ and $Y\implies Y'$
then $Y$ and $Y'$ are equivalent.

(If $Y_1,\ldots,Y_m$ are equivalent, respectively,
to $Y_1',\ldots,Y_m'$, then any
expression involving the joint entropy, mutual information, conditional
entropy, etc.,
involving the $Y_1,\ldots,Y_m$ equals the same expression
when each $Y_i$ is replaced
with $Y_i'$.)

If $Y_1,\ldots,Y_m$ are random variables $Y_i\from S\to \cY_i$,
then the {\em join of $Y_1,\ldots,Y_m$},
denoted $(Y_1,\ldots,Y_m)$, refers to the random variable
$S\to (\cY_1,\ldots,\cY_m)$.
For random variables $Y_1,Y_2$ we define
define their mutual information
$$
I(Y_1;Y_2) = H(Y_1)+H(Y_2) 
{\mygreen - H(Y_1,Y_2),}
$$
and it is known that $Y_1\implies Y_3$ and $Y_2\implies Y_3$ implies
that $I(Y_1;Y_2)\ge H(Y_3)$.
For the sake of discussing some results on coded caching, we will
assume the notion of
{\em conditional entropy}
(see Section~2.2 of \cite{cover_thomas}) 
$H(Y|X)$ of random variables $Y$ and $X$,
which is usually defined as the expected value over $x\in \cX$ of
of $H(Y|_x)$, where $Y|_x$ is the restriction of $Y$ to $X^{-1}(x)$;
one can show that $H(Y|X)=H(X,Y)-H(X)$.
It turns out that $X\implies Y$ is equivalent to
$H(Y|X)=0$ (or, equivalently, $H(X,Y)=H(X)$).

If $Y\from S\to \cY$ is any random variable, then $Y$ is equivalent
to the random variable where we discard any $y\in\cY$ with $p_y=0$;
since we assume that each $s\in S$ has positive probability,
this amounts to discarding all elements of $\cY$ that are not
in the image of $Y$.  This amounts to replacing $Y$ with the map
it induces $S\to {\rm Image}(Y)$, which is therefore a surjective map;
we call this new random the {\em surjective version of $Y$}.
If $Y_1,Y_2$ are surjective random variables, $Y_i\from S\to \cY_i$,
then $Y_1$ is isomorphic to $Y_2$ iff
there exists a bijection 
$\mu\from\cY_1\to \cY_2$ with
$Y_2=\mu \circ Y_1$.

\begin{definition}
Let $\field=\integers/2\integers$, and let $S$ be an $\field$-vector
space.
We view $S$ as a probability space with the uniform distribution,
i.e., each element occurs with probability $1/|S|=1/2^n$ where
$n=\dim(S)$.
By a {\em classical linear random variable} we mean
a linear map $Y\from S\to \cY$ where $\cY$ is an $\field$-vector space.
\end{definition}

\begin{proposition}\label{pr_classical_linear_equiv_quotient}
To any classical linear random variable $Y\from S\to \cY$, there
is an isomorphic random variable which is a quotient map
$\tilde Y\from S\to S/A$ where $A=\ker(Y)=\ker(\tilde Y)$.  Furthermore,
\begin{equation}\label{eq_entropy_of_linear_random_variable}
H(Y)=H(\tilde Y) = \log_2(|S/A|) = \dim(S/A) = \dim(S)-\dim(A) .
\end{equation} 
\end{proposition}
\begin{proof}
It is a standard fact (and easy to check)
that any linear map $Y\from S\to\cY$ factors as
$$
S \xrightarrow{f} S/\ker(Y) \xrightarrow{g} \cY,
$$
with $f$ surjective and $g$ injective (and $g$ is uniquely determined).
Note that $Y$ is equivalent to its surjective form; hence it suffices
to prove this proposition when $Y$ is surjective.
So assume that $Y$ is surjective;
then, since $f$ is surjective,
$g$ is also surjective; in this case
$g\from S/A \to \cY$ is (injective and surjective and hence)
a bijection, and
hence $g$ gives an equivalence of the surjective form of $Y$ 
and the map $\tilde Y\from S\to S/A$ where
$\ker(Y)=A$.
Since $\tilde Y$ is surjective and linear, it is uniform, and hence
using
\eqref{eq_entropy_of_uniform_random_variable} we have
$$
H(Y) = H(\tilde Y) = \log_2(|S/A|)
$$
and \eqref{eq_entropy_of_linear_random_variable} follows.
\end{proof}

We remark that if $Y\from S\to \cY$ is a classical linear random
variable, then the image of $Y$ is a subspace of $\cY$, and
hence this image is isomorphic to $\field^m$ for some $m$.
Hence $Y$ is equivalent to a surjective map $S\to\field^m$,
and $H_2(Y)=m$.

\begin{proposition}
If $A_1,A_2\subset S$ are two subsets of an $\field$-vector space,
then the random variables $Y_i\from S\to S/A_i$ are equivalent
iff $A_1=A_2$.
In particular, each classical linear random variable
$Y\from S\to \cY$ is equivalent to a unique quotient map $S\to S/A$.
\end{proposition}
\begin{proof}
$Y_i$ partitions $S$ into its $A_i$-cosets, one of which is $A_i$.
$Y_1$ and $Y_2$ are equivalent iff they induce the same partition;
since $A_1,A_2$ both contain the zero in $S$, if $Y_1$ and $Y_2$ are
equivalent then $A_1=A_2$.
Conversely, if $A_1=A_2$ then, of course,
$S\to S/A_i$ are the same map and hence
equivalent.
\end{proof}

Recall that if $\cL\from V\to W$ is a linear map,
then the map on dual spaces,
$\cL^*\from W^*\to V^*$, has image equal to $(V/\ker(\cL))^*$ viewed
as ``it sits'' in $V^*$,
i.e., viewed as the subspace of those elements of $V^*$ that 
take $\ker(\cL)$ to zero.
In particular, if $S\to S/A$ is a quotient map,
then the image of the dual map is $(S/A)^*$ as it sits in $S^*$,
i.e., the elements
of $S^*$ mapping all of $A$ to zero.

\begin{definition}\label{de_linear_random_vars}
Let $\field=\integers/2\integers$, and let $S$ be an $\field$-vector
space.
By the {\em universe} associated to $S$ we mean the dual space
$\cU=S^*$; by a {\em linear random variable} we mean
a subspace $V\subset \cU$, to which we associate the 
classical linear random variable $V_{\rm class}\from S\to S/A$ where
$A$ is the annihilator of $V$ in $S$, i.e.,
$$
A = \{ s\in S \ |\  \forall \ell\in V,\ \ell(s)=0 \}.
$$
Therefore, $V$ equals the image of $(S/A)^*$ as it sits in $S^*$.
We define the entropy $H_2(V)$ to be that of $H_2(V_{\rm class})$.
Conversely, to any classical linear random variable, $Y$, we associate
the unique linear random variable by setting $A=\ker(Y)$
(so that $Y$ is equivalent to $S\to S/A$) and associating to
$Y$ the subspace of $V\subset S^*$ which is the image of $(S/A)^*$ as it sits
in $S^*$.
\end{definition} 
In the above definition we have
$$
H_2(V) = H_2(V_{\rm class}) = \dim(S/A) =  \dim\bigl( (S/A)^* \bigr) =
\dim(V).
$$
Hence the entropy of $V$ is just its dimension.

It will turn out to be far more convenient to think of a classical
linear random variable as its associated linear random variable,
a subspace of $S^*$.

The last thing to note is how joint random variables work in the above
context, i.e., the linear case.
If $Y_1,Y_2$ are two random variables, then their joint random
variable $(Y_1,Y_2)$
denotes the random variable that is the Cartesian product map
$$
(Y_1,Y_2)\from S \to \cY_1\times \cY_2,
\quad\mbox{taking $s$ to $\bigl( Y_1(s),Y_2(s) \bigr)$.}
$$
If $S,\cY_1,\cY_2$ are vector spaces and $Y_1,Y_2$ are linear maps,
then $\cY_1\times\cY_2$ becomes a vector space---merely the direct sum
of $Y_1$ and $Y_2$---and $(Y_1,Y_2)$
is a linear map.

Recall that if $S$ is any finite-dimensional
$\field$-vector space and $A\subset S$ is a subspace,
then the {\em annihilator of $A$ in $S^*$} is the set of elements of
$S^*$ taking all of $A$ to $0$, which is a subspace of dimension
$\dim(S)-\dim(A)$; similarly, if $V\subset S^*$, by 
the {\em annihilator of $V$ (in $S$)} we mean the elements of $S$
that each element of $V$ takes to zero, and that this is a subspace
of dimension $\dim(S)-\dim(V)$.

\begin{proposition}
Let $\field=\integers/2\integers$, and let $S$ be an $\field$-vector
space.
Let $V^1,V^2\subset \cU=S^*$ be linear random variables,
whose classical forms are 
$V^i_{\rm class}\from S\to S/A_i$ (hence $A_i\subset S$ is the annihilator
in $S$ of $V_i$).
Then the linear random variable associated to the classical random variable
$(V^1_{\rm class},V^2_{\rm class})$ is $V^1+V^2$ (i.e., their span).
\end{proposition}
[The essential point of the proof below is 
(a fairly standard fact) that the annihilator of
$A_1\cap A_2$ is $V_1+V_2$;
this can also be proven by observing (see below) that the annihilator
of $V_1+V_2$ is $A_1\cap A_2$, and using the (standard fact) that
the annihilator of the annihilator of a subspace is itself.]
\begin{proof}
The kernel of the map $(V^1_{\rm class},V^2_{\rm class})$ is the kernel
of the map 
$S \to (S/A_1)\times (S/A_2)$, which is clearly $A_1\cap A_2$.  So
let $V$ be $(S/(A_1\cap A_2))^*$ as it sits in $S^*$.

Notice the annihilator of $A^1+A^2$ is precisely $V^1\cap V^2$,
since (1) $V^1\cap V^2$ annihilates both $A^1+A^2$, and (2) any
element of $S^*$ that does not lie in $V^1\cap V^2$ fails to lie in
at least one of $V^1$ or $V^2$ and hence fails to annihilate at least 
one of $A^1$ or $A^2$.
Hence
$$
\dim\bigl( S/(A^1+A^2)\bigr) = \dim( V^1\cap V^2).
$$
By the dimension theorem we then have
$$
\dim(V)=\dim(S)-\dim(A_1\cap A_2) =
\dim(S)-\dim(A_1)-\dim(A_2)+\dim(A_1+A_2)
$$
$$
=
\dim(V^1)+\dim(V^2)-\dim\bigl( S/(A^1+A^2)\bigr)
$$
$$
=
\dim(V^1)+\dim(V^2)-\dim( V^1\cap V^2)=\dim(V^1+V^2).
$$
Since each 
$V_i$ takes all of $A_i$ to zero, each $V_i$ certainly takes all of
$A_1\cap A_2$ to $0$.  Hence $V^1+V^2\subset V$.  But the previous
paragraph shows that
$\dim(V)=\dim(V^1+V^2)$, and hence $V^1+V^2=V$.
Hence the annihilator of $A_1\cap A_2$ is precisely $V^1+V^2$.
\end{proof}

When we study coded caching, we will often use the notation for
joint random variables that is more common there.

\begin{notation}\label{no_joint_random_variables_information_theory}
If $V_1,\ldots,V_m\subset \cU$ are subspaces of an $\field$-universe,
$\cU$, we use the following notation as an alternative to
$V_1+\cdots+V_m$ (i.e., the span of $V_1,\ldots,V_m$):
(1) $(V_1,\ldots,V_m)$; (2) $V_1,\ldots,V_m$; or, most simply,
(3) $V_1\ldots V_m$.
\end{notation}

\subsection{The Dimension Formula in Infinite Dimensions}

There is a better way to state the dimension formula 
(in Subsection~\ref{su_dimension_formula_proof})
when $U_1,U_2\subset \cU$ are possibly infinite-dimensional subspaces
of an infinite dimension $\field$-vector space $\cU$, namely that
$$
0 \to U_1\cap U_2 \to U_1 \underline\oplus U_2 \to U_1+U_2 \to 0
$$
is an {\em exact sequence}
meaning that the kernel of any
arrow equals the image of the preceding arrow.
In algebraic topology (see, for example, Section~1.1 of \cite{bott}), 
one typically works with (co)chains of infinite-dimensional vector spaces,
yet where typically the (co)homology groups are finite-dimensional.
It is therefore likely that some of our discussion regarding
subspaces of an ambient $\field$-universe hold in the infinite-dimensional
setting, using tools that already exist.
However, it is usually simpler to work with finite-dimensional vector
spaces, and our applications to information theory in this article
involve only finite-dimensional vector spaces; hence in this article
we mostly 
limit ourselves to discussion and theorems 
regarding finite-dimensional vector spaces.

\section{Preliminary Remarks about Coordination and Discoordination, 
and Main Results}
\label{se_Joel_coord_discoord_defs}

In this section we define the notion of the ``discoordination'' of
a collection of subspaces of a universe, which is the focus of the
linear algebra in this article.
In case a collection of
subspaces have zero discoordination, then they are
``coordinated,'' which gives very simple formulas regarding
the dimensions of such subspaces and subspaces obtained by
applying operations like $+,\cap$ and taking quotients.
The fact that two subspaces are always coordinated, but three subspaces 
are not, is well-known
(see, for example, Exercise 9, Section~3.3 (page 51) \cite{janich}).

After defining coordination and discoordination and discussing some
of their basic properties,
we will state most of the main results we will prove regarding
linear algebra (i.e., in
Sections~\ref{se_Joel_discoord_formula}
to~\ref{se_Joel_three_subspaces_main}),
including all the results we require for our study of coded caching.

\subsection{Coordination}

If $X$ is a set of linearly independent vectors in an $\field$-universe,
$\Univ$, and $A\subset\Univ$ is a subspace, then $X\cap A$ is a set of
linearly independent vectors in $A$, and hence
\begin{equation}\label{eq_simple_observation_for_coord_discoord}
\dim(A)-|X\cap A|\ge 0
\end{equation} 
with equality iff $X\cap A$ is a basis of $A$.  This observation leads
to a number of definitions that are the focus of this article.

\begin{definition}
Let $\cU$ be an $\field$-universe.  We use the notation
$$
{\rm Ind}(\cU)=\{ X \subset \cU \ | \ \mbox{the elements of $X$ are
linearly independent} \}
$$
to denote the set of linearly independent subsets of $\cU$.
Let $A_1,\ldots,A_m$ be subspaces of $\cU$.
We say that a subset, $X$, of $\cU$ {\em coordinates} $A_1,\ldots,A_m$
if 
\begin{enumerate}
\item 
$X\in{\rm Ind}(\cU)$, i.e.,
$X$ is a set of linearly independent vectors in $\cU$, and
\item 
for all $i=1,\ldots,m$ we have
$$
\dim(A_i) = |X \cap A_i|,
$$
or, equivalently, $X\cap A_i$ is a basis for $A_i$ 
(since $X\cap A_i$ is a linearly independent 
{\mygreen set of vectors in $A_i$}
whose size equals the dimension of $A_i$).
\end{enumerate}
If such an $X$ exists, we
say that $A_1,\ldots,A_m$ are {\em coordinated}; 
{\mygreen we also say that
the set $\{A_1,\ldots,A_m\}$ {\em is coordinated}.}
\end{definition}

\begin{proposition}\label{pr_X_coordination_closed_cap_sum}
If $X$ coordinates subspaces $A_1,A_2$ of an $\field$-universe, $\cU$,
then $X$ also coordinates $A_1\cap A_2$ and $A_1+A_2$.
\end{proposition}
We will give two proofs of this proposition.  The first
uses the dimension formula (see 
Subsection~\ref{su_dimension_formula_proof}).
The second proof will be given at the end of
Subsection~\ref{su_coordinate_subspaces}.
\begin{proof}
In view of \eqref{eq_simple_observation_for_coord_discoord}, for any
$X\in{\rm Ind}(\cU)$ we have
\begin{equation}\label{eq_inequalities_to_prove}
|X\cap (A_1+A_2)| \le \dim(A_1+A_2), \quad
|X\cap (A_1\cap A_2)| \le \dim(A_1\cap A_2);
\end{equation} 
$X$ coordinates $A_1+A_2$ and $A_1\cap A_2$ iff both these
inequalities hold with equality;
the Inequality Summation Principle
(Subsection~\ref{su_inequality_summation_principle}) implies that
equality holds in both iff their sum,
\begin{equation}\label{eq_easy_observation_A_1_A_2}
|X\cap (A_1+A_2)| + |X\cap (A_1\cap A_2)| \le \dim(A_1+A_2)+\dim(A_1\cap A_2),
\end{equation} 
holds with equality.  
Let us show this.

By (set theoretic) inclusion-exclusion we have
$$
| X\cap (A_1\cup A_2) | + |X\cap (A_1\cap A_2)|=
|X\cap A_1|+|X\cap A_2|,
$$
which, since $X$ coordinates $A_1,A_2$, equals
$$
\dim(A_1)+\dim(A_2)=\dim(A_1+A_2)+\dim(A_1\cap A_2)
$$
using
the dimension formula.  We conclude that
\begin{equation}\label{eq_almost_the_reverse}
| X\cap (A_1\cup A_2) | + |X\cap (A_1\cap A_2)|=
\dim(A_1+A_2)+\dim(A_1\cap A_2).
\end{equation} 
However $A_1\cup A_2$ is a subset of $A_1+A_2$, and hence
\begin{equation}\label{eq_extra_bit_of_info}
|X\cap (A_1+A_2)| + |X\cap (A_1\cap A_2)|
\ge |X\cap (A_1\cup A_2)| + |X\cap (A_1\cap A_2)|,
\end{equation} 
and hence, by \eqref{eq_almost_the_reverse},
$$
|X\cap (A_1+A_2)| + |X\cap (A_1\cap A_2)|
\ge
\dim(A_1+A_2)+\dim(A_1\cap A_2).
$$
This is reverse inequality of \eqref{eq_easy_observation_A_1_A_2},
and hence both hold with equality.
\end{proof}
We remark that the fact that \eqref{eq_extra_bit_of_info} holds in
the proof shows that
{\mygreen $X\cap(A_1+A_2)$ and $X\cap(A_1\cup A_2)$ have the same size,
i.e.,}
any element of $X\cap (A_1+A_2)$
must also lie in $A_1\cup A_2$.

The proposition above has an easy consequence, whose proof we 
leave to the reader.
\begin{proposition}\label{pr_X_coordination_closed_quotient}
Say that $X$ coordinates a subspace $A_2$ of an $\field$-universe, $\cU$.
Then $[X\setminus A_2]_{A_2}$ 
{\mygreen is a linearly independent set}
in $\cU/A_2$.
If $X$ also coordinates a subspace $A_1\subset \cU$, then
$[X\setminus A_2]_{A_2}$ coordinates $[A_1]_{A_2}\subset \cU/A_2$.
\end{proposition}

It follows that if $\cU$ is an $\field$-universe and
$X\in{\rm Ind}(\cU)$ coordinates
a family of subspaces, $\cA$, then $X$ also coordinates any subspace
obtained by a finite sequence of spans and intersections of members of
$\cA$.

Another basic observation about coordination
is that if $\cU$ is an $\field$-universe
of dimension $n$, and $X\in{\rm Ind}(\cU)$, then $X$ contains at most
$n$ vectors and hence $X$ coordinates at most $2^n$ distinct subspaces
of $\cU$.

The main point of this article is to describe which subspaces $A_1,\ldots,A_m$
are coordinated, or, if not, to describe their ``discoordination,'' which
measures the extent to which they
``fail to be coordinated.'' 
Before discussing discoordination, let us give a helpful way of thinking
about coordinated subspaces.

\subsection{Coordinate Subspaces}
\label{su_coordinate_subspaces}

If $\field$ is a field, we use $\field^n$ to denote the usual
product of $n$ copies of $\field$, and use
$e_1,\ldots,e_n$ to denote the standard basis vectors of $\field^n$
(hence $e_i$ is a vector with a $1$ in the $i$-th coordinate and $0$'s
elsewhere).
For a subset $I\subset [n[=\{1,\ldots,n\}$, we set $e_I$ to be
$$
e_I = {\rm Span}\Bigl( \{e_i\}_{i\in I} \Bigr) \subset \field^n;
$$
hence $e_I$ is a subspace of dimension $|I|$ which we call
the {\em $I$-coordinate subspace of $\field^n$};
hence $e_\emptyset=\{0\}$ and $e_{[n]}=\field^n$.
We easily see that if $I,J\in[n]$, then
\begin{equation}\label{eq_coordinate_space_sum_cup}
e_I + e_J = e_{I\cup J}, \quad
e_I \cap e_J = e_{I\cap J}.
\end{equation} 

This gives us another view of coordination:
if $\cU$ is an $n$-dimensional $\field$-universe and $X=\{x_1,\ldots,x_n\}$
is a basis of $\cU$, then there is a unique isomorphism
$f\from\cU\to\field^n$ of vector spaces such that
$x_i\in \cU$ is taken to $e_i\in\field^n$.
In this case a subspace $A\subset \cU$ is coordinated by $X$
iff $f(A)$ is a coordinate subspace in $\field^n$.

\begin{proof}[Alternate proof
of Proposition~\ref{pr_X_coordination_closed_cap_sum}]
Let $A_1,A_2\subset \cU$ be coordinated
by $X\in{\rm Ind}(\cU)$; replace $X$ by an extension of
$X$ to a basis of $\cU$; clearly
{\mygreen such an extension}
also coordinates $A_1,A_2$.
Then, letting
$X=\{x_1,\ldots,x_n\}$,
there is a unique isomorphism
$f\from\cU\to\field^n$ taking $x_i$ to $e_i\in\field^n$, and
we have
$f(A_1)=e_I$ for the subset $I\subset [m]$ consists of those 
$i\in[{\mygreen n}]$
such that $e_i\in f(A_1\cap X)$;
similarly $f(A_2)=e_J$
{\mygreen for some $J\subset[n]$}.
Note that any isomorphism of vector spaces preserves the operations
$+,\cup$, and in particular this is true of 
$f^{-1}\from\field^n\to\cU$.
Hence,
in view of \eqref{eq_coordinate_space_sum_cup}, 
we have that
$$
A_1 + A_2 = f^{-1}(e_{I\cup J}), \quad
A_1\cap A_2 = f^{-1}(e_{I\cap J})
$$
are coordinated by $\{f^{-1}(e_1),\ldots,f^{-1}(e_n)\}=X$.
\end{proof}

{\mygreen
\begin{remark}
We warn the reader of one fundamental difference between 
subsets and subspaces:
namely, if $I,J\subset[n]$ for some $n$, then their usual
Venn diagram contains three pieces,
$$
I\setminus J, \ I\cap J, \ J\setminus I ,
$$
all of which are subsets of $I\cup J$; $I$ is the union of
the first two pieces above, and $J$ of the last two.
However, the closest analogous ``Venn diagram'' for two subspaces
$A,B\subset\cU$ of some universe consists of the three ``pieces''
\begin{equation}\label{eq_two_subspaces_three_universes}
A/B \subset \cU/B,
\quad
A\cap B\subset \cU,
\quad
B/A \subset \cU/A,
\end{equation} 
each of which lies in a different universe.
What is true is that if $A'$ is a complement of $A\cap B$ in $A$,
and $B'$ one of $A\cap B$ in $B$, then
$$
A',\ A\cap B,\ B' \subset (A+B) \subset\cU,
$$
and $A$ is isomorphic to the direct sum of its subspaces
$A'$ and $A\cap B$, and similarly for $B$.
But the choice of $A',B'$ is not canonical.
One way to choose pick an $A'$ is to pick a basis for $A$
relative to $A\cap B$, as discussed in 
Subsection~\ref{su_quotient_space}.
We easily see that an equivalent way to construct $A'$ (and similarly
for $B'$) is 
to choose an isomorphism $f\from\cU\to\field^n$ as done in the
proof above,
with $f(A)=e_I$ and $f(B)=e_J$.
Then $e_I$ is isomorphic to the direct sum of its subspaces
$e_{I\setminus J}$ and $e_{I\cap J}$, and so $A'=f^{-1}(e_{I\setminus J})$
is a complement of $A\cap B$ in $A$.
\end{remark}
}

\subsection{Discoordination and Minimizers}

Next we define 
{\mygreen a}
measure of ``the extent to which given subspaces of a 
universe may fail to be coordinated.''

\begin{definition}
If $\cU$ is an $\field$-universe, we use ${\rm Ind}(\cU)$ to denote
the set of all linearly independent subsets $X\subset \cU$.
Let $A_1,\ldots,A_m$ be subspaces of an $\field$-universe, $\cU$.
If $X\in{\rm Ind}(\cU)$ (i.e., $X$ is a subset of linearly independent 
vectors in $\cU$), we define
the {\em discoordination of $X$ with respect to $A_1,\ldots,A_m$} to be
$$
{\rm DisCoord}_X (A_1,\ldots,A_m) =
\sum_{i=1}^m \Bigl( \dim(A_i)-|X\cap A_i| \Bigr).
$$
We define
the {\em discoordination of $A_1,\ldots,A_m$} to be
$$
{\rm DisCoord}(A_1,\ldots,A_m)
= \min_{X\in{\rm Ind}(\cU)} {\rm DisCoord}_X (A_1,\ldots,A_m),
$$
where ${\rm Ind}(\cU)$ denotes the set of all linearly independent
subsets $X\subset \cU$;
and we call any $X\in{\rm Ind}(\cU)$ at which the above minimum is
attained
a {\em discoordination minimizer} (or simply a {\em minimizer})
{\em of $A_1,\ldots,A_m$}.
\end{definition}

In view of \eqref{eq_simple_observation_for_coord_discoord},
$A_1,\ldots,A_m$ are coordinated iff their discoordination
equals $0$, and, if so, then $X$ is a minimizer of $A_1,\ldots,A_m$
iff $X$ coordinates $A_1,\ldots,A_m$.

Notice also that in the above definition, if $X\subset X'$ and 
$X'\in{\rm Ind}(\cU)$, then
$$
{\rm DisCoord}_X (A_1,\ldots,A_m) \ge 
{\rm DisCoord}_{X'} (A_1,\ldots,A_m) .
$$
It follows that if $X$ is a minimizer of $A_1,\ldots,A_m$, then if $X$
is not a basis of $\cU$ we can extend 
$X$ to obtain a basis $X'$ of $\cU$ containing $X$, 
{\mygreen which leaves the discoordination unchanged.}
Hence there 
{\mygreen exists a minimizer that is a basis for $\cU$}.

\subsection{The Main Theorem Regarding Three Subspaces}

In this article we develop some foundational theorems regarding
coordination and discoordination.
Our main theorem regarding three subspaces is the following.

\begin{theorem}\label{th_main_three_subspaces_decomp}
Let $A,B,C$ be three subspaces of an arbitrary $\field$-universe, $\cU$.
Then there is a 
decomposition $\cU_1,\cU_2$ of $\cU$ through which $A,B,C$ all
factor, such that
\begin{enumerate}
\item 
$A\cap \cU_1,B\cap \cU_1,C \cap \cU_1$ 
are coordinated in $\cU_1$, and
\item
there is an isomorphism $\mu\from \cU_2\to \field^2\times\field^m$ which takes 
$A\cap \cU_2,B\cap \cU_2,C \cap \cU_2$, respectively, to
$$
{\rm Span}(e_1)\otimes\field^m,
\ {\rm Span}(e_2)\otimes\field^m,
\ {\rm Span}(e_1+e_2)\otimes\field^m.
$$
\end{enumerate}
Furthermore, let 
$$
S_2 =S_2(A,B,C) = (A\cap B) + (A\cap C) + (B\cap C).
$$
Then the following integers are equal:
\begin{enumerate}
\item
$m$ as above;
\item ${\rm DisCoord}(A,B,C)$;
\item 
the minimum of $\dim(C)-|X\cap C|$ over all
$X\in{\rm Ind}(\cU)$ that coordinate $A$ and $B$ (and such an $X$ exists);
\item 
the dimension
in $\cU/S_2$ of the space
$([A]_{S_2}+[B]_{S_2}) \cap [C]_{S_2}$; 
\item 
the dimension
in $\cU/S_2$ of the space
$[(A+B)\cap C]_{S_2}$; and
\item
{\mygreen
(of course) any of (2)---(5) with $A,B,C$ permuted in some fashion
(since ${\rm DisCoord}(A,B,C)$ does not depend on how we order
$A,B,C$).}
\end{enumerate}
\end{theorem}

After proving this theorem, we will be able to write a number
of important formulas involving $A,B,C$ in terms of the
discoordination.  
Let us first state the general principle.

\begin{definition}
Let $f=f(A_1,\ldots,A_s)$ be a formula that is an $\integers$-linear
combination of terms of the form 
$\dim^{\cU/A''}([A']_{A''})$, where 
$A',A''$ are formulas in the operations $\cap,+$ and 
the variables $A_1,\ldots,A_s$ (and parenthesis) in an $\field$-universe,
$\cU$;
hence $f$ is a function that
takes arbitrary subspaces $A_1,\ldots,A_m$ of some
$\field$-universe, $\cU$, and returns an integer.
We say that $f$ is a {\em balanced formula} if
$f(A_1,\ldots,A_s)=0$ whenever $A_1,\ldots,A_s$ are coordinated.
\end{definition}

\begin{example}
The following are examples of balanced functions $f=f(A,B,C)$:
\begin{enumerate}
\item 
$
\dim\bigl( (A+B)\cap C \bigr) -
\dim(A\cap C)-\dim(B\cap C) + \dim(A\cap B\cap C)
$;
\item
$
\dim(A\cap B)-\dim^{\cU/C}\bigr([A]_C \cap [B]_C \bigr)
-\dim(A\cap B\cap C)
$;
\item
$
\dim(A\cap B)-\dim^{\cU/C}\bigr([A\cap B]_C \bigr)
-\dim(A\cap B\cap C)
$;
and
\item
$
\dim\bigl( I(A;B;C) \bigr) - \dim(A\cap B\cap C)
$,
with $I(A;B;C)$ as in \eqref{eq_three_way_mutual_eq1} and
\eqref{eq_three_way_mutual_eq2}.
\end{enumerate}
To verify that these formulas are balanced,
it suffices to take $A,B,C$ to equal the coordinate
subspaces
$e_I,e_J,e_K$ with $I,J,K$ subsets of a finite set,
(with notation as in Subsection~\ref{su_coordinate_subspaces}),
whereupon the dimensions of the subspaces in these formulas
amount to the sizes of unions and intersections of $I,J,K$.
See also the algorithm with Venn diagrams, e.g.,
Figure~\ref{fi_venn_diagram_balanced}.
\end{example}

\begin{corollary}\label{co_balanced_formula_and_discoordination}
Let $f=f(A,B,C)$ be a balanced formula.  Then for any
subspaces $A,B,C$ we have
$$
f(A,B,C)= k\ {\rm DisCoord}(A,B,C),
$$
where
\begin{equation}\label{eq_k_given_by_a_balanced_formula}
k = f\bigl( 
{\rm Span}(e_1),
{\rm Span}(e_2),
{\rm Span}(e_1+e_2)  \bigr),
\end{equation} 
where $e_1,e_2$ are the standard basis vectors in $\field^2$ for
any field $\field$.
\end{corollary}
The corollary is an immediate consequence of
Theorem~\ref{th_main_three_subspaces_decomp} and
Theorem~\ref{th_factorization_under_cap_sum} and the paragraph
just below it, since 
{\mygreen together they imply, with notation as in
Theorem~\ref{th_main_three_subspaces_decomp}, that}
$$
f(A,B,C) = \sum_{i=1}^2 f(A\cap \cU_i,B\cap \cU_i, C\cap \cU_i);
$$
the $i=1$ term 
{\mygreen above}
vanishes since this term involves coordinated 
subspaces, and the $i=2$ term is isomorphic to the direct sum of
$m$ copies of $\field^2$, in which $f$ restricted to each copy equals
$k$ above.

{\mygreen
\begin{figure}[h]
\caption{Checking $f(A,B,C)=0$ when $A=e_I$, $B=e_J$, $C=e_K$ in Item~(6)
of Corollary~\ref{co_various_formulas_for_discoordination_of_three}}
\label{fi_venn_diagram_balanced}
\end{figure}
}

\begin{corollary}\label{co_various_formulas_for_discoordination_of_three}
Let $A,B,C$ be three subspaces of an arbitrary $\field$-universe, $\cU$.
Then ${\rm DisCoord}(A,B,C)$ also equals:
\begin{enumerate}
\item $\dim(A\cap B\cap C)-I(A;B;C)$, where
\begin{align*}
I(A;B;C) & \eqdef\dim(A+B+C) - \dim(A+B) - \dim(A+C)
- \dim(B+C)  \\
& + \dim(A) + \dim(B) + \dim(C).
\end{align*}
\item
$
\dim(C \cap (A+B)) - \dim(C \cap A) - \dim(C \cap B) + \dim(A \cap B \cap C);
$
\item
$
\dim^{\cU/C}([A]_C\cap [B]_C) + \dim(A \cap B \cap C) - \dim(A\cap B)
$
\item
$
\dim((A+ C )\cap (B+ C )) - \dim(C) + \dim(A \cap B \cap C)
-\dim(A\cap B)
$;
{\mygreen  
\item
$
\dim\bigl( (A+C)\cap (B+C) \bigr) - \dim\bigl( (A\cap B) + C \bigr)
$;
\item
$I(B;C|A)-\dim^{\cU/A}([B\cap C]_A)$, where
$$
I(B;C|A) \eqdef \dim(A+B) + \dim(A+C) - \dim(A) - \dim(A + B + C) ;
$$
}
and
\item
(of course)
the same expression as in~(1)--(6) with $A,B,C$ permuted in any
order.
\end{enumerate}
(Here $\dim$ refers to $\dim^\cU$ unless otherwise indicated.)
\end{corollary}
{\mygreen
To prove the corollary, one easily checks that all the expressions
$f(A,B,C)$ of items~(1)--(6) are balanced equations, and
have $k=1$ in \eqref{eq_k_given_by_a_balanced_formula}.
In this above, the notation
$I(B;C|A)$ is the information theory analog;
see, e.g., \cite{hammerEtAl2000}, just below~(13), page~456.

\begin{remark}
To check that an $f(A,B,C)$ is balanced, 
one can set
$A=e_I$, $B=e_J$, $C=e_K$,
draw Venn diagrams,
and check that the coefficients in each piece adds to zero.  An example
is given in Figure~\ref{fi_venn_diagram_balanced} for
item~(6) in Corollary~\ref{co_various_formulas_for_discoordination_of_three}.
\end{remark}

{\mygreen
\begin{figure}[h]
\caption{Size of Venn diagram pieces and $A=e_I$, $B=e_J$, $C=e_K$.} 
\label{fi_venn_diagram_depiction}
\end{figure}
}

\begin{remark}
The reader may be put off by expressions such as
$\dim^{\cU/A}([B\cap C]_A)$.  However, such expressions are
the natural way on describes pieces of the ``Venn diagram'' of
$I,J,K$ when $A=e_I$, $B=e_J$, $C=e_K$.
By our conventions, $\dim^{\cU/A}([B\cap C]_A)$ can be
written more briefly as
$(B\cap C)/A$, but we usually prefer the longer notation for clarity,
i.e., to emphasize the universe and quotienting involved.
We illustrate this in
Figure~\ref{fi_venn_diagram_depiction}, which illustrates the
size of the corresponding piece.
Regarding piece 2 in this figure, we have
\begin{equation}\label{eq_equality_example_when_coordinated}
\mbox{${\rm DisCoord}(A,B,C)=0$}
\ \implies\ %
(A\cap B)/C = \bigl( (A+C)\cap (B+C) \bigr) / C,
\end{equation} 
although $(A\cap B)/C$ is a proper subset of 
$\bigl( (A+C)\cap (B+C) \bigr) / C$
when $A,B,C$ are discoordinated; hence the coordinated case is
simpler due to equalities such as 
\eqref{eq_equality_example_when_coordinated}.
Note also that piece 5 is the only one that lives in $\cU$ itself;
any other expression that lives in $\cU$,, e.g., $\dim^\cU(A\cap B)$,
involves piece 5 plus and some other piece(s).
\end{remark}

}

In our study of coded caching we will need the following theorem,
which studies how the discoordination of $A,B,C$ in an $\field$-universe,
$\cU$ changes when considering the image of $A,B,C$ in a quotient
universe $\cU/D$ for some subspace $D\subset \cU$.
Before stating this theorem,
we remark that without assumptions on $D$,
$$
{\rm DisCoord}^{\cU/D}\bigl( [A]_D,[B]_D,[C]_D \bigr)
$$
can be larger or smaller than ${\rm DisCoord}^\cU(A,B,C)$, as the
following examples show:
\begin{enumerate}
\item
if $\cU=\field^3$ and 
$A={\rm Span}(e_1)$, $B={\rm Span}(e_2)$, $C={\rm Span}(e_3)$,
and $D={\rm Span}(e_1+e_2+e_3)$, then
$A,B,C$ have zero discoordination in $\cU$, but
$\cU/D$ is two dimensional and $[A]_D,[B]_D,[C]_D$ are three distinct
one dimensional subspaces, hence are discoordinated in $\cU/D$;
hence the discoordination can increase 
when passing from $\cU$ to $\cU/D$;
and
\item
if $A,B,C$ have positive discoordination in $\cU$ and $D=\cU$,
then, of course, their discoordination in $\cU/D=\{0\}$ is zero;
hence the discoordination can decrease 
when passing from $\cU$ to $\cU/D$.
\end{enumerate}

\begin{theorem}\label{th_quotient_via_subspace_in_two}
Let $A,B,C,D$ be four subspaces of an arbitrary $\field$-universe, $\cU$,
such that $D\subset A\cap B$.
Then
$$
{\rm DisCoord}^\cU(A,B,C)
=
{\rm DisCoord}^{\cU/D}\bigl( [A]_D,[B]_D,[C]_D \bigr),
$$
i.e., the discoordination of $A,B,C$ in $\cU$ is the same as that
of the images of $A,B,C$ in the quotient $\cU/D$.
\end{theorem}
We will 
{\mygreen prove this result in 
Subsection~\ref{su_discoordination_from_cU_to_cU_mod_D_two}.}

In the next two subsections we give some important general results about
discoordination and about decompositions and factorization;
these results
will be helpful in our proof of 
Theorem~\ref{th_main_three_subspaces_decomp} and in our 
proof of
Theorem~\ref{th_quotient_via_subspace_in_two}.

\subsection{The Discoordination Formula}

Some of the results in this paper are based on a detailed
description of how to build minimizers for subspaces $A_1,\ldots,A_m$
in a universe.  This description gives an interesting ``formula''
for the discoordination which we will use to prove
Theorem~\ref{th_quotient_via_subspace_in_two}.
Both results are stated as a single theorem, namely
Theorem~\ref{th_greedy_algorithm};
however, in this section we will use only the second
result, which we now state separately.

\begin{theorem}\label{th_discoordination_formula}
Let $A_1,\ldots,A_m$ be subspaces of an $\field$-universe, $\cU$.
For each $k=1,\ldots,m$, let $S_k$ be the span of all intersections
of any $k$ of $A_1,\ldots,A_k$, i.e., 
$$
S_k = \sum_{1\le i_1<\ldots<i_k\le m} A_{i_1}\cap \ldots \cap A_{i_k}
$$
(see also Definition~\ref{de_k_fold_sums}).
Then 
$$
{\rm DisCoord}(A_1,\ldots,A_m) =
\sum_{i=1}^m \dim(A_i) -
\sum_{i=1}^m \dim(S_i).
$$
\end{theorem}

When we study coded caching, we will see that it is usually
difficult to determine $S_2$, 
{\mygreen 
and often other of the
$S_i$ with $i\ge 3$;
}
this generally requires detailed information
on the way $A_1,\ldots,A_m$ are related to each other
as subspaces of $\cU$;
in our applications only
$S_1=A_1+\cdots+A_m$ will be easy to determine.
Hence Theorem~\ref{th_discoordination_formula} gives only
partial insight into the discoordination of three or more subspaces.

\subsection{Factorization and Discoordination}

Our discoordination formula, 
Theorem~\ref{th_discoordination_formula},
has the following important consequence, in view of
Theorem~\ref{th_factorization_under_cap_sum}.

\begin{theorem}\label{th_discoordination_factors}
Let $A_1,\ldots,A_m$ be subspaces of an $\field$-universe $\cU$ that
all factor through a decomposition $\cU_1,\ldots,\cU_r$ of $\cU$.
Then
$$
{\rm DisCoord}^\cU(A_1,\ldots,A_m) 
=
\sum_{i=1}^r
{\rm DisCoord}^{\cU_i}\bigl( A_1\cap \cU_i,\ldots,A_m\cap \cU_i \bigr).
$$
\end{theorem}
\begin{proof}
Each of the $S_i$ in
Theorem~\ref{th_discoordination_formula} is a formula involving
$+,\cap$ and the subspaces $A_1,\ldots,A_m$
(and parenthesis); hence
Theorem~\ref{th_factorization_under_cap_sum} implies
that 
the dimensions of the subspaces $S_1,\ldots,S_m$ in
Theorem~\ref{th_discoordination_formula}
can be computed as the sum over $i$ of the dimensions of the
analogs of
$S_1,\ldots,S_m$ of $A_1\cap \cU_i,\ldots,A_m\cap\cU_i$.
\end{proof}


We can explain 
{\mygreen in rough terms}
how the above theorem is used in
our proof 
of Theorem~\ref{th_quotient_via_subspace_in_two}:
if $A,B,C\subset \cU$ are subspaces of an $\field$-universe
$\cU$, and $D\subset A\cap B$, then for the decomposition
of $\cU$ as $\cU_1,\cU_2$ given
in Theorem~\ref{th_main_three_subspaces_decomp}
we have $D\subset \cU_1$
(essentially since ${\rm Span}(e_1)\cap{\rm Span}(e_2)=0$).
It follows that $\cU/D$ is isomorphic to 
{\mygreen the direct sum of}
$\cU_1/D$ and $\cU_2$.
It follows that the discoordination of the images of $A,B,C$ in $\cU_2$ is
unchanged when passing from $\cU$ to $\cU/D$;
one then has to prove that the discoordination of
the images of $A,B,C$ in
$\cU_1/D$ remains equal to zero.
After doing so, we apply Theorem~\ref{th_discoordination_factors}.

\subsection{Addition Results about Coordination}

We finish this section by stating one more result on coordination,
namely Theorem~\ref{th_two_increasing_sequences} and the resulting
corollary.
We found this theorem convenient in studying coded caching,
although we have avoided its use in this article.
The result is that
if $A_1\subset\cdots A_s$ and
$B_1\subset \cdots B_t$ are two increasing subsequences of subspaces
of an $\field$-universe, then all these subspaces are coordinated.
In particular, the case $s=1$ and $t=2$ implies the
sometimes convenient fact that
if $A,B,C$ are subspaces with $B\subset C$, then
the discoordination of $A,B,C$ vanishes.
We also note that 
the case $s=1$ and $t=1$ of
Corollary~\ref{co_two_increasing_sequences}
is 
{\mygreen the key to proving}
the dimension formula.

\section{Quasi-Increasing Sequences are Coordinated, and Applications}
\label{se_Joel_coordination_method}

The main goal of this section is to 
{\mygreen prove a number of theorems that state that certain sequences
of subspaces of a universe are coordinated.
The proofs can be given ``from scratch,'' but to simplify
the proofs we will introduce a notion of
{\em quasi-increasing sequences.}

This section begins by stating all the coordination theorems.
We then discuss quasi-increasing sequences, and use this idea
to prove all the coordination theorems.
We finish this section by discussing the fact that all the
{\em quasi-increasing sequences} in this section 
satisfy a stronger property, that we call
{\em strongly quasi-increasing}.
}

\subsection{Statement of Some Coordination Theorems}

\begin{theorem}\label{th_two_increasing_sequences}
Let $\cU$ be an $\field$-universe, and let 
$$
A_1\subset\cdots\subset A_s, \quad
B_1\subset\cdots\subset B_t
$$
be two sequences of increasing subspaces of $\cU$.
Then the set of subspaces $\{A_i\cap B_j\}_{i,j}$ ranging over 
all $i\in[s]$ and $j\in [t]$ are coordinated.
\end{theorem}
This theorem has the following corollary.
\begin{corollary}\label{co_two_increasing_sequences}
Let $\cU$ be an $\field$-universe, and let 
$$
A_1\subset\cdots\subset A_s, \quad
B_1\subset\cdots\subset B_t
$$
be two sequences of increasing subspaces of $\cU$.
Then the subspaces $A_1,\ldots,A_s,B_1,\ldots,B_t$ are coordinated.
\end{corollary}
The corollary is obtained from the theorem by extending the sequences
of vector spaces by setting $A_{s+1}=B_{t+1}=\cU$;
then for all $i\in[s]$, $A_i\cap B_{t+1}=A_i\cap\cU=A_i$ and similarly
$A_{s+1}\cap B_j = B_j$ for all $j\in[t]$.
The corollary is more succinct, since the intersection of any two
subspaces coordinated by some basis $X$ is again coordinated by $X$.

For our analysis of coded caching, we have found the result
with $s=2$ and $t=1$ helpful.  However in simplifying our results
we have been able to forgo any use of the above theorem.

Note that if $s=t=2$ and we
$A_2=B_2=\cU$, the above theorem implies that $A_1\cap B_1,A_1,B_1$
are coordinated, which is how one proves
the dimension formula.  Hence
Theorem~\ref{th_two_increasing_sequences} can be viewed as a generalization
of the dimension formula.

The other main theorem in this section is the following.
\begin{theorem}\label{th_six_out_of_seven}
Let $A,B,C$ be subspaces of an $\field$-universe, $\cU$.
Then the six spaces
$$
A\cap B\cap C,\ A\cap B, \ A\cap C,\ B\cap C,\ A,\ B
$$
are coordinated.
\end{theorem}
Of course, 
{\mygreen by Proposition~\ref{pr_X_coordination_closed_cap_sum},}
in this theorem
it suffices to state $A\cap C,B\cap C,A,B$ are coordinated; we include
the other subspaces since they will also be used explicitly to
prove Theorem~\ref{th_main_three_subspaces_decomp}.
This theorem will be crucial to our theorem about the discoordination
of three subspaces.
One proof of Theorem~\ref{th_quotient_via_subspace_in_two}
we will involve part of the following minor improvement of
Theorem~\ref{th_six_out_of_seven}.

\begin{theorem}\label{th_six_out_of_seven_and_D}
Let $A,B,C$ be subspaces of an $\field$-universe, $\cU$,
and $D\subset A\cap B$ another subspace.
Then the spaces
$$
A\cap B\cap C\cap D,
\ A\cap B\cap C,\ D,\ A\cap B, \ A\cap C,\ B\cap C,\ A,\ B
$$
are coordinated.
\end{theorem}
This theorem implies that there is always a minimizer of $A,B,C$
that coordinates $D$, see the proof of 
Lemma~\ref{le_coordinating_with_D} and the remark below it; in fact
this remark shows that we can alternatively 
coordinate the same eight subspaces in
Theorem~\ref{th_six_out_of_seven_and_D} where the subspace $B$ (or $A$)
replaced with $C$.

{\mygreen
Here is the last main theorem of this section.

\begin{theorem}\label{th_m_minus_one_fold_intersections_are_coordinated}
Let $A_1,\ldots,A_m$ be subspaces of an $\field$-universe, $\cU$.
For each $i\in[m]$ let
\begin{equation}\label{eq_A_without_notation}
\Awithout{i}
\eqdef  \bigcap_{j\ne i} A_j 
= A_1\cap\ldots A_{i-1}\cap A_{i+1}\cap\ldots\cap A_m,
\end{equation} 
and let $V_0 = A_1\cap\ldots\cap A_m$.
Then $V_0,\Awithout{1},\ldots,\Awithout{m}$ are coordinated.
\end{theorem}
In other words, the set of all $(m-1)$-fold intersections of
$A_1,\ldots,A_m$ are coordinated
(and therefore so is their intersection, namely $A_1\cap\ldots\cap A_m$).
}

It turns out that all the theorems stated above
can be proven by a strategy
that generalizes the proof of the dimension formula,
which we now describe.

\subsection{Quasi-Increasing Sequences}

If $V_1\subset\cdots V_m$ are a set of increasing subspaces of
some universe, then one easily argues that this sequence is coordinated:
one begins with a basis for $V_1$, and successively increases this
to a basis for $V_2$ and so on.
In this subsection we give a more general situation where a similar
strategy works. 

\begin{definition}\label{de_quasi_increasing}
Let $V_1,\ldots,V_m$ be a sequence of vector spaces in some universe.
For $r=2,\ldots,m$,
we say that this sequence is {\em quasi-increasing in position $r$} 
(or {\em at $V_r$}) if whenever
\begin{equation}\label{eq_write_v_r_as_included_vectors}
v_r = v_1+\cdots+v_{r-1} \quad\mbox{such that}\quad
v_1\in V_1,\ldots,v_r\in V_r,
\end{equation} 
one also has 
\begin{equation}
{\mygreen \label{eq_write_v_r_as_included_vectors_primed}}
v_r = v'_1+\cdots+v'_{r-1} \quad\mbox{for some}\quad
v'_1\in V_1,\ldots,v'_{r-1}\in V_{r-1},
\end{equation} 
where in addition
\begin{equation}
{\mygreen\label{eq_write_v_r_as_included_non_included_primed_vanish}}
v'_i \ne 0 \implies V_i\subset V_r
\end{equation} 
(i.e., if $i<r$ and $V_i\not\subset V_r$ then $v'_i=0$).
Furthermore, if this condition holds for all $r=2,\ldots,m$,
we say that $V_1,\ldots,V_m$ is 
{\em quasi-increasing}.
\end{definition}
We easily see that to be quasi-increasing in position $r$ 
is equivalent to
\begin{equation}\label{eq_quasi_increasing_second_condition}
V_r\cap (V_1+\cdots+ V_{r-1}) \subset 
\sum_{i<r\ {\rm and}\ V_i\subset V_r} V_i;
\end{equation} 
the reverse inclusion is clear, so we can replace $\subset$ with $=$
if we like.

{\mygreen
Of course, any sequence $V_1,\ldots,V_m$ is quasi-increasing in
position $r$ if $V_1,\ldots,V_{r-1}\subset V_r$,
and hence}
any increasing sequence $V_1\subset\cdots\subset V_m$
is also quasi-increasing, 

\begin{example}\label{ex_three_subspaces_dim_formula}
Let $A,B$ be any vector spaces in some universe, and let 
$V_1=A\cap B$, $V_2=A$, $V_3=B$.
Then $V_1\subset V_2$, but $V_2\not\subset V_3$.
However, if $v_1\in V_1$, $v_2\in V_2$, and $v_3\in V_3$ with
$$
v_3 = v_1+v_2,
$$
then in fact $v_1'=v_1+v_2$ also lies in $V_1$
(one sees this by first noting that $v_2=v_3-v_1$, and since $v_1,v_3\in B$
then also $v_2\in B$;
since $v_2\in A$ then $v_2\in A\cap B=V_1$).
Hence the above sequence is quasi-increasing, but not generally increasing.
\end{example}
The above example is the essential step in proving
the dimension formula: namely,
we let $X_1$ be a basis for $V_1=A\cap B$, $X_2$ a minimal set
such that $X_1\cup X_2$ spans $A$,
and $X_3$ a minimal set such that $X_1\cup X_3$ spans $B$.
We then see that $X_1\cup X_2$ is a basis for $V_2=A$;
to show that $X_1\cup X_2\cup X_3$ is a basis, we need to show
that there is no nontrivial relation between the vectors of $X_3$
and those of $X_1$ and $X_2$; but if so then we have
$$
v_1+v_2= v_3 \ne 0
$$
where each $v_i$ is a linear combination of vectors in $X_i$;
but then we have, as shown above, $v'_1=v_1+v_2$ actually
lies in $V_1$, which contradicts the fact that $X_1\cup X_3$ 
are linearly independent.

Hence the theorem below strengthens the method used to prove
the dimension formula.

\begin{theorem}\label{th_quasi_increasing_coordinated}
Any quasi-increasing sequence is coordinated.
In more detail,
let $V_1,\ldots,V_m$ be a sequence of quasi-increasing subspaces in
some universe.  
Let $X_1$ be any basis for $V_1$, and inductively
on $i=2,\ldots,m$, let $X_i$ be a minimal size 
set of vectors such that if
$$
X' = \bigcup_{i'\ s.t.\ V_{i'}\subset V_i} X_{i'} ,
$$
then $X_i\cup X'$ spans $V_i$.
Then $X_1,\ldots,X_m$ are pairwise disjoint and
$X=X_1\cup\cdots\cup X_m$ coordinate $V_1,\ldots,V_m$, and, more specifically,
for each $i$ we have
$$
X\cap V_i = \bigcup_{V_{i'}\subset V_i} X_{i'}
$$
is a basis for $V_i$.
\end{theorem}
\begin{proof}
We prove this by induction on $m$.  The base case $m=1$ is clear
since $X_1$ is simply a basis for $V_1$.

Now say that the theorem holds for some value of $m\ge 1$,
let $V_1,\ldots,V_m,V_{m+1}$ be a quasi-increasing sequence,
and $X_1,\ldots,X_m$ vectors as in the theorem.
Let $I=\{ i\in [m] \ | \ V_i\subset V_{m+1}\}$.
By Proposition~\ref{pr_X_coordination_closed_cap_sum},
$X'=\cup_{i\in I}X_i$ coordinates
\begin{equation}\label{eq_V_prime_sum_lower_Us}
U' = \sum_{i\in I} V_i ;
\end{equation} 
since $V_i\subset V_{m+1}$ for all $i\in I$,
we have $U'\subset V_{m+1}$.
Let $X_{m+1}$ be as specified in the theorem.  Then the vectors
$X'\cup X_{m+1}$ are linearly independent (and $X_{m+1}$ is disjoint
from $X'$).  
By assumption, $X_1\cup\cdots\cup X_m$ are (pairwise disjoint and)
linearly independent.
Hence if $X_1\cup\cdots\cup X_{m+1}$ is not linearly
independent (or if $X_{m+1}$ is not distinct from $X_1,\ldots,X_m$),
we have
$$
v_{m+1} = v_1 + \cdots + v_m,
$$
where each $v_i$ is in the span of $X_i$ and $v_{m+1}$ is nonzero.
But then we have 
$$
v_{m+1} = v_1'+\cdots +v'_m,
$$
where $v'_i\ne 0$ only if $i\in I$; hence $v'_1+\cdots+v'_m\in U'$,
which contradicts the fact that $X'\cup X_{m+1}$ are linearly
independent.
\end{proof}
The main subtlety in the above proof is to consider
$U'$ in \eqref{eq_V_prime_sum_lower_Us} which is the sum of all
subspaces that lie in $V_{m+1}$.
The dimension formula works with $V_1,V_2,V_3$ as in
Example~\ref{ex_three_subspaces_dim_formula},
where $V_1$ is the sole subspace that lies in $V_3$.

\subsection{Quasi-Increasing Sequences and Maximal Elements}

Say that we are trying to prove that a given sequence
$V_1,\ldots,V_m$ of subspaces in a universe is quasi-increasing.
Hence, given any $r$ between $2$ and $m$ and any equation
$$
v_r = v_{r-1}+\cdots+v_1,
$$
we wish to find $v'_1,\ldots,v'_{r-1}$ as in
Definition~\ref{de_quasi_increasing}, i.e., whose sum is also $v_r$,
with $v'_i\in V_i$ for $i\in[r-1]$, but such that $v'_i=0$ if 
$V_i\not\subset V_r$.
In practice one can simplify this task by noting that if $i<j<r$
and $V_i\subset V_j$, then we can always assume that $v_i=0$, by
replacing $v_j$ with $v_j+v_i$.
{\mygreen
We now make this precise; the reader may prefer to
skip directly to Corollary~\ref{co_maximals_suffice}, which is pretty
clear without the formalities below.}

\begin{definition}
Let $V_1,\ldots,V_m$ be a sequence of vector spaces in a universe, $\cU$.
For $r=[m]$ we define the {\em $r$-maximal index set}, denoted $I_r$, to be
\begin{equation}\label{eq_r_maximal}
I_r=
\{ i\in [r] \ | \ \mbox{if $j\ne i$ and $j\le r$, then $V_i\not\subset V_j$}\}
\end{equation} 
(equivalently $i\in I_r$ if $V_i$ is a maximal subspace under inclusion
among $V_1,\ldots,V_r$).
\end{definition}

\begin{proposition}\label{pr_maximals_suffice}
Let $V_1,\ldots,V_r$ be a sequence of vector spaces in a universe, $\cU$.
Then 
\begin{equation}\label{eq_sum_V_i_equals_sum_of_maximal}
V_1+\cdots+V_r = \sum_{i\in I_r} V_i,
\end{equation} 
i.e., any sum of elements in $V_1,\ldots,V_r$ can be written as a sum
of those $V_i$ that are maximal under inclusion.
Similarly if $I_r$ is replaced with any larger set in $[r]$, i.e.,
if $I_r$ contains all $i$ with $V_i$ maximal under inclusion
among $V_1,\ldots,V_r$.
\end{proposition}
{\mygreen
\begin{proof}
Clearly there is a map $f\from[r]\to I_r$ such that
for all $i\in[r]$, $V_i\subset V_{f(i)}$
{\mygreen (the reader
can easily supply a formal proof by induction on $r$).}
Hence
$$
\sum_{i=1}^r V_i \subset
\sum_{i=1}^r V_{f(i)} \subset
\sum_{i\in I_r} V_i.
$$
Clearly this also holds if $I_r$ is replaced with any larger subset of $[r]$.
\end{proof}

\begin{corollary}\label{co_maximals_suffice}
To show that a sequence $V_1,\ldots,V_r$ is quasi-increasing in
position $r$ (i.e., at $V_r$), we can assume that in
\eqref{eq_write_v_r_as_included_vectors} we have
$v_i=0$ if $V_i$ is not maximal under inclusion among $V_1,\ldots,V_{r-1}$.
\end{corollary}
\begin{proof}
Apply Proposition~\ref{pr_maximals_suffice} with $r$ replaced with
$r-1$.
Then \eqref{eq_sum_V_i_equals_sum_of_maximal} shows that
any vector $v_1+\cdots+v_{r-1}$ can be written as another such
sum with $v_i=0$ if $i\notin I_{r-1}$.
\end{proof}
We will make constant use of the above corollary to simplify the
task of verifying that a sequence is quasi-increasing.
}


\subsection{Proof of Theorem~\ref{th_two_increasing_sequences}}

In this subsection we prove Theorem~\ref{th_two_increasing_sequences}.
According to Theorem~\ref{th_quasi_increasing_coordinated}, it suffices
to prove the following stronger theorem.

{\mygreen
\begin{figure}[h]
$$
\begin{array}{|c|c|c|c|}
\hline 
V_{s(t-1)+1}=A_1\cap B_t & V_{s(t-1)+2}=A_2\cap B_t 
& \cdots & V_{st}=A_s\cap B_t \\
\hline 
\vdots & \vdots & \vdots & \vdots \\
\hline 
V_{s+1}=A_1\cap B_2 & V_{s+2}=A_2\cap B_2 & \cdots & V_{2s}=A_s\cap B_2 \\
\hline 
V_1 = A_1\cap B_1 & V_2 = A_2\cap B_1 & \cdots & V_s = A_s \cap B_1 \\
\hline 
\end{array}
$$
\caption{The sequence $V_1,\ldots,V_{st}$}
\label{fi_V_one_to_V_st_based_on_As_and_Bt}
\end{figure}
}

\begin{theorem}\label{th_two_increasing_sequences_as_quasi_increasing}
Let $\cU$ be an $\field$-universe, and let 
$$
A_1\subset\cdots\subset A_s, \quad
B_1\subset\cdots\subset B_t
$$
be two sequences of increasing subspaces of $\cU$.
Let us order the $st$ subspaces of the form $A_i\cap B_j$ as follows:
$$
V_1=A_1\cap B_1,\ldots,V_s=A_s\cap B_1, V_{s+1}=A_1\cap B_2,
\ldots, V_{2s}=A_s\cap B_2,\ldots V_{st}=A_s\cap B_t. 
$$
(i.e., for all $i\in[s]$ and $j\in[t]$, we set
$V_{i+s(j-1)}=A_i\cap B_j$).
Then $V_1,\ldots,V_{st}$
is a quasi-increasing sequence of subspaces.
\end{theorem}
{\mygreen
We depict the sequence $V_1,\ldots,V_{st}$ in 
Figure~\ref{fi_V_one_to_V_st_based_on_As_and_Bt}.
}

\begin{proof}
Let us prove the theorem by induction on $t$.  For $t=1$,
the sequence $V_1,\ldots,V_s$ is increasing, and therefore quasi-increasing.

For the inductive step, say that the theorem holds 
whenever $t\le T$,
and consider the theorem in case $t=T+1$.
We already know that $V_1,\ldots,V_{sT}$ is quasi-increasing.
Let us verify that the condition of being quasi-increasing
continues to hold at $V_r$ (i.e., in position $r$)
with $r=sT+1,sT+2,\ldots,s(T+1)$;
{\mygreen 
hence we need to verify that whenever
\eqref{eq_write_v_r_as_included_vectors} holds, one can also write
\eqref{eq_write_v_r_as_included_vectors_primed} such that
\eqref{eq_write_v_r_as_included_non_included_primed_vanish} holds.}
It is simpler to determine the maximal subsets among
$V_1,\ldots,V_{r-1}$ and to use
Corollary~\ref{co_maximals_suffice}.

{\mygreen
\begin{figure}[h]
$$
\begin{array}{|c|c|c|c|c|}
\hline 
\cdots 
& \boxed{\bf V_{sT+i-1}=A_{i-1}\cap B_{T+1}} & V_r=V_{sT+i}=A_i\cap B_{T+1} 
& & \\
\hline 
\cdots 
& V_{s(T-1)+i-1}=A_{i-1}\cap B_T & V_{s(T-1)+i}=A_i\cap B_T & \cdots & 
\boxed{\bf V_{sT}=A_s\cap B_T} \\
\hline 
\vdots & \vdots & \vdots & \vdots & \vdots \\
\hline 
\cdots & V_{i-1}=A_{i-1}\cap B_1 & \cdots & \cdots &
V_{s}=A_s\cap B_1  \\
\hline 
\end{array}
$$
\caption{The two maximal subsets among $V_1,\ldots,V_{sT+i-1}$ in {\bf bold face}.}
\label{fi_r_equals_s_T_plus_two_or_more_maximal_previous}
\end{figure}
}

For $r=sT+i$, 
with $2\le i\le s$, 
each of $V_1,\ldots,V_{r-1}$ is contained in either $A_{i-1}$ or
$B_T$; hence
each of 
$V_1,\ldots,V_{r-1}$ is contained in at least one of 
(the two maximal subsets)
$A_{i-1}\cap B_{T+1}=V_{sT+i-1}$ and $A_s\cap B_T=V_{sT}$
(see Figure~\ref{fi_r_equals_s_T_plus_two_or_more_maximal_previous},
which indicates these two maximal subsets in {\bf bold face}).
Similarly, for $r=sT+1$, $V_{sT+i-1}=V_{sT}$ is the unique maximal subset.

For any $i\in[s]$, let $r=sT+i$, and let us verify the condition in
Definition~\ref{de_quasi_increasing}.
First consider the case $i\ge 2$ where there are two maximal subspaces.
If $v_r = v_1+\ldots+v_{r-1}$, then also
\begin{equation}\label{eq_v_r_equals_w_one_plus_w_two}
v_r = w_1+w_2
\end{equation} 
with $w_1\in A_{i-1}\cap B_{T+1}$ and $w_2\in A_s\cap B_T$.
But then $w_2=v_r-w_1$, and both $v_r,w_1$ lie in $A_i$; hence
$w_2\in A_i$, and therefore $w_2\in A_i\cap B_T$.
But both $A_{i-1}\cap B_{T+1}$ and $A_1\cap B_T$ are subsets
of $A_i\cap B_{T+1}=V_r$ that occur in the list $V_1,\ldots,V_{r-1}$.
This establishes the condition of
quasi-increasing for these values of $r$.

The remaining case is the case $r=sT+i$ with $i=1$.
In this case the only maximal subspace is
$V_{r-1}= A_s\cap B_T$, and hence 
\eqref{eq_v_r_equals_w_one_plus_w_two} is replaced
with the equation $v_r=w_2$;
hence the same argument as in the previous paragraph works
(with $w_1=0$).
\end{proof}

\subsection{A Proof of Theorem~\ref{th_six_out_of_seven} and a Partial
Generalization}

In this subsection we will prove 
Theorem~\ref{th_six_out_of_seven}.
Again, it will suffice to prove this stronger result.

\begin{theorem}\label{th_six_out_of_seven_quasi_increasing}
Let $A,B,C$ be subspaces of an $\field$-universe, $\cU$.
Then the sequence 
$$
V_1=A\cap B\cap C,\ V_2=A\cap B, \ V_3=A\cap C,\ V_4=B\cap C,\ V_5=A,\ V_6=B
$$
is quasi-increasing.
\end{theorem}
\begin{proof}
We need to verify that the sequence is
quasi-increasing in positions $r=2,3,\ldots,6$.

For $r=2$ we have $V_1\subset V_2$ so the condition holds.
For $r=3$, $V_2\cap V_3=V_1$, so the verification is the same as for 
the dimension formula.

For $r\ge 4$, since $V_1\subset V_2$, we can omit $v_1$ from an
equation \eqref{eq_write_v_r_as_included_vectors}.

For $r=4$, $V_4=B\cap C$, consider
an equation $v_4=v_2+v_3$.  Then since $v_4,v_2\in B$, the equation
$v_3=v_4-v_2\in B$ shows that $v_3\in V_3\cap B=V_1$.
The same argument with $B$ and $C$ exchanged
shows that $v_2\in V_1$.  Hence we may take $v'_1=v_2+v_3\in V_1$
and we have $v_4=v'_1$.

For $r=5$, we consider an equation
$$
v_5=v_2+v_3+v_4
$$
Since $v_2,v_3,v_5\in A$ also $v_4\in A$ and hence $v_4\in A\cap V_4=V_1$.
Furthermore $V_1,V_2,V_3\subset A=V_5$, so the verification is complete.

For $r=6$, since $V_1,V_2,V_3\subset A=V_5$, it suffices to consider
equations 
$v_6=v_4+v_5$.  Since $v_4,v_6\in B$ we have $v_5\in B$ and hence
$v_5\in B\cap A=V_2\subset B=V_6$.
Since $V_2,V_4\subset V_6$, the verification is complete there.
\end{proof}

We remark that the same 
{\mygreen method for showing that $V_1,\ldots,V_4$ are
coordinated in the proof above}
can be used to show that
for any $V_1,\ldots,V_m$, the set of all intersections of any $m-1$
of the $V_1,\ldots,V_m$ is coordinated.  However, for $m\ge 4$,
the set of all $(m-2)$-fold intersections can be discoordinated.
For example, in $\field^3$ let
$$
V_1={\rm Span}(e_1,e_2),
\ V_2={\rm Span}(e_1,e_3),
\ V_3={\rm Span}(e_2,e_1+e_3),
\ V_4={\rm Span}(e_3);
$$
then the $2$-fold intersections include the one dimensional spaces
spanned by $e_1,e_2,e_3,e_1+e_3$, 
{\mygreen which are therefore not coordinated.}
And if $m\ge 5$, we can set $V_i=\cU$ for $i\ge 5$, and therefore,
again, 
the $(m-2)$-fold intersections are not coordinated.

\subsection{A Proof of Theorem~\ref{th_six_out_of_seven_and_D}}

Similar to previous proofs, 
to prove Theorem~\ref{th_six_out_of_seven_and_D}, it
clearly suffices to prove the following stronger theorem.

\begin{theorem}\label{th_six_out_of_seven_and_D_quasi_increasing}
Let $A,B,C,D$ be subspaces of an $\field$-universe, $\cU$.
Then the sequence 
\begin{equation}\label{eq_list_of_quasi_increasing_for_ABCD}
V_0,V_1,V_1',V_2,V_3,V_4,V_5,V_6
\end{equation} 
where
\begin{equation}\label{eq_quasi_increaing_for_ABCD}
\begin{gathered}
V_0=A\cap B\cap C\cap D, \ 
V_1=A\cap B\cap C,
\ V_1'=D,\\
V_2=A\cap B, V_3=A\cap C,\ V_4=B\cap C,\ V_5=A,\ V_6=B
\end{gathered}
\end{equation} 
is quasi-increasing.
\end{theorem}
\begin{proof}
Since $V_0\subset V_1$, 
\eqref{eq_list_of_quasi_increasing_for_ABCD} 
is quasi-increasing at $V_1$.

If $u_1'\in V_1'=D$ equals a sum $u_0+u_1$ with
$u_i\in V_i$, then $u_1=u_1'-u_0\in D$, and hence 
$u_1\in A\cap B\cap C\cap D=V_0$.  Since $V_0\subset V_1'$, this
proves $u_0+u_1$ already lies in $V_0\subset V_1'$;
hence \eqref{eq_list_of_quasi_increasing_for_ABCD}
is quasi-increasing at $V_1'$.

{\mygreen Since $V_2$ contains $V_0,V_1,V_1'$, 
\eqref{eq_list_of_quasi_increasing_for_ABCD} is quasi-increasing
at $V_2$.}

From here we finish the proof as the proof of
Theorem~\ref{th_six_out_of_seven_quasi_increasing}:
since $V_0\subset V_1$ and $V_1'\subset V_2$,
both $V_0$ and $V_1'$ are not maximal elements of the sequence
$V_0,V_1,V_1',V_2,\ldots,V_{r-1}$ for all $r\ge 3$.
Hence
writing any element of $V_r$ with $r\ge 3$ as the sum of
elements of earlier members of the sequence
\eqref{eq_quasi_increaing_for_ABCD} gives
this element as a sum of elements in $V_1,V_2,\ldots,V_{r-1}$;
and hence the verification for $r\ge 3$ in the proof of
Theorem~\ref{th_six_out_of_seven_quasi_increasing} holds
here as well.
\end{proof}

\subsection{A Proof of 
Theorem~\ref{th_m_minus_one_fold_intersections_are_coordinated}}

{\mygreen
Similar to previous proofs, 
to prove Theorem~\ref{th_m_minus_one_fold_intersections_are_coordinated}
it clearly suffices to prove the following stronger theorem.

\begin{theorem}\label{th_m_minus_one_fold_intersections_are_quasi_increasing}
Let $A_1,\ldots,A_m$ be subspaces of an $\field$-universe, $\cU$.
For each $i\in[m]$ let
$$
V_i = \Awithout{i}
\eqdef \bigcap_{j\ne i} A_j 
= A_1\cap\ldots A_{i-1}\cap A_{i+1}\cap\ldots\cap A_m,
$$
and let $V_0 = A_1\cap\ldots\cap A_m$.
Then $V_0,\ldots,V_m$ is quasi-increasing.
\end{theorem}
\begin{proof}
We need to show that $V_0,\ldots,V_m$
is coordinated at $V_j$ for any $j\in[m]$.
So let
$$
v_j = v_0 + \cdots + v_{j-1}
$$
with $v_i\in V_i$.  Then
$$
v_1 = v_j - v_2 - v_3 -\cdots -v_{j-2}\in A_1,
$$
and hence $v_1$ lies in both $A_1$ and $A_2\cap\ldots\cap A_m$,
and therefore $v_1\in V_0$.  By the same argument, $v_2,\ldots,v_{j-1}\in V_0$.
Hence $v_j = v'_0$ with $v'_0\in V_0\subset V_j$.
\end{proof}
}

\subsection{Strongly Quasi-Increasing Sequences}

We remark that the sequences in
Theorems~\ref{th_two_increasing_sequences_as_quasi_increasing},
\ref{th_six_out_of_seven_quasi_increasing},
\ref{th_six_out_of_seven_and_D_quasi_increasing},
and~\ref{th_m_minus_one_fold_intersections_are_quasi_increasing}
satisfy a stronger property than being quasi-increasing, which
we now define.

{\mygreen
\begin{definition}\label{de_strongly_quasi_increasing}
Let $V_1,\ldots,V_m$ be a sequence of distinct
vector spaces in some universe.
For $r\in[m]$, 
we say that this sequence is {\em strongly quasi-increasing in position $r$} 
(or {\em at $V_r$}) if setting
$$
J_r= \bigl\{ i\in[m] \ \bigm| V_r\not\subset V_i\bigr\},
\quad
K_r= \bigl\{ i\in[m] \ \bigm| \mbox{$i\ne r$ and $V_i\subset V_r$}\bigr\},
$$
then
whenever
\begin{equation}\label{eq_write_v_r_as_included_vectors_strongly_q_i}
v_r = \sum_{i\in J_r} v_i 
\quad\mbox{such that}\quad
\forall i\in J_r, \ v_i\in V_i,
\end{equation} 
one also has 
\begin{equation}
{\mygreen \label{eq_write_v_r_as_included_vectors_primed_strongly_q_i}}
v_r = \sum_{i\in K_r} v_i'
\quad\mbox{such that}\quad
\forall i\in K_r, \ v_i'\in V_i.
\end{equation} 
Furthermore, if this condition holds for all $r\in[m]$,
we say that $V_1,\ldots,V_m$ is 
{\em strongly quasi-increasing}.
\end{definition}
Notice that $J_r,K_r$ above are defined independent of the order
of the sequence $V_1,\ldots,V_m$; hence the notion of 
strongly quasi-increasing is independent of the order of the sequence.

We emphasize that in the above definition, the vector spaces
$V_1,\ldots,V_m$ must be {\em distinct};
if not, the same definition would work for our results below, 
but we would need to add to the 
definition of $K_r$ the condition that $V_i\ne V_r$.

We now claim that any strongly quasi-increasing sequence
can be ordered so that it is quasi-increasing.
Noticed that sequence of subspaces $V_1,\ldots,V_m$ is 
partially ordered, and hence has at least one compatible
total order, i.e., we can arrange $V_1,\ldots,V_m$
so that $V_i\subset V_j$ implies $i\le j$
(formally one can prove this by induction on $r$).

\begin{proposition}
Let $V_1,\ldots,V_m$ be a strongly quasi-increasing
sequence of distinct vector spaces in some universe.
Say that $V_1,\ldots,V_m$ are arranged in any non-decreasing order,
i.e., $V_i\subset V_j$ implies $i<j$, or, equivalently,
$V_1<\cdots<V_m$ is a total order compatible with the partial
order of inclusion.
Then the sequence $V_1,\ldots,V_m$ is quasi-increasing.
\end{proposition}
\begin{proof}
For any $i,r\in[m]$ with $i\ne r$, we have
$i\in K_r$ implies $i<r$.
Moreover if $i<r$ then $V_r\not\subset V_i$, so
$i\in J_r$.
Hence whenever
\eqref{eq_write_v_r_as_included_vectors} holds, then also
\eqref{eq_write_v_r_as_included_vectors_primed} holds with
\eqref{eq_write_v_r_as_included_non_included_primed_vanish}.
\end{proof}
}

Let us briefly show that in the quasi-increasing sequence used in
Theorem~\ref{th_two_increasing_sequences_as_quasi_increasing} is
actually strongly quasi-increasing.
We
leave it to the reader to verify the same for 
the sequences in
Theorems~\ref{th_six_out_of_seven_quasi_increasing},
\ref{th_six_out_of_seven_and_D_quasi_increasing},
and~\ref{th_m_minus_one_fold_intersections_are_quasi_increasing}.

So consider
the sets $A_i\cap B_j$ in
Theorem~\ref{th_two_increasing_sequences_as_quasi_increasing}: if
$A_i\cap B_j \not\subset A_{i'}\cap B_{j'}$
(and the $A_1,\ldots,A_m$ are distinct, as well as the
$B_1,\ldots,B_t$),
then either $i'<i$ or $j'< j$; hence
$A_{i'}\cap B_{j'}$ is a subset of either $A_{i-1}\cap B_t$ (and $i\ge 2$)
or a subset of $A_s\cap B_{j-1}$ (and $j\ge 2$).
But if
$$
v_r = w_1 + w_2
$$
with $v_r\in A_i\cap B_j$, and $w_1\in A_{i-1}\cap B_t$ (which does
not exist if $i=1$, so we can just take $w_1=0$) and
$w_2\in A_s\cap B_{j-1}$, then writing $w_2=v_r-w_1$ shows that
$w_2\in A_i$, and hence $w_2\in A_i\cap B_{j-1}$ which lies in
$A_i\cap B_j$; similarly for $w_1$.

\begin{remark}
We don't know if this strong quasi-increasing property is an
accident in the four applications in this section, or holds whenever 
a sequence is quasi-increasing.
\end{remark}

\begin{remark}
The following sequence of two-dimensional
subspaces of $\field^3$,
$$
V_1={\rm Span}(e_1,e_2),
\ V_2 = {\rm Span}(e_1,e_3),
\ V_3 = {\rm Span}(e_2,e_3)
$$
is not quasi-increasing (where $e_i$ denotes the $i$-th standard basis
vector).
Hence a sequence can be coordinated even if it is not quasi-increasing.
However, if we add to $V_1,V_2,V_3$ above all the intersections
of these subspaces, then the resulting set of subspaces is 
(strongly) quasi-increasing.
We do not presently know of a set of vector subspaces
$V_1,\ldots,V_m$ that is coordinated but the set of
all intersections of $V_1,\ldots,V_m$ cannot be ordered into a
quasi-increasing sequence.
\end{remark}

\section{The Discoordination Formula, Minimizers, and Greedy
Algorithms}
\label{se_Joel_discoord_formula}

The point of this section is to prove theorems regarding the structure of
discoordination minimizers and the
``formula'' in
{\mygreen
Theorem~\ref{th_discoordination_formula}
} for the discoordination of a
collection $A_1,\ldots,A_m$ of subspaces of a universe.
As mentioned just after we stated this theorem,
our ``formula'' is stated in terms of certain
subspaces $S_1,\ldots,S_m$ defined in terms of the $A_i$, and
{\mygreen
it is not generally easy to determine the $S_i$ and their dimensions;
}
without a good understanding of the $S_i$ we get only partial
information about the discoordination 
{\mygreen of $A_1,\ldots,A_m$.}
Still, this discoordination formula, and related theorems we prove
in this section will be crucial to later prove
Theorem~\ref{th_main_three_subspaces_decomp}
in a fairly simple fashion.

\subsection{Meet Numbers and Basic Greedy Considerations}

There are a number of properties of discoordination minimizers, $X$, of 
subsets $A_1,\ldots,A_m$ that we now describe.
Notice that since
$$
{\rm DisCoord}_X (A_1,\ldots,A_m) =
\sum_{i=1}^m \Bigl( \dim(A_i)-|X\cap A_i| \Bigr)
$$
$$
=\sum_{i=1}^m \dim(A_i) - \sum_{i=1}^m |X\cap A_i|,
$$
$X\in{\rm Ind}(\cU)$ is a discoordination minimizer iff
$X\in {\rm Ind}(\cU)$ maximizes
$$
\sum_{i=1}^m |X\cap A_i|.
$$

\begin{definition}
Let $A_1,\ldots,A_m$ be subspaces of an $\field$-universe, $\cU$.
For a finite subset $X\subset\cU$ we define
the {\em meet of $X$ (in $A_1,\ldots,A_m$)} to be
$$
{\rm Meet}(X)=
{\rm Meet}(X;A_1,\ldots,A_m)=\sum_{i=1}^m |X\cap A_i|.
$$
If $x\in\cU$, we define the {\em (pointwise)
meeting number of $x$ (in $A_1,\ldots,A_m$)}
to be
$$
{\rm meet}(x)=
{\rm meet}(x;A_1,\ldots,A_m)
={\rm Meet}(\{x\};A_1,\ldots,A_m)
=\{ i\in [m] \ | \ x\in A_i \}.
$$
If $X=\{x_1,\ldots,x_m\}$ is a finite subset of $\cU$, we say that
$x_1,\ldots,x_m$ is {\em (arranged in) decreasing meeting order} if
$$
{\rm meet}(x_1)\ge {\rm meet}(x_2) \ge \cdots\ge {\rm meet}(x_n).
$$
\end{definition}
Usually $A_1,\ldots,A_m$ will be fixed, so we may simply write
${\rm Meet}(X)$ and ${\rm meet}(x)$ without confusion.
Of course ${\rm meet}(x)$ is the same as ${\rm Meet}(\{x\})$,
{\mygreen and}
we distinguish between ``meet'' and ``Meet'' for 
{\mygreen clarify} 
({\mygreen although} confusion is unlikely to occur).

The next two propositions motivate some of the definitions above.
{\mygreen
\begin{proposition}\label{pr_Meet_maximizer_is_DisCoord_minimizer}
Let $A_1,\ldots,A_m$ be subspaces of an $\field$-universe, $\cU$.
For all $X\in{\rm Ind}(\cU)$ we have
\begin{equation}\label{eq_Meet_maximizer_is_DisCoord_minimizer}
{\rm DisCoord}_X (A_1,\ldots,A_m)= 
\sum_{i=1}^m \dim(A_i) - \sum_{i=1}^m |X\cap A_i|
= \sum_{i=1}^m \dim(A_i) - f(X),
\end{equation} 
where 
\begin{equation}\label{eq_that_defines_f_X_as_Meet_with_A_i}
f(X) = 
\sum_{i=1}^m |X\cap A_i| .
\end{equation} 
Furthermore,
\begin{equation}\label{eq_f_X_as_pointwise_meeting_sum}
f(X) 
={\rm Meet}(X;A_1,\ldots,A_m)
= \sum_{x\in X} {\rm meet}(x;A_1,\ldots,A_m).
\end{equation} 
Hence $X\in{\rm Ind}(\cU)$ is a discoordination minimizer of
$A_1,\ldots,A_m$ iff $X$ maximizes $f(X)$ above.
\end{proposition}
\begin{proof}
By definition,
$$
{\rm DisCoord}_X (A_1,\ldots,A_m)=
\sum_{i=1}^m \bigl( \dim(A_i) - |X\cap A_i| \bigr)
=
\left( \sum_{i=1}^m \dim(A_i)  \right) - f(X),
$$
with $f(X)$ as 
in~\eqref{eq_that_defines_f_X_as_Meet_with_A_i}.
Since 
$$
\sum_{i=1}^m |X\cap A_i|
=
{\rm Meet}(X;A_1,\ldots,A_m) ,
$$
which clearly equals
$$
\sum_{x\in X} {\rm meet}(x;A_1,\ldots,A_m),
$$
we have \eqref{eq_f_X_as_pointwise_meeting_sum}.
\end{proof}
}

Here is an important remark about minimizers that is related to
the ``greedy algorithm'' we will discuss in the next subsection.

\begin{proposition}\label{pr_substitute}
Let $A_1,\ldots,A_m$ be subspaces of an $\field$-universe, $\cU$,
and let $X=\{x_1,\ldots,x_n\}$ be a basis of $\cU$
that is
a discoordination minimizer of $A_1,\ldots,A_m$,
such that the $x_i$'s are
arranged in meeting decreasing order, i.e.,
$$
{\rm meet}(x_1)\ge {\rm meet}(x_2) \ge \cdots\ge {\rm meet}(x_n).
$$
Let $X'=\{x'_1,\ldots,x'_k\}$ be any other independent set in $\cU$
arranged in meet decreasing order, i.e., 
$$
{\rm meet}(x'_1)\ge {\rm meet}(x'_2) \ge \cdots\ge {\rm meet}(x'_k).
$$
Then ${\rm meet}(x_k)\ge {\rm meet}(x'_k)$.
\end{proposition}
\begin{proof}
Each $x'_{j'}$ with $j'\in [k]$
may be written uniquely as a linear combination
\begin{equation}\label{eq_write_x_primes_as_xes}
x'_{j'}=\gamma_{j'1}x_1 + \gamma_{j'2}x_2 + \cdots + \gamma_{j' n}x_n
\end{equation} 
where $\gamma_{j'i}\in\field$.
{\mygreen
We claim that for some $j'\in[k]$ and $j\ge k$
we have $\gamma_{j'j}\ne 0$:
}
otherwise $\gamma_{j'j}=0$ for all $j$ with $k\le j\le n$, 
and then
$$
x'_{j'} \in S = {\rm Span}(x_1,\ldots,x_{k-1});
$$
{\mygreen but this impossible,}
since $S$ is of dimension $k-1$, {\mygreen and hence}
$S$ cannot contain
the $k$ linearly independent vectors $x'_1,\ldots,x'_k$.
It follows that for some $j'\in[k]$ and $j\ge k$
we have $\gamma_{j'j}\ne 0$;
fix any such $j',j$.

Since $\gamma_{j'j}\ne 0$
in \eqref{eq_write_x_primes_as_xes},
we may exchange $x'_{j'}$ for $x_j$ in $X$ and get a new
basis $X''$.
Now assume that ${\rm meet}(x_k)< {\rm meet}(x'_k)$, and let us derive
a contradiction: we have
$$
{\rm meet}(x_j) \le {\rm meet}(x_k)< {\rm meet}(x'_k)\le {\rm meet}(x'_{j'}),
$$
and hence 
$$
{\rm Meet}(X'') = {\rm Meet}(X)-{\rm meet}(x_j)+{\rm meet}(x'_{j'})
>{\rm Meet}(X).
$$
This contradicts the fact that $X$ is a discoordination minimizer
of $A_1,\ldots,A_m$.
\end{proof}

The above proposition implies that if $X,X'$ are two minimizers
of $A_1,\ldots,A_m$, both arranged in decreasing meeting order
$x_1,\ldots,x_n$ and $x'_1,\ldots,x'_{n'}$, then
for all $i\le\min(n,n')$ we have
${\rm meet}(x_i)={\rm meet}(x'_i)$.
Theorem~\ref{th_greedy_algorithm} below is a more precise result,
which gives a formula for the number of
$x_i$'s that have a given meeting number
{\mygreen for any minimizer, $X$, of $A_1,\ldots,A_m$.}

There are a few easy but useful 
corollaries of the above proposition 
that we wish to note.

\begin{theorem}\label{th_easy_minimizer_facts}
Let $X\in{\rm Ind}(\cU)$ be a minimizer of subsets $A_1,\ldots,A_m$
of some $\field$-universe, $\cU$, and let
$X'=X\cap(A_1+\cdots+A_m)$.
Then (1)
{\mygreen we have}
$$
|X'|= \dim(A_1+\cdots+A_m);
$$
(2) if $x'\in X'$, then ${\rm meet}(x')\ge 1$; and
(3) if $x\in X\setminus X'$, then ${\rm meet}(x)=0$.
\end{theorem}
\begin{proof}
Since $S=A_1+\cdots+A_m$ is spanned by $A_1\cup\cdots\cup A_m$,
one can write $S$ as the span of $s=\dim(S)$ vectors, $y_1,\ldots,y_s$,
each of which
lies in at least one $A_i$, i.e., ${\rm meet}(y_i)\ge 1$ for $i\in[s]$.
{\mygreen It follows from Proposition~\ref{pr_substitute}
that if the vectors of $X$ are arranged 
in meet decreasing order, $x_1,x_2,\ldots$, then ${\rm meet}(x_s)\ge 1$.
This gives $s$ linearly
independent vectors $x_1,\ldots,x_s$ for which ${\rm meet}(x_i)\ge 1$
for all $i\in[s]$, and therefore all lie in $X'$.
Since $\dim(S)=s$, $X'$ is a basis for $S$,
and the vectors in $X\setminus X'$ must lie outside of $S$.
These facts imply~(1)--(3) above.}
\end{proof}

The above theorem gives a small amount of structure regarding
minimizers.

\begin{definition}
Let $X\in{\rm Ind}(\cU)$ be a minimizer of subsets $A_1,\ldots,A_m$
of some $\field$-universe, $\cU$.  We say that $X$ is a {\em small minimizer}
(with respect to $A_1,\ldots,A_n$) if $X\subset A_1+\ldots+A_m$,
and is a {\em large minimizer} if $X$ is a basis for $\cU$.
\end{definition}

\begin{proposition}
Let $X\in{\rm Ind}(\cU)$ be a minimizer of subsets $A_1,\ldots,A_m$
of some $\field$-universe, $\cU$.  
Then
$X'\subset X\subset X''$ where $X'$ is a small minimizer and $X''$
is a large minimizer.
Furthermore all small minimizers are of size
$\dim(A_1+\cdots+A_m)$.
\end{proposition}
\begin{proof}
We have $X'=X\cap(A_1+\ldots+A_m)$ is a small minimizer and $X'\subset X$.
By the theorem above, $X'$ is of size
$\dim(A_1+\cdots+A_m)$ and spans all of $A_1+\cdots+A_m$.
If $X$ is not a basis for $\cU$ we can extend it to a basis $X''$ of $\cU$.
It follows that all elements of $X''\setminus X'$ lie outside of
$A_1+\cdots+A_m$, and hence each has meet zero with $A_1,\ldots,A_m$.
\end{proof}

\subsection{The Greedy Algorithm for Minimizers}

There is a simple ``greedy algorithm'' to build a minimizer, $X$, of 
subspaces $A_1,\ldots,A_m$ of a universe; the problem is
that this algorithm is stated in terms of certain subspaces derived
from the $A_i$---namely
the $S_i$ and $U_i$ defined below---and so our greedy algorithm provides
only partial information about the (dis)coordination of $A_1,\ldots,A_m$
and the structure of its minimizers.
Nonetheless, aspects of this ``greedy algorithm'' will help us prove the
main theorem regarding three subspaces of a universe.

\begin{definition}\label{de_k_fold_sums}
Let $A_1,\ldots,A_m$ be subspaces of an $\field$-universe, $\cU$.
For any $k$ between $1$ and $m$,
a {\em $k$-fold intersection of the $A_1,\ldots,A_m$} refers to any
subspace of the form 
$$
A_{i_1}\cap \ldots \cap A_{i_k} \quad\mbox{where}\quad
1\le i_1<\ldots<i_k \le m.
$$
For each $k=1,\dots,m$, we use 
{\mygreen $S_k=S_k(A_1,\ldots,A_m)$
and $U_k=U_k(A_1,\ldots,A_m)$}
respectively
to denote, respectively, the sum and union
of all the $k$-fold
intersections of the $A_1,\ldots,A_m$, i.e.,
$$
S_k = \sum_{1\le i_1<\ldots<i_k{\mygreen \le m}} 
A_{i_1}\cap \ldots \cap A_{i_k},
\quad
U_k = \bigcup_{1\le i_1<\ldots<i_k{\mygreen \le m}} 
A_{i_1}\cap \ldots \cap A_{i_k},
$$
i.e.,
$$
{\mygreen S_0 = \cU,\quad}
S_1 = A_1 + \cdots + A_m, \quad 
S_2 = \sum_{j_1<j_2} A_{j_1}\cap A_{j_2},
$$
$$
S_3 = \sum_{j_1<j_2<j_3} A_{j_1}\cap A_{j_2}\cap A_{j_3},
\quad \ldots,\quad 
S_m = A_1\cap\ldots\cap A_m,
{\mygreen \quad S_{m+1}=0;}
$$
and
$$
U_0 = \cU,\quad
U_1 = A_1 \cup \cdots \cup A_m, \quad
U_2 = \bigcup_{j_1<j_2} A_{j_1}\cap A_{j_2},
$$
$$
U_3 = \bigcup_{j_1<j_2<j_3} A_{j_1}\cap A_{j_2}\cap A_{j_3},
\quad \ldots,\quad 
U_m = A_1\cap\ldots\cap A_m,
{\mygreen\quad U_{m+1}=0.}
$$
(The values of $S_0,U_0,S_{m+1},U_{m+1}$ are given as 
above either by convention
or by
a reasonable interpretation of 
an empty intersection, empty sum, and an empty union.)
\end{definition}

We make the following remarks regarding the definitions and notation above.
The $U_i$ defined above are subsets of $\cU$ 
(not generally subspaces!), the $S_i$ are subspaces of $\cU$, and satisfy
\begin{enumerate}
\item 
for all $i=0,\ldots,m$ we have
$U_i\subset S_i={\rm Span}(U_i)$ for all $i$; 
\item
$0 = S_{m+1} \subset S_m \ldots \subset S_1\subset S_0=\cU$;
\item
$\emptyset = U_{m+1} \subset U_m \ldots \subset U_1\subset U_0=\cU$;
\item
for all $y\in\cU$ and $i=0,\ldots,m$ we have
${\rm meet}(y)=i$ iff $y\in U_i\setminus U_{i+1}$; 
\item 
for all $i=0,\ldots,m$ we have
$S_i/S_{i+1}$ is spanned by the images
of the elements of $U_i$ in the quotient space $\cU/S_{i+1}$;
said otherwise, $S_i={\rm Span}(S_{i+1},U_i)$; hence
\item
for all $i=0,\ldots,m$, the image of $U_i$ in $S_i/S_{i+1}$
spans this quotient space; hence there is a subset $Y_i\subset U_i$
whose image in $S_i/S_{i+1}$ is a basis; since no such element
of $Y_i$ can lie in $S_{i+1}$ (i.e., equal $0$ in $S_i/S_{i+1}$), 
we have that any such $Y_i$
consists entirely of elements $y\in\cU$ such that
${\rm Meet}(y)=i$.
\end{enumerate}

The $Y_i$ described above turn out to be essential to our greedy
algorithm, and merit a formal definition.
\begin{definition}\label{de_purely_i_th_intersection_basis}
Let $A_1,\ldots,A_m$ be subspaces of an $\field$-universe, $\cU$,
and let notation be as in Definition~\ref{de_k_fold_sums}.
For any $i=0,\ldots,m$, we say that a set $Y$ is a {\em purely
$i$-th intersection
basis (for $A_1,\ldots,A_m$)}
if
\begin{enumerate}
\item $Y$ is a basis in $S_i$ relative to $S_{i+1}$, and
\item {\mygreen for all $y\in Y$,}
${\rm meet}(y)={\rm meet}(y;A_1,\ldots,A_m)=i$.
\end{enumerate}
{\mygreen Note that (2) can also be replaced with
${\rm meet}(y)\ge i$, since the fact that $Y$ is a basis of $S_i$
relative to $S_{i+1}$ implies that ${\rm meet}(y)\le i$.}
\end{definition}

It is worth making the following easily proven remark.
\begin{proposition}\label{pr_purely_i_th_intersection}
Let $A_1,\ldots,A_m$ be subspaces of an $\field$-universe, $\cU$,
and let notation be as in Definition~\ref{de_k_fold_sums}.
Then for each $i=0,\ldots,m$, there exists a purely
$i$-th intersection basis.
\end{proposition}
\begin{proof}
The proof consists of unwinding the definitions.  Setting
$U_i'=U_i\setminus U_{i+1}$ we have
$$
U_i'=U_i\setminus U_{i+1} = 
\{ y\in\cU \ | \ {\rm Meet}(y;A_1,\ldots,A_m)=i \}.
$$
Since (1) $U_i=U_i'\cup U_{i+1}$, (2) $S_i={\rm Span}(U_i)$, and 
(3) $U_{i+1}\subset S_{i+1}$,
it follows that $S_i$ is spanned by
$S_{i+1}$ and the elements 
{\mygreen of $U_i'$.}
Hence 
(by Proposition~\ref{pr_basis_exchange}, item (2))
there exists a $Y_i$ consisting entirely of elements of $U_i'$
such that $Y_i$ is a basis of $S_i$ relative to $S_{i+1}$.
\end{proof}

Let us describe in rough terms
our ``greedy algorithm''
to construct a discoordination minimizer, $X$,
of subspaces $A_1,\ldots,A_m$ of a universe.
Our approach is to equivalently choose an $X\in{\rm Ind}(\cU)$ that
maximizes $f(X)$ in
\eqref{eq_f_X_as_pointwise_meeting_sum}.
Since ${\rm meet}(x)$ takes values between $0$ and
$m$, 
our ``greedy algorithm'' first 
chooses the largest possible subset $X_m\in{\rm Ind}(\cU)$ consisting
of elements in $x$ with ${\rm meet}(x)=m$;
hence $X_m$ can be as large as $\dim(S_m)$, and such a $X_m$ is a basis
for $S_m$.
The second step is 
to choose the largest subset $X_{m-1}$ possible consisting of $x$ with
${\rm meet}(x)=m-1$ and such that $X_m\cup X_{m-1}$ remains
linearly independent;
it is not hard to see that (see below) that the
largest possible $X_{m-1}$ 
is of size $\dim(S_{m-1}/S_m)$ and must be a purely $(m-1)$-th intersection
basis.
The $i$-th step, for $i=3,\ldots,m$ is that given
$X_m,X_{m-1},\ldots,X_{m-i+2}$, we choose $X_{m-i+1}$ to consist
of $x\in\cU$ with ${\rm meet}(x)=m-i+1$ and as large as possible with
$X_m \cup X_{m-1}\cup\cdots\cup X_{m-i+1}$ linearly independent;
by induction we easily see that $X_{m-i+1}$ must be
a pure $(m-i+1)$-th intersection basis.

Theorem~\ref{th_greedy_algorithm}
below proves that the above ``greedy algorithm'' always produces a minimizer,
and each minimizer is constructed as such.
A novel point is that each $X_i$ 
is an arbitrary purely $i$-th intersection basis,
and hence the choice of $X_i$ is independent of the 
choice of $X_m,\ldots,X_{i+1}$ and $X_{i-1},\ldots,X_1$.
A consequence of this fact is that we get a simple formula for the
discoordination in terms of the $A_i$'s and $S_i$'s.
Let us state and prove this result formally (we state this theorem
in a way that makes each subsequent claim easy to prove,
although the overall statement is a bit long).

{\mygreen
\begin{theorem}\label{th_greedy_algorithm}
Let $A_1,\ldots,A_m$ be subspaces of an $\field$-universe, $\cU$,
and let notation be as in Definition~\ref{de_k_fold_sums}.
Let $X\in{\rm Ind}(\cU)$, and for $i=0,\ldots,m$ set
\begin{equation}\label{eq_X_i_as_a_function_of_X}
X_i = \{ x\in X \ | \ {\rm meet}(x;A_1,\ldots,A_m)=i\}.
\end{equation} 
Then $X_0,\ldots,X_m$ are pairwise disjoint, and 
\begin{equation}\label{eq_obvious_meet_Meet_equalities}
{\rm Meet}(X;A_1,\ldots,A_m)
=
\sum_{i=1}^m i|X_i|
=
\sum_{i=1}^m \bigl( |X_i| + \cdots + |X_m| \bigr);
\end{equation} 
\begin{equation}\label{eq_intermediate_X_i_sum_formula}
\forall i\in[m],
\quad 
|X_i|+\cdots+|X_m| \le \dim(S_i);
\end{equation}
and
\begin{equation}\label{eq_Meet_X_versus_S_i_sum}
{\rm Meet}(X;A_1,\ldots,A_m)\le \sum_{i=1}^m \dim(S_i).
\end{equation} 
Furthermore, the following are equivalent:
\begin{enumerate}
\item
equality holds in \eqref{eq_Meet_X_versus_S_i_sum};
\item
equality holds in \eqref{eq_intermediate_X_i_sum_formula} for all
$i\in[m]$;
\item
for all $i\in[m]$, $X_i\cup\ldots\cup X_m$ is a basis of $S_i$
(hence $X_i$ is a basis of $S_i$ relative to $S_{i+1}$
and $|X_i|=\dim(S_i/S_{i+1})$);
and
\item
we have
\begin{enumerate}
\item
for each $i\in [m]$, $X_i$ is a purely $i$-intersection
basis, i.e., $X_i$ is a basis of $S_i$ relative to $S_{i+1}$
and all elements of $X_i$ lies in exactly $i$ of $A_1,\ldots,A_m$;
and
\item
$X_0\subset\cU$ is any set whose image in $\cU/S_1=S_0/S_1$ is
a set of linearly independent vectors.
\end{enumerate} 
\end{enumerate} 
Hence for any $X\in{\rm Ind}(\cU)$ we have
\begin{equation}\label{eq_discoordination_lower_bound_for_all_ind_X}
{\rm DisCoord}_X(A_1,\ldots,A_m) \ge 
\sum_{i=1}^m \dim(A_i)
-
\sum_{i=1}^m \dim(S_i),
\end{equation} 
with equality holding for any
$X_0,\ldots,X_m$ that
satisfy~(4a) and~(4b) above, and hence
\begin{equation}\label{eq_discoordination_formula_via_A_i_S_i}
{\rm DisCoord}(A_1,\ldots,A_m)
=
\sum_{i=1}^m \dim(A_i)
-
\sum_{i=1}^m \dim(S_i).
\end{equation} 
In addition, for such a discoordination minimizer $X$ we have
\begin{equation}\label{eq_sum_S_i_convenient}
\sum_{i=1}^m \dim(S_i) 
= \sum_{i=1}^m i\,|X_i| 
= \sum_{i=1}^m i \dim(S_i/S_{i+1}),
\end{equation} 
and hence we may also write
\begin{equation} 
\label{eq_discoordination_formula_sometimes_convenient}
{\rm DisCoord}(A_1,\ldots,A_m) =
\sum_{i=1}^m \dim(A_i) - 
\sum_{i=1}^m i\,\dim(S_i/S_{i+1})
\end{equation} 
\end{theorem}
\begin{proof}
The $X_i$ are pairwise disjoint in view of \eqref{eq_X_i_as_a_function_of_X}.
The first equality in \eqref{eq_obvious_meet_Meet_equalities} follows
since
$$
{\rm Meet}(X;A_1,\ldots,A_m)
= \sum_{x\in X} {\rm meet}(x;A_1,\ldots,A_m),
$$
and the second equality is clear.
Since $X_i,\ldots,X_m$ are pairwise disjoint and their union is a
linearly independent set in $S_i$,
\eqref{eq_intermediate_X_i_sum_formula} follows.
Summing \eqref{eq_intermediate_X_i_sum_formula} over all $i$ we get
$$
\sum_{i=1}^m \bigl( |X_i| + \cdots + |X_m| \bigr)=
\sum_{i=1}^m \dim(S_i),
$$
which combined with \eqref{eq_obvious_meet_Meet_equalities}
yields
\eqref{eq_Meet_X_versus_S_i_sum}.
Moreover, the Inequality Summation Principle implies
that condition (1) of the theorem holds
iff~(2) holds.

(2) $\implies$ (3): 
$X_0,\ldots,X_m$ are pairwise
disjoint; for each $i\in[m]$,
$X_i\cup\ldots\cup X_m$ lie in $S_i$, and hence if 
\eqref{eq_intermediate_X_i_sum_formula} holds with equality,
then $X_i\cup\ldots\cup X_m$ are a basis for $S_i$.

Clearly (3) $\implies$ (2).

(3) $\implies$ (4): for any $i\in[m]$, $X_{i+1}\cup\ldots\cup X_m$
is a basis for $S_{i+1}$; since 
$X_i\cup\ldots\cup X_m$ is a basis for $S_i$, it follows
that $X_i$ is a basis for $S_i$ relative to
$S_{i+1}$
(see the paragraph after Definition~\ref{de_relative_basis}).
Given~(4a),
(4b)~follows since $X_1\cup\ldots\cup X_m$ is a basis for $S_1$.

(4) $\implies$ (3):
we easily show this by descending induction 
for $i=m,m-1,\ldots,1$.

In view of 
Proposition~\ref{pr_purely_i_th_intersection},
$X_1,\ldots,X_m$ satisfying~(4a) exist, and 
hence~\eqref{eq_Meet_X_versus_S_i_sum} is attained with equality
for any $X$ that is the union of such $X_1,\ldots,X_m$.
Hence, by Proposition~\ref{pr_Meet_maximizer_is_DisCoord_minimizer},
\eqref{eq_discoordination_lower_bound_for_all_ind_X}
and
\eqref{eq_discoordination_formula_via_A_i_S_i} holds.
In this case 
equality holds in \eqref{eq_Meet_X_versus_S_i_sum}, and by
\eqref{eq_obvious_meet_Meet_equalities} we have
the first equality in 
\eqref{eq_sum_S_i_convenient};
the second equality there holds since 
$X_i$ is a basis of $S_i$ relative to $S_{i+1}$.
Combining 
\eqref{eq_sum_S_i_convenient}
and
\eqref{eq_discoordination_formula_via_A_i_S_i}
yields
\eqref{eq_discoordination_formula_sometimes_convenient}.
\end{proof}
}

\subsection{An Equivalent Discoordination Formula and Interpretation
of the Greedy Algorithm}

In this subsection we use the greedy algorithm
to give another interpretation of discoordination.

\begin{theorem}\label{th_equivalent_discoordination_using_greedy}
Let $A_1,\ldots,A_m$ be subspaces of an $\field$-universe, $\cU$,
and let $S_i=S_i(A_1,\ldots,A_m)$
be as in Definition~\ref{de_k_fold_sums}.
Let $X$ be any discoordination minimizer of
$A_1,\ldots,A_m$.
\begin{enumerate}
\item
We have
\begin{equation}\label{eq_discoordination_equals_sum_d_i}
{\rm DisCoord}(A_1,\ldots,A_m)  = d_1 + \cdots + d_m, 
\end{equation} 
where
\begin{equation}\label{eq_definition_of_d_i}
d_j=d_j(A_1,\ldots,A_m) \eqdef
\left(
\sum_{i=1}^m 
\dim^{\cU/S_{j+1}}([A_i\cap S_j]_{S_{j+1}}) 
\right)
- j\dim(S_j/S_{j+1}).
\end{equation} 
\item
{\mygreen
We have
$$
d_j = 
\left(
\sum_{i=1}^m 
\dim^{\cU/S_{j+1}}([A_i\cap S_j]_{S_{j+1}}) 
\right)
- j |X_j|
$$
where $X_j=\{x\in X \ | \ {\rm meet}(x;A_1,\ldots,A_m)=j\}$.
\item
For all $j\in[m]$, 
$d_j\ge 0$.
\item
For all $j\in[m]$, 
$d_j=0$ iff $[X_j]_{S_{j+1}}$ coordinates 
$[A_1\cap S_j]_{S_{j+1}},\ldots,[A_m\cap S_j]_{S_{j+1}}$ in $\cU/S_{j+1}$.
\item
$d_j=0$ 
iff $[A_1\cap S_j]_{S_{j+1}},\ldots,[A_m\cap S_j]_{S_{j+1}}$ are coordinated
in $\cU/S_{j+1}$.
\item
We have $d_m=0$.
}
\end{enumerate}
\end{theorem}
{\mygreen 
\begin{remark}
Hence $A_1,\ldots,A_m$ are coordinated iff $d_1=\ldots=d_m=0$.
Beyond the fact that $d_m=0$, with the help of
Theorem~\ref{th_k_fold_intersections_coordinated}
we will see that $d_{m-1}=d_m=0$ 
(in Corollary~\ref{co_d_m_minus_one_vanishes}).
Theorem~\ref{th_k_fold_intersections_coordinated} addresses
a more general phenomenon.
\end{remark}
}
\begin{proof}
(1): For arbitrary
$$
0=S_{m+1}\subset S_m \subset \cdots \subset S_1 \subset S_0=\cU,
$$
and any subspace $B\subset \cU$, we have
$$
\dim^{\cU/S_{j+1}}([B\cap S_j]_{S_{j+1}}) = \dim(B\cap S_j)-\dim(B\cap S_{j+1})
,
$$
{\mygreen which upon summing over all $j$}
allows us to write
$$
\dim(B) = 
\sum_{j={\mygreen 0}}^m
\dim^{\cU/S_{j+1}}([B\cap S_j]_{S_{j+1}}).
$$
If $B\subset S_1=A_1+\cdots+A_n$, then $B/S_1=0$, and hence
the $j=0$ term above vanishes; hence
\begin{equation}\label{eq_increasing_sequence_dim_sum}
B\subset S_1 \quad\implies\quad
\dim(B) = 
\sum_{j={\mygreen 1}}^m
\dim^{\cU/S_{j+1}}([B\cap S_j]_{S_{j+1}}).
\end{equation} 

{\mygreen By~\eqref{eq_discoordination_formula_sometimes_convenient}}
we have
$$
{\rm DisCoord}(A_1,\ldots,A_m) = 
\sum_{i=1}^m \dim(A_i) - 
\sum_{j=1}^m j\dim(S_j/S_{j+1}).
$$
By \eqref{eq_increasing_sequence_dim_sum} we have
$$
\dim(A_i)=\sum_{j=1}^m 
\dim^{\cU/S_{j+1}}([A_i\cap S_j]_{S_{j+1}}),
$$
and hence
\eqref{eq_discoordination_equals_sum_d_i} follows.

{\mygreen (2): follows from~(1) and the fact that 
$|X_j|=\dim(S_j/S_{j+1})$ 
for any minimizer $X$
(by~(3) of Theorem~\ref{th_greedy_algorithm}).

(3): in $\cU/S_{j+1}$, the image of $X_j$ lies in exactly
$j$ of $[A_1\cap S_j]_{S_{j+1}},\ldots,[A_m\cap S_j]_{S_{j+1}}$,
since $X_j$ is a purely $j$-th intersection basis.  Since $X_j$ are
linearly independent in $\cU/S_{j+1}$, we have
\begin{equation}\label{eq_d_j_are_non_negative}
\sum_{i=1}^m \dim^{\cU/S_{j+1}}\bigl( [A_i\cap S_j]_{S_{j+1}} \bigr)  \ge
\sum_{i=1}^m \bigl|[A_i\cap S_j]_{S_{j+1}}\cap[X_j]_{S_{j+1}}\bigr|^{\cU/S_{j+1}}
=
j|X_j| .
\end{equation} 

(4): $d_j=0$ iff
equality holds in \eqref{eq_d_j_are_non_negative} iff for
each $i\in[m]$ we have
$$
\dim^{\cU/S_{j+1}}\bigl( [A_i\cap S_j]_{S_{j+1}} \bigr)  =
\bigl|[A_i\cap S_j]_{S_{j+1}}\cap[X_j]_{S_{j+1}}\bigr|^{\cU/S_{j+1}} .
$$
Hence $d_j=0$ iff $[X_j]_{S_{j+1}}$ coordinates
$[A_1\cap S_j]_{S_{j+1}},\ldots, [A_m\cap S_j]_{S_{j+1}}$ 
in $\cU/S_{j+1}$.

(5): ``if'' is implied by~(4), so it suffices to prove ``only if.''
So if $[A_1\cap S_j]_{S_{j+1}},\ldots, [A_m\cap S_j]_{S_{j+1}}$
are coordinated in $\cU/S_{j+1}$, then they are coordinated by
some basis $X'$ in $S_j/S_{j+1}$.
If $\tilde X_j\subset S_j$ are any representatives of
$X'$ (i.e., $\tilde X_j$ is obtained by choosing some element of
each $S_{j+1}$-coset in of $X'$),
then $\tilde X_j$ is a purely $j$-th intersection basis 
for $A_1,\ldots,A_m$.  Hence, by
Theorem~\ref{th_greedy_algorithm}, (4a),
in any minimizer $X$, we may replace $X_j\subset X$ with
$\tilde X_j$ and get another minimizer.  But then
$$
\dim^{\cU/S_{j+1}}\bigl( [A_i\cap S_j]_{S_{j+1}} \bigr)  =
\bigl|[A_i\cap S_j]_{S_{j+1}}\cap[\tilde X_j]_{S_{j+1}}\bigr|^{\cU/S_{j+1}} 
$$
for all $i$, and hence (by~(2) above), $d_j=0$.

(6): $S_{m+1}=0$, and $S_m=A_1\cap\ldots\cap A_m$; hence
$A_i\cap S_m=S_m$, so any basis of $S_m$ coordinates
$[A_1\cap S_j]_{S_{j+1}},\ldots, [A_m\cap S_j]_{S_{j+1}}$ for $j=m$.}
\end{proof}

\subsection{Decomposing Discoordination into ``$j$-Fold Intersection''
Parts}
\label{su_decomposing_discoordination}
\ 
{\mygreen Theorem~\ref{th_m_minus_one_fold_intersections_are_coordinated}}
shows
that for any subspaces
$A_1,A_2,A_3$ of an $\field$-universe, $\cU$,
the 2-fold intersections
$$
A_1\cap A_2, \ A_1\cap A_3, \ A_2\cap A_3
$$
are coordinated.  The theorems 
{\mygreen in this subsection} will prove 
a few facts
{\mygreen that are important in proving 
Theorem~\ref{th_main_three_subspaces_decomp} and when we study coded-caching,
such as}
\begin{enumerate}
\item
we have
\begin{equation}\label{eq_discoordination_passing_to_S_two_for_three_subspaces}
{\rm DisCoord}^\cU(A_1,A_2,A_3)
=
{\rm DisCoord}^{\cU/S_2}([A_1]_{S_2},[A_2]_{S_2},[A_3]_{S_2}),
\end{equation} 
\item
the images of
$A_1\cap A_2, \ A_1\cap A_3, \ A_2\cap A_3$
in $\cU/S_3$ (with $S_3=A_1\cap A_2\cap A_3$) are linearly independent
\item $A_1,A_2,A_3$ are coordinated iff
$$
[A_1]_{S_2},
\ [A_2]_{S_2}, \ [A_3]_{S_2}
$$
are linearly independent (in $\cU/S_2$).
\end{enumerate}
{\mygreen In this section we prove a number of stronger results
that imply~(1)--(3): in particular,
Theorems~\ref{th_k_fold_intersections_coordinated} studies the
situation in
Theorem~\ref{th_equivalent_discoordination_using_greedy}
where $d_m=\ldots=d_k=0$ for some $k$, and we will use this
theorem to prove that this
always holds with $k=m-1$ (we already know this holds with $k=m$
from part~(6) of
Theorem~\ref{th_equivalent_discoordination_using_greedy});
this therefore implies~(2) and~(3) above.
Theorem~\ref{th_discoordination_quotient_by_S_k} gives a general 
inequality
of the form
\begin{equation}\label{eq_discoordination_passing_to_quotient_summary}
{\rm DisCoord}^{\cU/S_k}\bigl([A_1]_{S_k},\ldots,[A_m]_{S_k} \bigr)
\le
{\rm DisCoord}^\cU\bigl(A_1,\ldots,A_m\bigr) ,
\end{equation} 
and gives one set of conditions for the above to hold with equality.
We will want to know that equality holds in the above when $m=3$
and $k=2$; one can prove this using~(4) and~(5)
of Theorem~\ref{th_main_three_subspaces_decomp}.
However, 
after proving 
Theorem~\ref{th_discoordination_quotient_by_S_k}, we will show that
(1) for any $m\ge 3$ and $1\le k\le m-2$, strict equality can hold 
in \eqref{eq_discoordination_passing_to_quotient_summary}, and
(2) for any $m\ge 3$ and $k=m-1,m$, equality always holds.
}

\begin{theorem}\label{th_k_fold_intersections_coordinated}
Let $A_1,\ldots,A_m$ be subspaces of an $\field$-universe, $\cU$.
{\mygreen
For each $I\subset[m]$, let
$$
A_I = \bigcap_{i\in I} A_i.
$$
Let $S_i=S_i(A_1,\ldots,A_m)$ 
be as in Definition~\ref{de_k_fold_sums}.
Say that for some $k\in[m]$, 
the set $\{A_I\}_{|I|=k}$ is coordinated;
let $Z\in{\rm Ind}(\cU)$ coordinate all $A_I$ with $|I|=k$.
}
For $j=0,\ldots,m$ let 
$$
Z_j = \{ z\in Z \ | \ {\rm meet}(z;A_1,\ldots,A_m)=j \}
$$
and
$$
Z_{\ge j} = Z_j \cup Z_{j+1} \cup \cdots\cup Z_m 
= \{ z\in Z \ | \ {\rm meet}(z;A_1,\ldots,A_m)\ge j \} .
$$
Then the following statements hold.
\begin{enumerate}
\item 
For any $I\subset [m]$ with $|I|=j\ge k$, $A_I$ is coordinated by $Z$.
\item
For any $I\subset [m]$ with $|I|=j\ge k$, $A_I$ is coordinated by $Z_{\ge j}$,
{\mygreen i.e.,}
$$
A_I\cap Z_{\ge j} 
\quad \mbox{is a basis of} \quad A_I.
$$
\item 
For any $j\ge k$, $S_j$ is coordinated by $Z_{\ge j}$, and $Z_{\ge j}$
is a basis for $S_j$.
\item
For any $j\ge k$, 
$Z_j$ is a basis for 
$S_j$ relative to $S_{j+1}$.
\item
For each $I\subset [m]$ with $|I|=j\ge k$, in
$\cU/S_{j+1}$, the set $[A_I\cap Z_j]_{S_{j+1}}$ 
is a basis for $[A_I]_{S_{j+1}}$ 
of size $|A_I\cap Z_j|$.
\item 
For any $j\ge k$, each element of $Z_j$ is in a unique element of
$A_I$ such that $I\subset [m]$ satisfies $|I|=j$, and so
$Z_j$ is partitioned into subsets $\{A_I\cap Z_j\}_I$
with $I$ ranging over all $I\subset [m]$ with $|I|=j$.
\item
For any $j\ge k$, the images of
$\{ A_I \}_{I\subset[m],\ |I|=j}$ in $\cU/S_{j+1}$, i.e.,
the subspaces
$$
\{ [A_I]_{S_{j+1}} \}_{I\subset[m],\ |I|=j} ,
$$
are linearly independent subspaces of $S_j/S_{j+1}$.
\item
With $d_i$ as in \eqref{eq_definition_of_d_i}, we have
$d_m=\ldots=d_k=0$.
\item
If $X$ is {\em any} minimizer of $A_1,\ldots,A_m$, and
$X_j$ consists of those $x\in X$ with ${\rm meet}(x;A_1,\ldots,A_m)=j$,
then for each $j\ge k$, $[X_j]_{S_{j+1}}$ coordinates
$[A_i\cap S_j]_{S_{j+1}}$ in $\cU/S_{j+1}$ for all $i\in[m]$.
\item
If $X$ is {\em any} minimizer of $A_1,\ldots,A_m$, and
$X_j$ consists of those $x\in X$ with ${\rm meet}(x;A_1,\ldots,A_m)=j$,
then for each $j\ge k$ and $|I|=j$ we have 
$|X_j\cap A_I|=\dim(A_I/S_{j+1})$.
\end{enumerate}
\end{theorem}
\begin{remark}
We do not presently know, regarding~(9) and~(10)
above, if any minimizer of
$A_1,\ldots,A_m$ necessarily coordinates each $A_I$ with $|I|\ge k$.\footnote{
In other words, say that for some $k$, $\{A_I\}_{|I|=k}$ are 
coordinated; then the discoordination
of $\{A_1,\ldots,A_m\}$
equals that of $\{A_1,\ldots,A_m\}\cup\{A_I\}_{|I|=k}$
since we may take $Z_j$ with $j\ge k$ as the purely $j$-th
intersection basis of our minimizer; however, 
does the set of minimizers decrease?
}
\end{remark}
\begin{proof}
Most of the implications easily result from the previous ones, 
often making use of
Proposition~\ref{pr_X_coordination_closed_cap_sum}; let
us give some details.

{\mygreen
(1): We prove (1) by induction on $j=k,k+1,\ldots,m$.
The case $j=k$ holds by assumption.
For the inductive step, assume that $Z$ coordinates all $A_I$
with $|I|=j$ for some $j$ with $k\le j\le m-1$,
and let $I\subset[m]$ with $|I|=j+1$.
Then $j+1\ge k+1\ge 2$, hence we can choose
distinct elements
}
$i_1,i_2$ of $I$ and set
$I_1=I\setminus \{ i_1\}$, $I_2=I\setminus\{i_2\}$.  
{\mygreen 
Since $|I_1|=|I_2|=j$, $Z$ coordinates $A_{I_1},A_{I_2}$.
Since
}
$$
A_{I_1}\cap A_{I_2} = \bigcap_{i\in I_1\cup I_2} A_i = A_I,
$$
Proposition~\ref{pr_X_coordination_closed_cap_sum} implies that
$Z$ coordinates $A_I$.

(2): If $|I|=j\ge k$, any element $z\in Z\cap A_I$
meets all $A_i$ with $i\in I$, and hence ${\rm meet}(z)\ge j$.
Hence $Z\cap A_I=Z_{\ge j}\cap A_I$.  By (1), $Z$ coordinates $A_I$,
so $Z\cap A_I=Z_{\ge j}\cap A_I$ is a basis for $A_I$.

(3): $S_j$ is the span of all $A_I$ with $|I|=j$.  Since $Z_{\ge j}$ 
coordinates each such $A_I$,
Proposition~\ref{pr_X_coordination_closed_cap_sum} implies that $Z_{\ge j}$
coordinates $S_j$ and that $S_j\cap Z_{\ge j}$ is a basis for $S_j$.
However, each element of $Z_{\ge j}$ meets $j$ of the $A_1,\ldots,A_m$, and
hence each element of $Z_{\ge j}$ lies in some $A_I$ with $|I|=j$, and
hence also lies in $S_j$.  
{\mygreen
Hence $Z_{\ge j}\subset S_j$,
and therefore $Z_{\ge j}\cap S_j=Z_{\ge j}$.
}
Hence $S_j$ has a basis consisting
of $Z_{\ge j}\cap S_j=Z_{\ge j}$.

(4): By (3), we have $Z_{\ge j+1},Z_{\ge j}$ are respective bases for
$S_{j+1},S_j$.  It follows
(by the discussion below
Definition~\ref{de_relative_basis}) that the set
$Z_{\ge j}\setminus Z_{\ge j+1}$ is a basis for $S_j$ relative to
$S_{j+1}$.  But $Z_{\ge j}\setminus Z_{\ge j+1}$ equals $Z_j$.

(5): According to (4), in $\cU/S_{j+1}$, the vectors in the set
$[A_I\cap Z_j]_{S_{j+1}}$ are linearly independent
and 
{\mygreen is a set}
of size $|A_I\cap Z_j|$.
{\mygreen
By~(2) above,
}
$$
A_I = {\rm Span}(Z_{\ge j}\cap A_I),
$$
{\mygreen and hence}
we have 
$$
[A_I]_{S_{j+1}}=A_I+S_{j+1}={\rm Span}(Z')
$$
where
$$
Z' = (Z_{\ge j}\cap A_I)\cup Z_{\ge j+1} = (Z_j\cap A_I)\cup Z_{\ge j+1}.
$$
It follows that in $\cU/S_{j+1}$, $[Z_j\cap A_I]_{S_{j+1}}$ spans
the image of $Z_j\cap A_I$ there.  Hence, in $\cU/S_{j+1}$,
$[Z_j\cap A_I]_{S_{j+1}}$ are linearly independent and span
$[A_I]_{S_{j+1}}$, and hence are a basis
{\mygreen for the span of 
$[A_I]_{S_{j+1}}$ in $\cU/S_{j+1}$, i.e.,
for $[S_j]_{S_{j+1}}$ in $\cU/S_{j+1}$;
i.e., $Z_j$ is a basis for $S_j$ relative to $S_{j+1}$.}

(6): is immediate from the definition of $Z_j$ as those $z\in Z$
with ${\rm meet}(z)=j$.

(7): We have
$$
|Z_j | = \sum_{|I|=j} |Z_j\cap A_I|,
$$
and so in $\cU/S_{j+1}$ we have
$$
\dim^{\cU/S_{j+1}}([S_j]_{S_{j+1}}) = 
\sum_{|I|=j} 
\dim^{\cU/S_{j+1}}([A_I]_{S_{j+1}}) .
$$
Since the $A_I$ with $|I|=j$ span all of $S_j$,
the $[A_I]_{S_{j+1}}$ span all of $S_j/S_{j+1}$ in $\cU/S_{j+1}$.
Hence by \eqref{eq_linear_independence_dim_sum_criterion}
(in Definition~\ref{de_linearly_independent_subspaces}),
these subspaces are linearly independent
{\mygreen in $\cU/S_{j+1}$, and hence in $S_j/S_{j+1}$.}

{\mygreen
(8): from~(7) above, for any $j\ge k$,
the $A_I$ with $|I|=j$ are linearly independent
in $\cU/S_{j+1}$, and therefore coordinated, and
from~(5) of Theorem~\ref{th_equivalent_discoordination_using_greedy}
we have $d_j=0$.

(9): by~(8) above, $d_m=\cdots=d_k=0$, and hence for all
$j\ge k$ we have that $[X_j]_{S_{j+1}}$ coordinates
$[A_i\cap S_j]_{S_{j+1}}$ in $\cU/S_{j+1}$.

(10): for any minimizer, $X$, of $A_1,\ldots,A_m$, and any $j\in[m]$,
$X_j$ is a purely $j$-th intersection basis of $A_1,\ldots,A_m$.
Hence for any $I$ with $|I|=j$,
$|X_j\cap A_I|\le \dim^{\cU/S_{j+1}}([A_I]_{S_{j+1}})$.
From~(7) above we know that for all $j\ge k$,
the $[A_I]_{S_{j+1}}$ are linearly
independent in $\cU/S_{j+1}$, and hence, summing over all $I$ with
$|I|=j$ we have 
$$
|X_j| =
\sum_{|I|=j} |X_j\cap A_I| 
\le \sum_{|I|=j} \dim^{\cU/S_{j+1}}([A_I]_{S_{j+1}})
= \dim(S_j/S_{j+1}).
$$
But since $|X_j|=\dim(S_j/S_{j+1})$, the Inequality Summation Principle
implies that 
$|X_j\cap A_I|\le \dim^{\cU/S_{j+1}}([A_I]_{S_{j+1}})$
must hold with equality for all $I$.
}
\end{proof}

{\mygreen
\begin{corollary}\label{co_d_m_minus_one_vanishes}
Let $A_1,\ldots,A_m$ be subspaces of an $\field$-universe, $\cU$,
with $m\ge 2$,
and let $d_i$ be as in 
\eqref{eq_definition_of_d_i}.  Then $d_m=d_{m-1}=0$.
\end{corollary}
\begin{proof}
According to 
Theorem~\ref{th_m_minus_one_fold_intersections_are_coordinated},
$A_I$ are coordinated for all $I\subset[m]$ with $|I|=m-1$.
So apply Theorem~\ref{th_k_fold_intersections_coordinated}
with $k=m-1$.
\end{proof}
}

\begin{theorem}\label{th_discoordination_quotient_by_S_k}
Let $A_1,\ldots,A_m$ be subspaces of an $\field$-universe, $\cU$,
and let 
{\mygreen $S_i=S_i(A_1,\ldots,A_m)$
be as in Definition~\ref{de_k_fold_sums}.}
Then for any $k\in[m]$ we have
\begin{equation}\label{eq_discoordination_under_S_k_quotient_inequality}
{\rm DisCoord}^{\cU/S_k}\bigl([A_1]_{S_k},\ldots,[A_m]_{S_k} \bigr)
\le
{\rm DisCoord}^\cU\bigl(A_1,\ldots,A_m\bigr) .
\end{equation} 
Furthermore, for any minimizer, $X$ (or really any subset of $\cU$), and
$0\le j\le m$ let
$$
X_j = \{ x\in X \ | \ {\rm meet}(x;A_1,\ldots,A_m)=j \}
$$
(as usual, and)
$$
X_{\ge k}=\bigcup_{j\ge k} X_j, \quad
X_{< k}=\bigcup_{j< k} X_j .
$$
Then 
\eqref{eq_discoordination_under_S_k_quotient_inequality}
holds with equality if
for some minimizer, $X$, 
the following conditions hold:
\begin{enumerate}
\item
$X_{\ge k}$ coordinates $A_i\cap S_k$ for all $i$, 
\item
for all $i\in[m]$,
$|A_i\cap X_{<k}|$ equals the size of the number of $S_k$-cosets
in $[A_i]_{S_k}\cap [X_{<k}]_{S_k}$,
{\mygreen
(our proof below shows that}
the first quantity is always
bounded above by the second, 
{\mygreen 
but our proof doesn't address
when equality holds),
}
and
\item
$X'=[X_{<k}]_{S_k}$ is a minimizer for
$[A_1]_{S_k},\ldots,[A_m]_{S_k}$.
\end{enumerate}
{\mygreen
Moreover, if (1)--(3) hold for some minimizer, $X$,
then (1)--(3)
}
hold for all minimizers, $X$,
{\mygreen of $A_1,\ldots,A_m$.}
\end{theorem}

\begin{proof}
Let $X$ be a minimizer of $A_1,\ldots,A_m$.  
{\mygreen
Then~(3)
of Theorem~\ref{th_greedy_algorithm}
}
implies that
$X_{\ge k}$ is a basis for $S_k$, and hence the map
from $X_{<k}$ to its image, $X'$, in $\cU/S_k$ is a bijection,
and the $X'$ are linearly independent in $\cU/S_k$.
Hence for any subspace 
${\mygreen B}\subset \cU$ we have
\begin{equation}
\label{eq_first_inequality_with_B_for_sum_inequality_below}
\bigl|B\cap X_{\ge k} \bigr| \le \dim(B\cap S_k);
\end{equation} 
since $X'\in{\rm Ind}(\cU/S_k)$ is a linearly independent set
{\mygreen in bijection with $X_{<k}$,}
\begin{equation} 
\label{eq_second_inequality_with_B_for_sum_inequality_below}
|B\cap X_{<k}|  \le 
\bigl| [B]_{S_k}\cap X' \bigr|^{\cU/S_k},
\end{equation} 
where the 
{\mygreen right-hand-side}
counts the 
{\mygreen number of $S_k$-cosets in}
$[B]_{S_k}\cap X'$ 
{\mygreen (note that strict inequality can hold, namely
when $B+S_k$ contains an element of $X'$ that doesn't lie in $B$).}
{\mygreen
Adding~\eqref{eq_first_inequality_with_B_for_sum_inequality_below}
and~\eqref{eq_second_inequality_with_B_for_sum_inequality_below}}
we get
\begin{equation}\label{eq_inequality_with_B_and_X_ge_k_and_X_lt_k}
|B\cap X| = |B\cap X_{\ge k}| + |B\cap X_{<k}|
\le \dim(B\cap S_k) + 
\bigl| [B]_{S_k}\cap X' \bigr|^{\cU/S_k} .
\end{equation} 
{\mygreen
Since $X$ is a minimizer for $A_1,\ldots,A_m$, we have}
$$
{\rm DisCoord}^\cU(A_1,\ldots,A_m)
=
\sum_{i=1}^m \bigl( \dim(A_i)-|A_i\cap X| \bigr)
$$
{\mygreen
which, in view of~\eqref{eq_inequality_with_B_and_X_ge_k_and_X_lt_k}
summed over all $B=A_i$,}
\begin{equation}\label{eq_first_discoord_inequality_quotient_by_S_k}
\ge 
\sum_{i=1}^m \Bigl( 
\dim(A_i)-\dim(A_i\cap S_k)-
\bigl| [A_i]_{S_k}\cap X' \bigr|^{\cU/S_k}
\Bigr)
\end{equation} 
$$
= \sum_{i=1}^m \Bigl( 
\dim^{\cU/S_k} \bigl( [A_i]_{S_k} \bigr)-
\bigl| [A_i]_{S_k}\cap X' \bigr|^{\cU/S_k} 
\Bigr)
={\rm DisCoord}_{X'}^{\cU/S_k}\bigl([A_1]_{S_k},\ldots,[A_m]_{S_k} \bigr)
$$
\begin{equation}\label{eq_second_discoord_inequality_quotient_by_S_k}
\ge 
{\rm DisCoord}^{\cU/S_k}\bigl([A_1]_{S_k},\ldots,[A_m]_{S_k} \bigr),
\end{equation} 
which implies
\eqref{eq_discoordination_under_S_k_quotient_inequality}.

{\mygreen
Note that
\eqref{eq_second_discoord_inequality_quotient_by_S_k} holds
with equality iff
condition~(3) of the theorem holds.  Note also that
\eqref{eq_first_discoord_inequality_quotient_by_S_k} is
equivalent to 
\eqref{eq_inequality_with_B_and_X_ge_k_and_X_lt_k} for $B=A_i$ for all $i$,
which is equivalent to both
\eqref{eq_first_inequality_with_B_for_sum_inequality_below} and
\eqref{eq_second_inequality_with_B_for_sum_inequality_below}
for $B=A_i$ for all $i$,
which are equivalent to~(1) and~(2).
}
Hence if (1)--(3) hold for some minimizer, $X$, of $A_1,\ldots,A_m$,
then
$$
{\rm DisCoord}^\cU(A_1,\ldots,A_m)
=
{\rm DisCoord}^{\cU/S_k}\bigl([A_1]_{S_k},\ldots,[A_m]_{S_k} \bigr),
$$
and then (1)--(3) must hold for any other minimizer $\tilde X$,
for otherwise
replacing $\tilde X$ by $X$ in the above, strict inequality
would hold for at least one of 
\eqref{eq_first_discoord_inequality_quotient_by_S_k} or
\eqref{eq_second_discoord_inequality_quotient_by_S_k},
and hence strict inequality would hold in
\eqref{eq_discoordination_under_S_k_quotient_inequality},
which is impossible.
\end{proof}

{\mygreen

Because conditions~(2) and~(3) of the above theorem look less direct
to verify than condition~(1), we make the following observation.

\begin{proposition}
In Theorem~\ref{th_discoordination_quotient_by_S_k}, 
conditions~(2) and~(3) hold
provided that for all $j\in[m]$ and $i_1<\cdots<i_j$ we have
\begin{equation}\label{eq_S_k_factors_through_intersection}
[A_{i_1}]_{S_k} \cap \ldots \cap [A_{i_j}]_{S_k}
=
[A_{i_1}\cap\ldots\cap A_{i_j}]_{S_k}.
\end{equation}
\end{proposition}
Of course, the right-hand-side of
\eqref{eq_S_k_factors_through_intersection} is always a 
subset of the left-hand-side.
\begin{proof}
To verify condition~(2), say that
$[A_1]_{S_k}=[x]_{S_k}$ for some $x\in X_j$ with $j<k$; we 
need to show that $x\in A_1$; if not, then since $X_j$ is a
pure $j$-th intersection basis of $A_1,\ldots,A_m$, we have
$x\in A_{i_1}\cap\ldots\cap A_{i_j}$ for unique
$1\le i_1<\cdots<i_j\le m$; the fact that $x\notin A_1$ implies
that $i_1>1$.  But then
$$
x \in [A_1]_{S_k}\cap [A_{i_1}]_{S_k} \cap \ldots \cap [A_{i_j}]_{S_k}
= [A_1\cap A_{i_1}\cap\ldots\cap A_{i_j}]_{S_k},
$$
and hence $x\in [S_{j+1}]_{S_k}=S_{j+1}$; 
this is impossible, since $x\in X_j$ and $X_j$
is a basis for $S_j$ relative to $S_{j+1}$.

To verify condition~(3), it suffices to show that for each $j<k$,
$X_j$ is a pure $j$-th intersection basis for 
$[A_1]_{S_k},\ldots,[A_m]_{S_k}$.
So fix $j<k$ and let
$$
\tilde S_j \eqdef S_j\bigl( [A_1]_{S_k},\ldots,[A_m]_{S_k} \bigr) .
$$
Summing \eqref{eq_S_k_factors_through_intersection} over all
possible $i_1<\ldots<i_j$ we have
$$
\tilde S_j =
\Bigl[ S_j(A_1,\ldots,A_m) \Bigr]_{S_k},
$$
and hence for $j<k$ we have that 
\begin{equation}\label{eq_tilde_S_j_and_S_j_quotients}
\tilde S_j/\tilde S_{j+1} = [S_j/S_{j+1}]_{S_k} = S_j/S_{j+1}
\end{equation} 
since $S_k\subset S_{j+1}$.  
Since the image of $X_j$ in
$S_j/S_{j+1}$ is a basis,
\eqref{eq_tilde_S_j_and_S_j_quotients} implies that the image of $X_j$
in $\tilde S_j/\tilde S_{j+1}$ is a basis.
Moreover, since each $x\in X_j$ lies in $j$ of
$A_1,\ldots,A_m$, we have $[x]_{S_k}$ lies in at least
$j$ of $[A_1]_{S_k},\ldots,[A_m]_{S_k}$.
\end{proof}

\begin{remark}
In \eqref{eq_discoordination_under_S_k_quotient_inequality}
we have equality for $m=3$ and $k=2$, thanks to~(4) and~(5) of 
Theorem~\ref{th_main_three_subspaces_decomp};
For any $m=3$ and $k=1$, the left-hand-side of 
\eqref{eq_discoordination_under_S_k_quotient_inequality},
and hence strict inequality can occur in this case.
\end{remark}

\begin{remark}
Similarly, for $m\ge 4$ and any $1\le k\le m-2$,
we note that strict inequality
can occur in \eqref{eq_discoordination_under_S_k_quotient_inequality}:
indeed, if $\cU=\field^2=A_4,\ldots,A_m$, 
and $A_1,A_2,A_3$ are distinct one-dimensional
spaces of $\cU=\field^2$, then $S_{m-2}=\cU$, and hence
$S_k=\cU$.  Hence $\cU/S_k=0$, and while $A_1,\ldots,A_m$ are
not coordinated, the left-hand-side of 
\eqref{eq_discoordination_under_S_k_quotient_inequality} is $0$.
\end{remark}

We will now show that for $k=m-1,m$, in contrast to the last remark,
equality always
holds in \eqref{eq_discoordination_under_S_k_quotient_inequality}.
To do so we need a subtle lemma.

\begin{lemma}\label{le_intersection_A_i_with_S_m_minus_one}
Let $A_1,\ldots,A_m$ be subspaces of an $\field$-universe, $\cU$,
and let 
$S_i=S_i(A_1,\ldots,A_m)$
be as in Definition~\ref{de_k_fold_sums}.
For $k=m,m-1$ we have that for all $j\in[m]$,
\eqref{eq_S_k_factors_through_intersection} holds.
\end{lemma}
\begin{proof}
By symmetry it suffices to show that for any $j\in[m]$
$$
[A_1]_{S_k}\cap\ldots\cap [A_j]_{S_k}=
[A_1\cap\ldots\cap A_j]_{S_k}.
$$
To do so it suffices to show that
\begin{equation}\label{eq_subtle_s_k_inclusion}
[A_1]_{S_k}\cap\ldots\cap [A_j]_{S_k} \subset
[A_1\cap\ldots\cap A_j]_{S_k},
\end{equation} 
since the reverse inclusion is immediate.
For $k=m$,
\eqref{eq_subtle_s_k_inclusion} is immediate, since if
\begin{equation}\label{eq_subtle_a_i_S_k_equality_condition}
[a_1]_{S_k}=\ldots = [a_j]_{S_k}
\end{equation} 
for some $a_i\in A_i$ for all $i\in[j]$, i.e.
\begin{equation}\label{eq_subtle_a_i_S_k_equality_condition_explicit}
b = a_1+s_1 = \cdots = a_j+s_j
\end{equation} 
for some $s_i\in S_m$ for all $i\in [j]$,
then $b=a_i+s_i\subset A_i$ for all $i\in[j]$,
and hence $b\in A_1\cap\ldots\cap A_j$.
Hence \eqref{eq_subtle_s_k_inclusion} holds.

Next say that $k=m-1$ and \eqref{eq_subtle_a_i_S_k_equality_condition}
holds, i.e.,
\eqref{eq_subtle_a_i_S_k_equality_condition_explicit} with
$a_i\in A_i$ and $s_i\in S_{m-1}$.  Then for all $i\in[j]$ we have
$$
b = a_i + s_{i,1} + \cdots + s_{i,m}
$$
where $s_{i,\ell}\in \Awithout{\ell}$ with
$\Awithout{\ell}$ as in
\eqref{eq_A_without_notation}.  
Then $b'=b-s_{1,1}-\cdots-s_{j,j}$
lies in each $A_i$ with $i\in[j]$, since $b'$ 
equals $a_i$ plus
a sum of terms in $\Awithout{\ell}$ with $\ell\ne i$.
Hence $b'\in A_1\cap\ldots\cap A_m$, and 
$[b']_{S_{m-1}}=[b]_{S_{m-1}}=[a_i]_{S_{m-1}}$ for all $i\in [j]$.
Hence \eqref{eq_subtle_s_k_inclusion} holds.
\end{proof}


\begin{theorem}
\label{th_A_one_to_m_discoord_same_as_modulo_S_m_minus_one}
In Theorem~\ref{th_discoordination_quotient_by_S_k},
\eqref{eq_discoordination_under_S_k_quotient_inequality}
holds with equality for $k=m$ and $k=m-1$.
\end{theorem}

\begin{proof}
It suffices to verify conditions~(1)--(3) of
Theorem~\ref{th_discoordination_quotient_by_S_k} holds with
equality for any maximizer, $X$, of $A_1,\ldots,A_m$.
Conditions~(2) and~(3) follow
from Lemma~\ref{le_intersection_A_i_with_S_m_minus_one}.

For $k=m$, Condition~(1) follows since $A_i\cap S_m=S_m$
(since $S_m\subset A_i$), which $X_m$ coordinates since
$X_m$ is a pure intersection basis.

For $k=m-1$, let $\Awithout{i}$ be as in
\eqref{eq_A_without_notation}.
Since $\Awithout{1},\ldots,\Awithout{m}$ are
coordinated 
(by Theorem~\ref{th_m_minus_one_fold_intersections_are_coordinated}), 
we have (by~(4) and~(5)
of Theorem~\ref{th_equivalent_discoordination_using_greedy})
that $X_{\ge m-1}$ coordinates each.
Let us show that $A_1\cap S_{m-1}=\Awithout{2}+\cdots+\Awithout{m}$:
indeed, an element of $A_1\cap S_{m-1}$ can be written as
$a_1=s_1+\cdots+s_m$ with $s_i\in\Awithout{i}$,
and then $s_1=a_1-s_2-\cdots-s_m$ shows that $s_1\in A_1$;
hence $s_1\in S_m\subset \Awithout{2}$.
Hence  $A_1\cap S_{m-1}\subset \Awithout{2}+\cdots+\Awithout{m}$,
and the reverse inclusion is clear.
Hence $X_{\ge m-1}$ coordinates $A_1\cap S_{m-1}$, and
by symmetry also $A_i\cap S_{m-1}$ for any $i\in[m]$.
Hence condition~(1) holds for $k=m-1$ as well.
\end{proof}
}


\section{Proof of the Main Theorems Regarding Three Subspaces}
\label{se_Joel_three_subspaces_main}

The goal of this section is to prove 
Theorems~\ref{th_main_three_subspaces_decomp}
and~\ref{th_quotient_via_subspace_in_two}.
Theorem~\ref{th_main_three_subspaces_decomp}.

\subsection{Theorem~\ref{th_main_three_subspaces_decomp} in
the Case When $S_2=0$}

As a first step to proving Theorem~\ref{th_main_three_subspaces_decomp},
we address the case when $S_2(A,B,C)=0$
with $S_2$ as in Theorem~\ref{th_greedy_algorithm}
and \eqref{eq_discoordination_formula_via_A_i_S_i}.

\begin{theorem}\label{th_main_special_case_S_two_zero}
Let $A,B,C$ be any subspaces of an $\field$-universe $\cU$ such that
$A\cap B=A\cap C=B\cap C=0$.
Let $m={\rm DisCoord}(A,B,C)$.  
Then:
\begin{enumerate}
\item
$\dim(A+B)=\dim(A)+\dim(B)$;
\item 
$m=\dim( (A+B)\cap C)$; and
\item 
there are bases 
$a_1,\ldots,a_{m_1}$ of $A$,
$b_1,\ldots,b_{m_2}$ of $B$, and
$c_1,\ldots,c_{m_3}$ of $C$,
such that $m\le m_i$ for $i=1,2,3$,
$$
c_i=a_i+b_i \quad\mbox{for}\quad i\in [m],
$$
and 
\begin{equation}\label{eq_three_subspace_no_two_fold_basis}
a_1,\ldots,a_{m_1}, \ b_1,\ldots,b_{m_2},\ c_{m+1},\ldots,c_{m_3}
\end{equation} 
is a basis for $A+B+C$.
\end{enumerate}
\end{theorem}
\begin{proof}
By the dimension formula, since $A\cap B=0$, we have
$\dim(A+B)=\dim(A)+\dim(B)$.  

Let $S_1,S_2,S_3$ be as in
Theorem~\ref{th_greedy_algorithm}.
The hypothesis of this theorem implies that 
{\mygreen in the formula for
discoordination}
\eqref{eq_discoordination_formula_via_A_i_S_i},
$S_2=S_3=0$; since $S_1=A+B+C$ we have
$$
m = \dim(A)+\dim(B)+\dim(C) - \dim(A+B+C),
$$
and since $\dim(A)+\dim(B)=\dim(A+B)$,
$$
m = \dim(A+B)+\dim(C) - \dim(A+B+C).
$$
Combining the dimension formula applied to $A+B$ and $C$ then implies that
$m=\dim( (A+B)\cap C )$.

Let $c_1,\ldots,c_m$ be a basis for $(A+B)\cap C$; since each $c_i$
also lies in $A+B$, we may write each $c_i$ as $a_i+b_i$.
We claim that $a_1,\ldots,a_m$ are linearly independent, for if
not then for some $\gamma_1,\ldots,\gamma_m\in\field$ we have
$$
\gamma_1 a_1 + \cdots + \gamma_m a_m = 0
$$
{\mygreen
where $\gamma_i\ne 0$ for at least one $i$;
hence
$$
\gamma_1 c_1 + \cdots + \gamma_m c_m = 
\gamma_1 b_1 + \cdots + \gamma_m b_m; 
$$
}but this is impossible, since the left-hand-side is a nonzero element
of $C$, and the right-hand-side is an element of $B$, which
would imply that $C\cap B$ contains a nonzero element, contrary
to the hypothesis in the theorem.

Similarly the $b_1,\ldots,b_m$ are linearly independent.
By basis extension, we may extend these vectors to a basis,
$b_1,\ldots,b_{m_2}$ of $B$ with $m_2\ge m$.
Similarly we extend the $a_1,\ldots,a_m$ to get a basis
$a_1,\ldots,a_{m_1}$ of $A$, with $m_1\ge m$.
Since $(A+B)\cap C$ is a subspace of dimension $m$ in $C$,
with a basis $c_1,\ldots,c_m$, we may extend this to get a
basis $c_1,\ldots,c_{m_3}$ of $C$ with $m\ge m_3$.
{\mygreen It follows that
$A+B+C$ is spanned the vectors in
\eqref{eq_three_subspace_no_two_fold_basis} (since $c_i=a_i+b_i$);
let us verify that these vectors are linearly independent
(this can be done in a number of ways):}
notice that \eqref{eq_three_subspace_no_two_fold_basis} has
$m'=m_1+m_2+(m_3-m)$ vectors;
by construction $m_1,m_2$ are the dimensions of $A,B$ and
$$
\dim(C) = \dim\bigl( C\cap (A+B) \bigr) + m_3-m;
$$
by the dimension formula,
$$
\dim(A+B+C) =
\dim(A+B)+\dim(C)-\dim\bigl( C\cap (A+B) \bigr)
$$
$$
=\dim(A)+\dim(B)+\dim(C) - \dim\bigl( (A+B)\cap C \bigr)
= m_1+m_2+m_3-m = m' ;
$$
since the collection of $m'$ vectors in 
\eqref{eq_three_subspace_no_two_fold_basis}
span $A+B+C$, and $\dim(A+B+C)=m'$, these vectors
must be a basis for $A+B+C$.
\end{proof}

\subsection{The Lifting Lemma}

Before we prove Theorem~\ref{th_six_out_of_seven}, it is helpful
to extract a simple ingredient of the proof that is conceptually
important.

\begin{lemma}[The Lifting Lemma]
\label{le_lifting_lemma}
Let $A,B,C$ be subspaces of an $\field$-universe, $\cU$,
and let
$$
S_2=S_2(A,B,C) = A\cap B + A\cap C + B\cap C.
$$
If for some $\tilde a\in A$, $\tilde b\in B$, and $\tilde c\in C$ we 
have
$$
[\tilde a+\tilde b]_{S_2} = [\tilde c]_{S_2},
$$
then there exist $a\in A$, $b\in B$, and $c\in C$ such that
$$
a+b=c
$$
and
\begin{equation}\label{eq_tilde_class_same}
[a]_{S_2} = [\tilde a]_{S_2},
\ [b]_{S_2} = [\tilde b]_{S_2},
\ [c]_{S_2} = [\tilde c]_{S_2}.
\end{equation} 
In particular we have
$$
[A+B]_{S_2} \cap [C]_{S_2} = [(A+B)\cap C]_{S_2}.
$$
\end{lemma} 
\begin{proof}
Let us start with the first claim.
We have
$$
[\tilde a+\tilde b-\tilde c]_{S_2} = [0]_{S_2}
$$
and therefore
$$
\tilde a+\tilde b-\tilde c = v_1 + v_2 + v_3
$$
for some $v_1\in A\cap B$, $v_2\in A\cap C$, and $v_3\in B\cap C$.
Then $v_2,v_3\in C$ so $c=\tilde c + v_2 + v_3\in C$ as well.
Similarly $v_1\in A$ so $a = \tilde a -v_1\in A$ as well.
Taking $b=\tilde b$ we then have $a+b=c$.
Since each $v_i$ lies in $S_2$, we have
\eqref{eq_tilde_class_same}.

To prove the second statement, it is immediate that
$$
[(A+B)\cap C]_{S_2} \subset
[A+B]_{S_2} \cap [C]_{S_2} ;
$$
to prove the reverse inclusion we note that an element of the 
right-hand-side of the above equation 
is a class $[\tilde c]_{S_2}$ with $\tilde c\in C$ which is also
a class of the form $[\tilde a+\tilde b]_{S_2}$;
{\mygreen by the previous paragraph,}
there are $a\in A$, $b\in B$, and $c\in C$
with $c=a+b$ and
{\mygreen that satisfy \eqref{eq_tilde_class_same}};
hence $c\in C\cap (A+B)$ with
$[c]_{S_2}=[\tilde c]_{S_2}$.
{\mygreen Hence 
$$
[C]_{S_2} \subset [(A+B)\cap C]_{S_2},
$$
and so the two sides are equal.}
\end{proof}

\subsection{Proof of Theorem~\ref{th_main_three_subspaces_decomp}}
\label{su_proof_of_th_main_three_subspaces_decomp}

\begin{proof}[Proof of Theorem~\ref{th_main_three_subspaces_decomp}]
According to Theorem~\ref{th_six_out_of_seven}, the subspaces
$$
A\cap B\cap C,\ A\cap B,\ A\cap C,\ C\cap B
$$
are coordinated; so let $X$ be a minimum sized set that coordinates
these three subspaces; hence
$$
|X| = \dim(S_2) = \dim( A\cap B+A\cap C + B\cap C ).
$$
Consider in $\cU/S_2$ the vector subspaces
$A' = [A]_{S_2}$, $B'=[B]_{S_2}$, $C'=[C]_{S_2}$;
apply Theorem~\ref{th_main_special_case_S_two_zero} two these three
subspaces (whose two-fold intersections clearly vanish)
$$
a'_1,\ldots,a'_{m_1}, \ b'_1,\ldots,b'_{m_2},\ c'_1,\ldots,c'_{m_3}
$$
be the respective bases for $A',B',C'$ with 
$c'_i=a'_i+b'_i$ for $i\in[m]$; according to
Theorem~\ref{th_main_special_case_S_two_zero},
\begin{equation}\label{eq_m_is_discoord_modulo_S_two}
m = 
\dim^{\cU/S_2}
\Bigl( \bigl( [A]_{S_2}+[B]_{S_2}\bigr) \cap [C]_{S_2} \Bigr)
\end{equation} 
and according to the lifting lemma
\begin{equation}\label{eq_lifting_lemma_modulo_S_two}
\dim^{\cU/S_2}
\Bigl( \bigl( [A]_{S_2}+[B]_{S_2}\bigr) \cap [C]_{S_2} \Bigr)
=
\dim^{\cU/S_2}
\Bigl( \bigl[ (A+B)\cap C \bigr]_{S_2} \Bigr) .
\end{equation} 

Each $a'_i$ is an $S_2$-coset, so 
for each $i\in[m_1]$ pick an arbitrary $\tilde a_i\in\cU$ with
$[\tilde a_i]_{S_2}=a'_i$, and similarly for $\tilde b_i$ for all
$i\in[m_2]$ and for $\tilde c_i$ with $i\in[m_3]$.
By the lifting lemma, for each $i\in[m]$  there exist
$a_i,b_i,c_i$ whose $S_2$-coset is the same as 
$\tilde a_i,\tilde b_i,\tilde c_i$ respectively, and satisfy
$a_i+b_i=c_i$.

For $i>m$, let $a_i=\tilde a_i$, and similarly for $b_i$ and $c_i$.
Setting 
$$
X' = \bigl\{ a_1,\ldots,a_{m_1}, \ b_1,\ldots,b_{m_2},
\ c_{m+1},\ldots,c_{m_3} \bigr\}
$$
we see that $X\cup X'$ is a basis for 
{\mygreen $A+B+C$},
since $X'$ is a basis
of 
{\mygreen $A+B+C$}
relative to $S_2$ and $X$ is a basis of $S_2$.
{\mygreen 
Let $Y$ be an arbitrary basis of $\cU$ relative to $A+B+C$;
hence $X\cup Y\cup X'$ is a basis for $\cU$.}
Set
$$
X_2 = \{ a_1,\ldots,a_m,b_1,\ldots,b_m\}
$$
and 
\begin{align*}
X_1 & =
(X\cup Y \cup X')\setminus X_2
\\
& =
{\mygreen X\cup Y
\{a_{m+1},\ldots,a_{m_1}\}\cup
\{b_{m+1},\ldots,b_{m_2}\}\cup
\{c_{m+1},\ldots,c_{m_3} \} ,}
\end{align*}
and set $\cU_1={\Span}(X_1)$
and $\cU_2={\rm Span}(X_2)$.
Then $X_1,X_2$ are disjoint sets whose union is a basis of $\cU$,
and hence $\cU_1,\cU_2$ form a decomposition of $\cU$.

{\mygreen
Let us prove that (1) $A$ factors through the decomposition of
$\cU$ into $\cU_1$ and $\cU_2$, and that
(2) $X_1$ coordinates $A\cap \cU_1$: to prove both, it suffices to show that
\begin{equation}\label{eq_suffices_to_show_A_factors_Us}
\dim(A) \le |A\cap X_1| + \dim(A\cap \cU_2),
\end{equation} 
for if so then
$$
\dim(A)\le \dim(A\cap \cU_1)+\dim(A\cap U_2) 
\le |A\cap X_1| + \dim(A\cap U_2)
$$
shows that both 
$$
\dim(A)=\dim(A\cap \cU_1) + \dim(A\cap U_2)
\quad\mbox{and}\quad
\dim(A\cap \cU_1) = |A\cap X_1|
$$
must hold.
To prove \eqref{eq_suffices_to_show_A_factors_Us},
let us first prove that
\begin{equation}\label{eq_A_cap_S_two}
A\cap S_2=A\cap B+A\cap C: 
\end{equation}
if} $v_1+v_2+v_3=a\in A$ with $v_1,v_2,v_3$ in, respectively
$A\cap B,A\cap C,B\cap C$, then $v_3=a-v_1-v_2\in A$, and
hence $v_3\in A\cap B\cap C$ so $a=(v_1+v_3)+v_2$ expresses
$a$ as a sum of elements of $A\cap B$ and $A\cap C$.
{\mygreen Hence \eqref{eq_A_cap_S_two} holds; since $X$ coordinates
$A\cap B$ and $A\cap C$ it also coordinates their sum,
i.e., $A\cap S_2$.
Hence
$$
\dim(A) = \dim(A\cap S_2)+\dim(A/S_2) 
= |A\cap X| + \dim(A/S_2) ;
$$
but since $a_1,\ldots,a_{m_1}$ is a basis of $A$ relative to $S_2$,
we have $\dim(A/S_2)=m_1$, and hence
$$
\dim(A) 
= |A\cap X| + m_1
= |A\cap X| + m + (m_1-m)
\le |A\cap X_1| + \dim(A\cap\cU_2)
$$
since 
$a_{m+1},\ldots,a_{m_1}\in X_1$, $X_1\cap X=\emptyset$, and
$a_1,\ldots,a_m\in \cU_2$ are linearly independent.
This proves
\eqref{eq_suffices_to_show_A_factors_Us}, and hence
$A$ factors through $\cU_1$ and $\cU_2$ and $X_1$ coordinates $A\cap \cU_1$.

The same argument with $A,B,C$ permuted shows that
\begin{align*}
\dim(B) & \le |B\cap X_1| + \dim(B\cap \cU_2), \\
\dim(C) & \le |C\cap X_1| + \dim(C\cap \cU_2)
\end{align*}
(the only difference between $C$ and $A$ is that $c_1,\ldots,c_m$ do not lie
in $X_2$, but they do lie in $\cU_2$).
Hence $A,B,C$ factor through the decomposition $\cU_1,\cU_2$,
and $A\cap\cU_1,B\cap\cU_1,C\cap\cU_1$ are coordinated
(by $X_1$).  
This establishes claim~(1) in 
Theorem~\ref{th_main_three_subspaces_decomp} and the statement before
it.
}

Next we have $\dim(\cU_2)=2m$ since it has $X_2$ for a basis;
if $\mu\from \cU_2\to\field^2\otimes\field^m$
is the isomorphism taking $a_i$ to $e_1\otimes e_i$
and $b_i$ to $e_2\otimes e_i$, then
$\mu$ takes $c_i$ to $(e_1+e_2)\otimes e_i$.  Hence $\mu$ satisfies
the required condition of claim~(2)
of Theorem~\ref{th_main_three_subspaces_decomp}.

Now we verify the second part of 
Theorem~\ref{th_main_three_subspaces_decomp},
i.e., that the quantities in (1)--(5) there are equal.
According to 
\eqref{eq_m_is_discoord_modulo_S_two} and
\eqref{eq_lifting_lemma_modulo_S_two},
$$
m = 
\dim^{\cU/S_2}
\Bigl( \bigl( [A]_{S_2}+[B]_{S_2} \bigr) \cap [C]_{S_2} \Bigr)
=
\dim^{\cU/S_2}
\Bigl( \bigl[ (A+B)\cap C \bigr]_{S_2} \Bigr) .
$$
Hence (1), (4), and (5) are equal.

Let us show that $m$ equals the discoordination of $A,B,C$.
Since $A,B,C$ factor through $\cU_1,\cU_2$, 
Theorem~\ref{th_discoordination_factors} implies that
\begin{equation}\label{eq_discoord_A_B_C_factors}
{\rm DisCoord}^\cU(A,B,C)
=
\sum_{i=1}^2
{\rm DisCoord}^{\cU_i}\bigl( A\cap \cU_i,B\cap\cU_i,C\cap\cU_i \bigr).
\end{equation} 
The $i=1$ discoordination term above is zero, 
since $A\cap \cU_1,B\cap\cU_1,C\cap\cU_1$
are coordinated; now we prove that the $i=2$ discoordination term
equals $m$.
Since $\mu$ gives an isomorphism from 
$A\cap \cU_2,B\cap\cU_2,C\cap\cU_2$ to
$$
E_1={\rm Span}(e_1)\otimes\field^m,
\ E_2={\rm Span}(e_2)\otimes\field^m,
\ E_3={\rm Span}(e_1+e_2)\otimes\field^m,
$$
the $i=2$ discoordination term equals the discoordination of
$E_1,E_2,E_3$ above.  Visibly the intersection of any two of
$E_1,E_2,E_3$ is zero, and hence
Theorem~\ref{th_main_special_case_S_two_zero} 
{\mygreen (with $A,B,C$ there replaced with $E_1,E_2,E_3$)}
implies that
$$
{\rm DisCoord}(E_1,E_2,E_3)
=
\dim\bigl(  (E_1+E_2)\cap E_3 \bigr)
=
\dim(E_3)=m.
$$
Hence the right-hand-side of \eqref{eq_discoord_A_B_C_factors} equals $m$.
Hence (1), 
{\mygreen (2),} (4), (5) of the second part of
Theorem~\ref{th_main_three_subspaces_decomp} are equal.

To see that (1), 
{\mygreen (2)} (4), (5) in
Theorem~\ref{th_main_three_subspaces_decomp} also equals the quantity
{\mygreen in (3),}
note that the basis $X_1\cup X_2$
{\mygreen of $\cU$}
contains all the $a_i,b_i,c_i$ except for $c_1,\ldots,c_m$,
and that $X_1\cup X_2$ coordinates $A$ and $B$ and satisfies
$$
{\rm DisCoord}_{X_1\cup X_2}(A,B,C) = m
= {\rm DisCoord}(A,B,C)
$$
{\mygreen by the equality of~(1) and~(2).}
Hence $X_1\cup X_2$ is a discoordination minimizer, and no other
$Y\in{\rm Ind}(\cU)$ can coordinate $A$ and $B$ and have
$\dim(C)-|C\cap Y|$ be any smaller than $m$
(for otherwise ${\rm DisCoord}(A,B,C)\le \dim(C)-|C\cap Y|\le m-1$).
Hence the quantity in~(2) is minimized by the independent set
$X_1\cup X_2$ and equals $m$.
\end{proof}

\subsection{Proof of Theorem~\ref{th_quotient_via_subspace_in_two}}
\label{su_discoordination_from_cU_to_cU_mod_D_two}

\begin{lemma}\label{le_coordinating_with_D}
Let $A,B,C\subset\cU$ be coordinated subspaces of an $\field$-universe,
$\cU$, and let $D\subset A\cap B$. 
Then $A,B,C,D$ are coordinated, and hence the images of
$A,B,C$ in $\cU/D$ are coordinated.
\end{lemma}
\begin{proof}
By Theorem~\ref{th_six_out_of_seven_and_D},
$$
A\cap B\cap C\cap D=C\cap D, \ A\cap B\cap C,\ D,
\ A\cap B,\ A\cap C,\ B\cap C,\ A,\ B
$$
are coordinated by some $\tilde X\in{\rm Ind}(\cU)$.  
{\mygreen
Let $X=\tilde X\cap S_2$, which therefore coordinates
\begin{equation}\label{eq_some_stuff_coordinated_by_the_coord_with_D_thm}
A\cap B\cap C\cap D=C\cap D, \ A\cap B\cap C,\ D,
\ A\cap B,\ A\cap C,\ B\cap C.
\end{equation} 
Now let us repeat the proof of 
Theorem~\ref{th_main_three_subspaces_decomp}
in Subsection~\ref{su_proof_of_th_main_three_subspaces_decomp},
with the above in mind.
In the notation there, with
$$
X_1 
= {\mygreen X \cup
\{a_{m+1},\ldots,a_{m_1}\}\cup
\{b_{m+1},\ldots,b_{m_2}\}\cup
\{c_{m+1},\ldots,c_{m_3} \} ,}
$$
and
$$
X_2 = \{ a_1,\ldots,a_m,b_1,\ldots,b_m\},
$$
we have that $X_1\cup X_2$
is a basis for $A+B+C$.  But since $A,B,C$ are coordinated,
$X_2=\emptyset$ and $m=0$.
Hence $X_1$ coordinates $A,B,C$.
Since $X_1$ contains $X$, $X_1$ also coordinates $D$,
since $X$ coordinates everything in
\eqref{eq_some_stuff_coordinated_by_the_coord_with_D_thm}.
Hence $X_1$ coordinates $A,B,C,D$.
Hence also $[X_1\setminus D]_D$ coordinates $[A]_D,[B]_D,[C]_D$ in $\cU/D$.}
\end{proof}

\begin{proof}[Proof of Theorem~\ref{th_quotient_via_subspace_in_two}]
Consider the decomposition of $\cU$ into $\cU_1,\cU_2$ given by
Theorem~\ref{th_main_three_subspaces_decomp}.
Since $D\subset A\cap B$, and $A\cap \cU_2$ and $B\cap \cU_2$ do not
intersect, and $A,B$ both factor through the decomposition,
we have $A\cap B\subset \cU_1$ and hence $D\subset\cU_1$.
Hence $D$ also factors through this decomposition.

Next we claim that
$$
{\rm DisCoord}^{\cU/D}([A]_D,[B]_D,[C]_D) 
$$
can be written as a formula involving $\dim$ of expressions involving 
the operations $+,\cap$ (and parenthesis) applied to $A,B,C,D$:
to see this, we write 
$$
\dim( [A]_D ) = \dim(A)-\dim(A\cap D),
$$
and similarly for $B,C$ replacing $A$; for $i=1,2,3$ we get similar
formula for 
$$
\dim^{\cU/D}\bigl( S_i([A]_D,[B]_D,[C]_D) \bigr)
=
\dim^\cU\bigl( S_i(A+D,B+D,C+D) \bigr) - \dim(D).
$$
In this way we can write the discoordination of the images of $A,B,C$
in $\cU/D$ as a formula involving 
$\dim$ and the $+,\cap$ applied to $A,B,C,D$.
It then
follows from Theorem~\ref{th_factorization_under_cap_sum} that
$$
{\rm DisCoord}^{\cU/D}([A]_D,[B]_D,[C]_D) 
$$
\begin{equation}\label{eq_discoordination_A_B_C_in_cU_i_mod_D_sum}
{\mygreen
= 
\sum_{i=1}^2 {\rm DisCoord}^{\cU_i/D}
\bigl( [A\cap\cU_i]_{D\cap\cU_i} ,
[B\cap\cU_i]_{D\cap\cU_i} ,
[C\cap\cU_i]_{D\cap\cU_i} \bigr) . }
\end{equation} 
Since $\cU_2\cap D=0$, the $\cU_2$ (i.e., $i=2$) term 
of \eqref{eq_discoordination_A_B_C_in_cU_i_mod_D_sum}
is just 
$$
{\rm DisCoord}(A,B,C).
$$
Since
$D\cap \cU_1 \subset (A\cap \cU_1)\cap (B\cap \cU_1)$,
and $A\cap \cU_1,B\cap \cU_1,C\cap \cU_1$ are coordinated,
Lemma~\ref{le_coordinating_with_D} implies that the $\cU_1$ (i.e., $i=1$) term
of \eqref{eq_discoordination_A_B_C_in_cU_i_mod_D_sum}
vanishes.
\end{proof}
[We remark that one can give a slight variation of the above proof,
using the fact that
$\cU/D$ is isomorphic to the direct
sum of $\cU_1/D$ and $\cU_2$; this gives another way to arrive
at the same calculation of the discoordination.]

\section{Coded Caching: Introduction and the Case $N=K=3$}
\label{se_Joel_coded_intro}

Recall our discussion of information theory and our particular
notions, including that of a {\em linear random variable}
(Definition~\ref{de_linear_random_vars})
in Subsection~\ref{su_review_info_theory}.
We remind the reader that for our entire discussion of coded 
caching---which comprises most of the rest of this article---we
will often use 
Notation~\ref{no_joint_random_variables_information_theory};
hence for linear random variables, $Y_1,\ldots,Y_m\subset\cU$
in an $\field$-universe, $\cU$, we often write $(Y_1,\ldots,Y_m)$
or simply $Y_1\ldots Y_m$---notation common in information theory---for
the subset $Y_1+\cdots+Y_m\subset\cU$.

In this section we introduce the problem of {\em coded caching} and 
discuss one special case (of $N=K=3$ in the standard notation)
that is likely one of the ``easiest'' open special case of this problem.

There is an extensive literature on the many variations of the
problem of {\em coded caching},
beginning with the seminal paper
\cite{MA_niesen_2014_seminal}; see
\cite{MA_avestimehr_2019_survey,saberali_thesis} 
for a survey of the literature;
we specifically use recent results from the impressive, computer-aided
inequalities of the work of Tian
\cite{tian_2018_computer_aided}.
Let us give the basic definitions; most authors use the original notation
of \cite{MA_niesen_2014_seminal}.

\subsection{Introduction to Coded Caching and Informal Description}
\label{su_coded_caching_notation}

We start by describing a mild simplification of Maddah-Ali and Niesen problem
\cite{MA_niesen_2014_seminal}.  We stick to their notation.
In this subsection we begin with an informal description,
before giving the formal (and less intuitive) description in the
next subsection.

{\mygreen For $N,K,F\in\naturals=\{1,2,\ldots\}$ 
and rationals $M,R\ge 0$ (and
$M,R\le N$ in practice), here is an informal description of an
{\em $(N,K,F)$-coded caching scheme} that achieves the
{\em memory-rate pair $(M,R)$}:}
a central server has access to $N\in\naturals=\{1,2,\ldots\}$ 
files or documents, 
$W_1,\ldots,W_N$,
each consisting of $F\in\naturals$ bits, i.e., each $W_i$ is an
element of $\{0,1\}^F$.
There are $K\in\naturals$ users, where each user has a ``cache'' 
(i.e., storage device, typically ``small'' in some sense) of 
size $MF$ for some rational number $M$, and we are interested in
the case where $0<M<N$, so the caches can store some information regarding
the documents, but
not all $N$ documents.  
The rough idea is that there are two phases in this process:
in the first phase, each user can examine all $NF$ bits of all the documents,
but the user does not have enough storage to store all $NF$ bits;
in this phase 
each user knows that in the second phase they will need to obtain
exactly one of the $N$ documents, but the user does not know which document
they will need until the first phase is over.
The first phase is called the {\em placement phase}, 
during which 
the server broadcasts all $NF$ bits in $W_1,\ldots,W_N$,
and for 
$i=1,\ldots,K$, 
user $i$ can store up to $MF$ bits of information, i.e., can store a
function
$Z_i=Z_i(W_1,\ldots,W_N)$ of $MF$ bits; 
we refer to $Z_i$ as the ``cache'' of user $i$;
the server knows the function $Z_i$ (and hence knows the values 
$Z_i(W_1,\ldots,W_N)$).
Between the first 
in the second phase, 
the server and each user are given (by something external
to this system) a vector
$\mec d=(d_1,\ldots,d_K)$ such that for each 
$i\in[K]=\{1,\ldots,K\}$,
$d_i\in [N]$,
and user $i$ requests to be
able to reconstruct $W_{d_i}$;
we refer to $\mec d=(d_1,\ldots,d_K)$ as 
the {\em demand vector}.
In the second phase---the {\em delivery phase}---the 
central server 
broadcasts a message $X_{\mec d}$.
By a {\em memory-rate pair} we mean a pair $(M,R)$ of rational numbers,
and we say that
such a pair is {\em achievable} (for a given value of $(N,K)$)
if for some $F\in\naturals$ there is a caching scheme as above, i.e.,
a choice of $Z_1,\ldots,Z_K$, each of size $MF$ bits, each 
$Z_i=Z_i(W_1,\ldots,W_N)$,
such that for
all $\mec d\in [N]^K=\{1,\ldots,N\}^K$ there exists a function
$X_{\mec d}=X_{\mec d}(W_1,\ldots,W_N)$ of at most 
$RF$ bits, such that for all $i=1,\ldots,K$,
the values of $X_{\mec d}$ and $Z_i$ (and $\mec d$)
determine the document
$W_{d_i}$ (needed by user $i$).

\begin{remark}
{\mygreen
If an $(N,K,F)$-coded caching scheme achieves a
memory-rate pair $(M,R)$,
then 
one easily sees (see Subsection~\ref{su_coded_caching_notation}
below) that such a scheme exists with $F$ replaced by
any multiple of $F$.  It easily follows that for any $F'\in\naturals$,
there is an $(N,K,F')$ with memory-rate pair
$(M+o(1),R+o(1))$ for $F'$ large.
For this reason, much of the coded-caching literature
studies which memory-rate pairs $(M,R)$ are achievable, without
regard to $F$.
However, if $F$ is very large (think of $F=10^{100}$), then
such a scheme may be wildly impractical for practical values of $F'$.}
\end{remark}

\begin{example}\label{ex_N_K_2_M_half}
Let $N=K=2$;
this case was solved 
in \cite{MA_niesen_2014_seminal} and illustrates
the novelty of this problem; their solution was complete in the
sense that for all rational $M\in[0,2]$ they determined the
smallest $R$ with $(M,R)$ achievable; we discuss this later.
Here is one of their caching schemes:
let $F=2$, and let
$W_1=(A_1,A_2)$ and $W_2=(B_1,B_2)$ where $A_1,A_2,B_1,B_2\in\{0,1\}$.
We claim that the pair $(M,R)=(1/2,1)$ is achievable:
indeed, let $Z_1=A_1\oplus B_1$ and $Z_2=A_2\oplus B_2$, where
$\oplus$ denotes addition modulo 2.
If $\mec d=(1,1)$, i.e., both users want document $1$, the we 
set $X_{(1,1)}=W_1$, i.e., in the delivery phase the server broadcasts
$W_1=(A_1,A_2)$.  Similarly we may take $X_{(2,2)}=W_2$.
If $\mec d=(1,2)$, i.e., user 1 wants document 1 and user 2 wants document
2, we see that we may take $X_{(1,2)}=A_2\oplus B_1$, so that
(1) $X_{(1,2)}$ and $Z_1=A_1\oplus B_1$ allow user $1$ to determine 
$W_1=(A_1,A_2)$, and 
(2) $X_{(1,2)}$ and $Z_2=A_2\oplus B_2$ allow user $1$ to determine 
$W_2=(B_1,B_2)$.
Similarly we can take $X_{(2,1)}=A_1\oplus B_2$.
Hence each cache $Z_i$ stores $MF=2$ bits, and each $X_{\mec d}$ can
consist of only $RF=2$ bits, which achieves
$(M,R)=(1/2,1)$.
\end{example}
Henceforth we will usually drop the parentheses and commas in writing
the $X_{\mec d}$, e.g., writing $X_{12}$ for $X_{(1,2)}$.

The motivation for coded caching comes from computer caches,
where phase one is a time of high bandwidth on the communication
network, and phase two is a one of low bandwidth.
We note that there are many other ways to view the coded caching problem;
for example, we may view the server as an online library, the $N$ documents as
books, and the $K$ users as students.
We may also view the server as a radio station, and the users as each
having a radio.
As such, we expect that this problem may have applications beyond the
original motivation in \cite{MA_niesen_2014_seminal}.

\subsection{Formal Definition of a Classical and Linear
Coded Caching Scheme}

In this subsection we define the usual (or {\em classical}) 
formal definition
of a coded caching scheme, and then we introduce the version
with $\field$-linear random variables for an arbitrary field, $\field$;
the case $\field=\integers/2\integers$ reduces to the linear
case of the classical definition.

\ignore{
Most of the literature works with coded caching schemes over the
field $\field=\integers/2\integers$; the lower bounds (i.e., outer bounds)
we prove in this paper are valid for any field, $\field$.

If $\field$ is a finite field of $q$ elements, 
then we scale the usual entropy $H=H_2$ 
(Subsection~\ref{su_review_info_theory},
\eqref{eq_that_defines_entropy})
by introducing $H_q=(1/\log_2(q))H$; this is needed to
coincide with the notion of dimension for linear coded caching schemes.
}

\begin{definition}[Classical Coded Caching Scheme]
\label{de_coded_caching_info_theory}
Let $N,K,F\in\naturals$.
By an {\em classical coded caching scheme
with $N$ documents of size $F$
and $K$ users},
{\mygreen or simply an {\em $(N,K,F)$-coded caching scheme},}
we mean a collection of random variables 
$$
\Bigl( \{W_i\}_{i\in [N]},
\{ Z_j\}_{j\in [K]} , \{ X_{\mec d}\}_{\mec d\in [N]^K} \Bigr)
$$
on a source (i.e., probability space) $(S,P)$, such that
\begin{enumerate}
\item 
$W_1,\ldots,W_N$ are independent uniformly distributed random
variables $S\to\{0,1\}^F$;
\item
for each $K$-tuple $\mec d=(d_1,\ldots,d_K)\in [N]^K$, i.e.,
with each $d_i\in [N]$,
we have 
\begin{equation}\label{eq_W_d_i_is_implied_by_Z_and_X}
\forall j\in[K],\quad
(Z_j,X_{\mec d})\implies W_{d_j} .
\end{equation} 
\end{enumerate}
We say that a scheme {\em achieves the memory-rate pair $(M,R)$}
if 
\begin{align*}
\forall j\in [N], & \quad  H(Z_i)\le MF \\
\forall \mec d \in [N]^K, & \quad  H(X_{\mec d})\le RF .
\end{align*}
\end{definition}
[One could generalize this setup by fixing a $q\in\naturals$ with $q\ge 3$,
replacing $\{0,1\}$ with $\{0,1,\ldots,q-1\}$, and replace
$H$ with $H_q=(1/\log_2q)H$.  We have not seen this in the literature
and will not address this in this article.]

We begin with a few remarks.

\begin{remark}\label{re_original_definition}
In the original definition of
Maddah-Ali and Niesen
(end of Section~II of \cite{MA_niesen_2014_seminal}),
there they add a parameter $\epsilon>0$, and replace 
\eqref{eq_W_d_i_is_implied_by_Z_and_X} by the condition that
\begin{equation}\label{eq_maddah_ali_niesen_epsilon}
{\rm Prob}_{(S,P)}[ (Z_j,X_{\mec d})\implies W_{d_j} ] \ge 1-\epsilon;
\end{equation}
then they define $(M,R)$ to be achievable if for any $\epsilon>0$
and $F$ sufficiently large there is a scheme
with parameters $N,F,K,\epsilon$ satisfying
\eqref{eq_maddah_ali_niesen_epsilon}.
This allows for a more general notation of a scheme, in which the
$Z_j$ and $X_{\mec d}$ are not necessarily functions of $W_1,\ldots,W_N$. 
All the lower bounds in \cite{MA_niesen_2014_seminal}
on $R$ as a function of $N,K,M$ are valid
for this more general notion, by appealing to Fano's inequality.
By contrast, all the caching schemes that we have seen in the 
literature that achieve an optimal $(M,R)$ value have the $Z_j,X_{\mec d}$
being linear functions of $W_1,\ldots,W_N$.
\end{remark}

\begin{remark}\label{re_deterministic}
Some authors (e.g., Tian in \cite{tian_2018_computer_aided})
use Definition~\ref{de_coded_caching_info_theory} with
\eqref{eq_W_d_i_is_implied_by_Z_and_X} rather than
the original definition.
This greatly simplifies matters:
in this case we easily see that:
\begin{enumerate}
\item
\eqref{eq_W_d_i_is_implied_by_Z_and_X} remains valid if we
replace the source $(S,P)$ by the (possibly) coarser source
that groups together all elements of $S$ with the same value 
of $(W_1,\ldots,W_N)$; hence
one can take the source to be 
the uniform distribution on $S=\{0,1\}^{NF}$, whose elements are
described by coordinates
$$
(w_{1,1},\ldots,w_{1,F}, w_{2,1},\ldots,w_{2,F},\ldots,w_{N,F} )
$$
with $w_{i,j}\in\{0,1\}$ and where
$W_i$ is the random variable $(w_{i,1},\ldots,w_{i,F})$;
\item
in doing so, $Z_j$'s and $X_{\mec d}$ become
functions of the $w_{i,j}$'s, or equivalently of 
$(W_1,\ldots,W_N)$.
\end{enumerate}
\end{remark}

For linear schemes it is simpler to work with
linear random variables in the sense of 
Definition~\ref{de_linear_random_vars},
Subsection~\ref{su_review_info_theory}.

\begin{definition}[Linear Coded Caching Scheme]
\label{de_linear_coded_caching}
Let $N,K,F\in\naturals$, and $\field$ be an arbitrary field.
By an {\em $\field$-linear coded caching scheme
with $N$ documents of size $F$
and $K$ users},
{\mygreen or simply an {\em $\field$-linear $(N,K,F)$-coded caching
scheme},}
we mean a collection
of subspaces
$$
\Bigl( \{W_i\}_{i\in [N]},
\{ Z_j\}_{j\in [K]} , \{ X_{\mec d}\}_{\mec d\in [N]^K} \Bigr)
$$
of an $\field$-universe, $\cU$, such that
\begin{enumerate}
\item 
$W_1,\ldots,W_N$ are independent subspaces, each of dimension $F$;
and
\item
$
\forall j\in[K],\ \mec d\in[N]^K,\quad
W_{d_j}\subspace Z_j+X_{\mec d} .
$
\end{enumerate}
We say that a scheme {\em achieves the memory-rate pair $(M,R)$}
if 
\begin{align*}
\forall j\in [N], & \quad  \dim(Z_i)\le MF \\
\forall \mec d \in [N]^K, & \quad  \dim(X_{\mec d})\le RF 
\end{align*}
\end{definition}
We will generally limit our discussion to the case 
$\field=\integers/2\integers$, although many of our results,
including the lower bounds
we prove, hold for arbitrary $\field$.

If in Definition~\ref{de_coded_caching_info_theory},
with notation as in Remark~\ref{re_deterministic},
the $Z_j$'s and $X_{\mec d}$'s are linear functions of
the $w_{i,j}$'s, then all the random variables
involved linear functions of the $w_{i,j}$'s, whose associated
linear random variables reduce to 
Definition~\ref{de_linear_coded_caching} in the case
$\field=\integers/2\integers$.

\subsection{Preliminary Remarks}
\label{su_preliminary_ramarks_on_coded_caching}

Next we make some important observations about coded caching
that mostly hold for either definitions we consider
(and also the original definition, as in
Remark~\ref{re_original_definition}).

\subsubsection{Concatenation of Caching Schemes}

A fundamental observation \cite{MA_niesen_2014_seminal}
is that one can concatenate caching schemes.
If for some $N,K$, and $\ell=1,2$ we have two caching schemes
$$
\cS^\ell = \Bigl( \{W^\ell_i\}_{i\in [N]},
\{ Z^\ell_j\}_{j\in [K]} , \{ X^\ell_{\mec d}\}_{\mec d\in [N]^K} \Bigr)
$$
with document size are $F^\ell$ (not necessarily equal),
on two sources, 
then we define their {\em concatenation}
to be the set of random variables
$$
W_i=\bigl(W^1_i,W^2_i\bigr), \quad
Z_i=\bigl(Z^1_j,Z^2_j\bigr), \quad
X_{\mec d}= \bigl( X^1_{\mec d},X^2_{\mec d} \bigr),
$$
defined on the source that is the 
product of the
sources of $\cS^1$ and $\cS^2$; hence $N,K$ are unchanged, but the
document size is $F^1+F^2$.
Similarly for linear caching schemes, defined by taking direct sums,
i.e.,
$$
W_i=W^1_i \underline\oplus W^2_i, \quad
Z_i=Z^1_j \underline\oplus Z^2_j, \quad
X_{\mec d}= X^1_{\mec d} \underline\oplus X^2_{\mec d} ,
$$
as subspaces of the
direct sum of the two universes.
If---in either the classical or linear setting---the 
two schemes, respectively, achieve the memory-rate trade-offs
$(M^1,R^1)$ and $(M^2,R^2)$, then their concatenation achieves
the memory-rate trade-off $(M,R)$ where
\begin{equation}\label{eq_convex_combinations_M_R}
M = \frac{M^1F^1+M^2F^2}{F^1+F^2}, \quad
R = \frac{R^1F^1+R^2F^2}{F^1+F^2}
\end{equation}
We can similarly concatenate and finite number of caching schemes.
It follows that for any $n_1,n_2\in\naturals$, we may concatenate 
$n_1$ concatenations of the first scheme with $n_2$ of the second and
achieve $(M,R)$ with
\begin{equation}\label{eq_convex_combinations_M_R_with_ns}
M = \frac{M^1(F^1n_1)+M^2(F^2n_2)}{F^1n_1+F^2n_2}, \quad
R = \frac{R^1(F^1n_1)+R^2(F^2n_2)}{F^1n_1+F^2n_2} .
\end{equation}
It follows that we can achieve any rational convex combination
of $(M^1,R^1)$ and $(M^2,R^2)$.

\subsubsection{Limit Achievable and Lower (or ``Outer'') Bounds}

It becomes convenient to say that for schemes with fixed $N,K$, and
fixed $\field$ for linear schemes,
a pair of non-negative real numbers
$(M,R)$ is {\em limit achievable} if it is the limit point
of achievable pairs.

The notion of limit achievable is mostly a convenience.  We remark
that if $\alpha,\beta,\gamma$ are positive reals for which we can
prove $\alpha M+\beta R\ge \gamma$ for any scheme (classical or linear)
that achieves trade-off $(M,R)$ with fixed $N,K$ (and $\field$
for linear schemes), then this
bound also holds for any limit point therefore.

The definition of ``limit achievable''
does raise some interesting questions: for example,
is there a rational point $(M,R)$ that is limit achievable
(with $N,K,\field$ fixed) that is not achievable (i.e., for some 
single scheme)?

\subsubsection{Easy Lower Bounds}

There are some obvious lower bounds on $(M,R)$ in both the classical
or linear case;
for example,
if $N=K$, then
$$
R + K M \ge K, \quad KR+M\ge K,
$$
which follow from the fact that $X_{(1,\ldots,K)},Z_1,\ldots,Z_K$ determine
$W_1,\ldots,W_K$, and $X_{(1,\ldots,1)},\ldots,X_{(K,\ldots,K)}$ and $Z_1$
also determine $W_1,\ldots,W_K$.
These are easy to prove using information theory for classical schemes, e.g.,
$$
RF + KMF \ge 
H(X_{1\ldots K})+H(Z_1)+\cdots+H(Z_K) 
$$
$$
\ge H(X_{1\ldots K},Z_1,\ldots Z_K)=H(W_1,\ldots,W_K)= KF  ,
$$
and similarly for linear schemes, with ``$\dim$'' replacing ``$H$.''

\subsubsection{Bounds on $\field$-Linear Caching Schemes}

Of course, a lower bound for classical coded caching schemes immediately
implies the same bound for $\field$-linear schemes
with $\field=\integers/2\integers$.

We remark that the lower bounds on $R,M$ in the coded caching
literature likely hold
for $\field$-linear schemes for an arbitrary 
field, $\field$: indeed, all the bounds we have
seen can be derived from ``elemental inequalities''
(e.g., \cite{tian_2018_computer_aided}, equations~(8) and~(9) of
Section~2.3), which presumably 
translate into inequalities on dimensions; e.g.,
$H(X|Y)\ge 0$ translates to $\dim^{\cU/Y}([X]_Y)\ge 0$, and
$I(X_1;X_2|Y)\ge 0$ translates to
$\dim^{\cU/Y}([X_1]_Y\cap [X_2]_Y)\ge 0$.

Whether bounds can translate the other way---at least in coded caching or
some other
``purely information theoretic'' problem---is hardly clear.
Certainly linear information theory is simpler and
more expressive than traditional
information theory: for example, the expression
$\dim^{\cU/Y}([X_1\cap X_2]_Y)$ doesn't appear to have
a non-linear analog
(this expression is less than or equal to
both $\dim^\cU(X_1\cap X_2)$ and $\dim^{\cU/Y}([X_1]_Y\cap [X_2]_Y)$,
and can be strictly less than both).
Similarly for intersections of three or more subspaces, for discoordination,
etc.

\subsubsection{The Case $N=K=2$}

The case $N=K=2$ case was entirely solved (i.e., for all $M$ the minimum
value of $R$ was determined) in
\cite{MA_niesen_2014_seminal}:
namely, aside from the obvious lower bounds $M+2R\ge 2$ and $2M+R\ge 2$,
they show that $M+R\ge 3/2$ using the following clever argument:
namely, they observe that
$$
2M+2R \ge H(X_{12},Z_1) + H(X_{21},Z_2) ,
$$
and then use 
$$
H(X_{12},Z_1) + H(X_{21},Z_2) =
H(X_{12},Z_1,X_{21},Z_2)+I(X_{12},Z_1;X_{21},Z_2)
\ge 2 + 1 = 3.
$$
The same bound applies for linear schemes over an arbitrary $\field$
by replacing ``$H$'' with ``$\dim$'' and
$I(X_{12},Z_1;X_{21},Z_2)$ with $\dim((X_{12}+Z_1)\cap (X_{21}+Z_2))$.

We will use a variation of this approach to get two lower (i.e., outer)
bounds for $N=K=3$ involving discoordination.

\subsubsection{Recent Literature and Currently Open Problems}

There are a large number of variants of the coded caching problem
(see, for example, \cite{saberali_thesis,MA_avestimehr_2019_survey}).
Results in \cite{MA_avestimehr_2019_survey} determined lower bounds
on $R$ as a function of $K,N,M$ for classical coded caching
that is provably optimal to within a multiplicative factor of 
2.00884.
There are also many values of $N,K$ where the optimal value is known
for some values of $M$.

By contrast,
there are relatively few values of $(N,K)$ for which
the optimal value of $R=R(M)$ is fully resolved, in the sense that
it is known for all 
$M\in [0,N]$:
the cases $K=2$ and $N\ge 3$ was fully resolved 
by Tian in \cite{tian_2018_computer_aided}, which also fully resolved
the case $(N,K)=(2,3)$.
As of Tian's work \cite{tian_2018_computer_aided},
all other cases of $K\ge 3$ were open for some values of $M$.

Tian \cite{tian_2018_computer_aided}
used an impressive computer-aided search to generate numerous
new lower bounds for some small pairs $(N,K)$; 
Tian's search is based on
a (generally large) collection ``elemental inequalities'' of
Yeung \cite{yeung97}, which
exploit the non-negativity of entropy and of 
the two-variable mutual information
(see equations~(8) and~(9) of \cite{tian_2018_computer_aided});
see Section~2 of \cite{tian_2018_computer_aided} for more details
on the algorithms and previous results.
Tian mentions that his computer-aided linear program for the case
$(N,K)=(2,4)$ would involve some 200 million inequalities, which
he therefore
reduces 
by exploiting symmetrization (which we discuss in
Section~\ref{se_Joel_symmetrization}) and other methods.

Our interest, like that in \cite{tian_2018_computer_aided}, is
to determine for small pairs $(N,K)$ the exact optimal value
of $R$ for every $M$ for linear schemes.
Our motivation is to develop new tools in linear algebra and information
theory that may arise to find these exact values,
such as our theorems on coordination and discoordination that
we developed in earlier sections.
Our article deals only with the case $(N,K)=3$.

\subsection{The Case $N=K=3$ and The Methods of Tian}
\label{su_review_N_K_both_3}

In this paper we focus entirely on the case $N=K=3$.

Prior to Tian's work, the optimal value of $R$ for a given $M$
was known for all $M$
except $1/3\le M\le 1$: in more detail,
the article \cite{MA_niesen_2014_seminal} showed that
$$
3R+M \ge 3,
\ 3R+2M \ge 5,
R+3M \ge 3,
$$
and that these lower bounds are tight for all $M\ge 1$, due to the 
achievability of $(M,R)=(1,1),(2,1/3),(3,0)$ by the 
caching schemes given there.
The achievability of $(M,R)=(1/3,2)$ was shown by
\cite{chen_journal}, which settled the case $M\le 1/3$ in view of
the inequality
$3R+M\ge 3$.
This left the case of $1/3<M<1$ open.

In \cite{tian_2018_computer_aided}, Tian 
gave the new lower bounds
\begin{equation}\label{eq_Tians_N_K_3_inequalities}
M+R\ge 2,
\quad
2M+R \ge 8/3 ,
\end{equation} 
with human readable proofs (tables A24--A27 there).
Below we give a simpler derivation of Tian's inequality $M+R\ge 2$.
The intersection point of Tian's inequalities 
\eqref{eq_Tians_N_K_3_inequalities},
is the point $(M,R)=(2/3,4/3)$,
and Tian proves that this memory-rate trade-off is unlikely to
be achieved by a linear scheme, which involves a rather
ingenious technique to give a lower bound $R$;
we shall refer to this as {\em Tian's method}
(see Theorem~\ref{th_tian_type_scheme} below),
and refer to the type of caching scheme Tian studies
regarding $(2/3,4/3)$ as a {\em Tian scheme}
(see Definition~\ref{de_Z_schemes} below).

In more detail,
Tian reports that his linear programs derive the
following values for various joint entropies at the point
$(2/3,4/3)$,
as scheme, see \cite{tian_2018_computer_aided}, Table 4:
namely
setting with $m=1/3$ (which equals $M/2$),
Tian reports
\begin{equation}\label{eq_tian_reports}
\begin{gathered}
H(Z_1|W_1)=2m,
\ H(Z_1|W_1W_2)=m,\\
H(Z_1Z_2|W_1W_2)=2m,
\ H(Z_1Z_2Z_3|W_1W_2)=3m
\end{gathered}
\end{equation} 

(Tian remarks that these
results are reported by a floating point computation,
without giving a human readable proof, and
so there is a chance that what appears to be, say, $2m$, is actually
$(2\pm \epsilon)m$ where $\epsilon$ is a presumably small machine error;
see the
remarks in Section~5.4 of \cite{tian_2018_computer_aided}; Tian's
conclusions regarding $(2/3,4/3)$
would still hold at this point to within a small additive
multiple of $\epsilon$.)

Assuming the values of
\eqref{eq_tian_reports} hold (exactly),
Tian concludes
(see discussion below Table~4 there) that
if all random variables are linear functions of the bits of $W_1,W_2,W_3$,
then (1) $F$ must be divisible by 3 (to achieve $(2/3,4/3)$), and
(2) based on \eqref{eq_tian_reports}
each 
$W_i$ decomposes as a sum of three (linearly independent) subspaces,
\begin{equation}\label{eq_disjoint_scheme}
W_1=A_1+A_2+A_3,
\ W_2=B_1+B_2+B_3,
\ W_3=C_1+C_2+C_3,
\end{equation} 
each factor of dimension $F/3$, such that for $i=1,2,3$ we have
\begin{equation}\label{eq_cL_for_Z_i}
Z_i = \cL_i(A_i,B_i,C_i)
\end{equation} 
where $\cL_i$ is some linear function.
Tian then gives an extremely clever argument to show that no such
$\cL_i$ can achieve the $(2/3,4/3)$ bound; a direct linear algebraic
proof seems difficult, and Tian challenges the reader to find such a 
proof.

In fact, Tian's argument (as is) can be used to show that any
scheme with properties similar to those required to achieve
$(2/3,4/3)$ must satisfy $2R+3M\ge 5$.
Let us make this precise.

\begin{definition}\label{de_separated}
Consider an $\field$-linear coded caching for $N=K=3$, where $F$ is divisible
by three, and each document $W_i$ is decomposed into factors
of dimension 
$F/3$ as in \eqref{eq_disjoint_scheme}.
We say that $Z_1,Z_2,Z_3$ are {\em separated} 
if we have
\eqref{eq_cL_for_Z_i} for each $i=1,2,3$, for some linear function
$\cL_i$.  
\end{definition} 
We can similarly define a separated scheme for classical schemes,
where each $A_i,B_i,C_i$ has entropy $F/3$.

Recall the meaning of $A\oplus B$ 
(Subsection~\ref{su_sum_direct_sum_oplus}) when $A,B$ are subspaces
of the same dimension of a vector space: this means
that we understand that we have an isomorphism $\nu\from A\to B$,
and we set
$$
A\oplus B= A\oplus_\nu B
= \{ a+\nu(a) \ | \ a\in A \};
$$
this is equivalent to choosing bases $a_1,\ldots,a_m$ of $A$,
and $b_1,\ldots,b_m$ of $B$,
and $A\oplus B$ to be the subspace spanned by $a_i+b_i$.

Tian shows that \eqref{eq_tian_reports} implies that,
with notation as in Definition~\ref{de_separated},
we must have $Z_i$ that is spanned by $A_i\oplus B_i$ and $A_i\oplus C_i$
(i.e., $a_i+c_i$, where $c_1,\ldots,c_m$ is some basis for $C_i$).
In this case $X_{123},Z_1$ allows user $1$ to infer
$$
A_1,A_2,A_3,B_1,C_1,
$$
and similarly for users $2$ and $3$.

Tian's argument in Section~5.4 of \cite{tian_2018_computer_aided}
can prove the following more general theorem.

\begin{theorem}\label{th_tian_type_scheme}
For $N=K=3$ and $F$ divisible by $3$, let $W_i$ be decomposed
into subspaces of dimension $F/3$ as in \eqref{eq_disjoint_scheme}.
Let $Z_1,Z_2,Z_3$ be any linear scheme such that
for $i=1,2,3$, 
$$
(X_{123},Z_i) \implies W_i,A_i,B_i,C_i.
$$
Then if such as scheme achieves the memory-rate trade-off $(M,R)$, we have
$$
2R+3M \ge 5.
$$
Moreover, if $R'=\dim(X_{123})$, then $2R'+3M\ge 5$, and similarly
with the indices $1,2,3$ permuted in any way.
\end{theorem}
Tian used this result to show that $(2/3,4/3)$ cannot be achieved
by a linear scheme, assuming \eqref{eq_tian_reports}.
We remark that this result is quite strong, in that this implies
that for each distinct $d_1,d_2,d_3\in [3]$, $\dim(X_{\mec d})\ge (5-3M)/2$;
as we will show in the next section, other optimal bounds, such
as the optimal
bound $R\ge 1$ for $M=1$, does not imply that
$\dim(X_{\mec d})\ge 1$ whenever $d_1,d_2,d_3\in [3]$ are distinct, but
only for the worst (or average) case of distinct $d_1,d_2,d_3$.

\begin{proof}
The dimension of $X_{123}$ is $R'F$ for some $R'\le R$; we will show that
$$
2R'+3M \ge 5.
$$

First, we wish to introduce coordinates on $W$ so that
each element of $W$---and therefore of $X_{123}$ (and $Z_1,Z_2,Z_3$)---is
associated to a vector of $3F$ scalars, i.e., an element of
$\field^{3F}$.
To do so,
choose an arbitrary basis, $\cA_1$, of $A_1$, and similarly bases
$\cA_2,\ldots,\cC_3$ of $A_2,\ldots,C_3$;
hence each basis contains $F/3$ elements of $W$, and we let
$\cW$ be the union of these bases, $\cA_1\cup\cdots\cup\cC_3$.
If $u\in W$, we use $\iota_\cW(u)$, or simply $\iota(u)$, 
to denote the element
of $\field^{3F}$ associated to $u$ in the coordinates $\cW$.
Hence $\iota$ can be viewed as an isomorphism $W\to\field^{3F}$.
It will be useful to describe vectors in $\field^{3F}$ as blocks
of $9$ vectors (and similarly for matrices each of whose rows 
are vectors in $\field^{3F}$); in this case we will understand that
we have ordered $\cW$ as
$$
\cA_1,\cA_2,\cA_3,
\cB_1,\cB_2,\cB_3,
\cC_1,\cC_2,\cC_3
$$
(the order of the basis elements in each block $\cA_1,\cdots,\cC_1$
is unimportant).

Let $\dim(X_{123})=R'$, so $R'\le R$.
We therefore have $\iota(X_{123})$ is a subspace of $\field^{3F}$;
choose an arbitrary basis of $X_{123}$ and let $G$ be the matrix
whose rows are $\iota$ of these basis vectors; hence
$$
\iota( X_{123} ) = {\rm RowSpace}(G)
$$
(the row space of $G$)
where $G$ is an $R'\times 3F$ matrix which we view as consisting
of $9$ blocks
\begin{equation}\label{eq_G_1_thru_G_9_Tian_scheme}
G = [G_1\ G_2\ G_3\ G_4\ \cdots \ G_9].
\end{equation} 

Similarly, choose a basis for $Z_1$, which allows us to write
$$
\iota (Z_1) = {\rm RowSpace}(G')
$$
where
$$
G' = [G'_1\ G'_2\ G'_3\ G'_4\ \cdots \ G'_9].
$$
It follows that $\iota(X_{123}+Z_1)$ equals the row space of the block matrix
$$
\iota(X_{123}+Z_1) = {\rm RowSpace}
\left(
\begin{bmatrix}G' \\ G  \end{bmatrix} 
\right)
,
$$
where
$$
\begin{bmatrix}G' \\ G  \end{bmatrix} 
=
\begin{bmatrix}
G'_1 & G'_2 & G'_3 & G'_4 & G'_5 & G'_6 & G'_7 & G'_8 & G'_9 \\
G_1 & G_2 & G_3 & G_4 & G_5 & G_6 & G_7 & G_8 & G_9 \\
\end{bmatrix} .
$$
Consider this matrix with its columns rearranged into two blocks:
$$
\tilde G =
\begin{bmatrix}
G'_1 & G'_2 & G'_3 & G'_4 & G'_7 & \vline & G'_5 & G'_6 & G'_8 & G'_9 \\
G_1 & G_2 & G_3 & G_4 & G_7 & \vline & G_5 & G_6 & G_8 & G_9 \\
\end{bmatrix} ;
$$
clearly $G$ and $\tilde G$ have the same rank.
Since $X_{123}$ and $Z_1$ determine $A_1,A_2,A_3,B_1,C_1$, 
$X_{123}+Z_1$ contains each of these subspaces of $W$,
and hence contains each vector of the bases $\cA_1,\cA_2,\cA_3,\cB_1,\cC_1$.
Hence
$\iota(X_{123}+Z_1)$ contains each standard basis vector associated
the these five bases.
By the basis exchange theorem\footnote{
  We remark that the rows of $G'$ and of $G$ are not necessarily
  independent, namely if $Z_1$ and $X_{123}$ have a non-trivial
  intersection.  Still, we can choose a subset of the rows of the
  matrix formed by the rows of $G'$ and $G$ and apply the basis
  exchange theorem there.  (Alternatively, one can do this proof
  by repeatedly discarding rows of $G'$ that create a linear dependence
  between the rows of $G'$ and $G$, leaving $G'$ to be
  a subset of rows of the original $G'$
  such that the
  rows of $G$ and the new $G'$ are linearly independent but still span
  $\iota(X_{123}+Z_1)$.)
  }, we can apply elementary
(i.e., invertible) row operations on $\tilde G$ to get a matrix
$$
\hat G=
\begin{bmatrix}
I & 0  \\
L_1 & L_2
\end{bmatrix},
$$
where $I$ is the $5F/3 \times 5F/3$ identity matrix,
and $0$ is the $5F/3$ by $4F/3$ zero matrix, and $L_1,L_2$ are 
some matrices; since the total number of rows of $\tilde G$ is at most
$MF+R'F$, the number of rows in the $L_1,L_2$ block matrix is at most
$(M+R'-5/3)F$.
Hence the column space of the two rightmost blocks,
$$
{\rm ColumnSpace}
\left(
\begin{bmatrix}
 0  \\
L_2
\end{bmatrix}
\right)
$$
is at most $(M+R'-5/3)F$.  But since the row operations bringing
$\tilde G$ to $\hat G$ do not change the 
dimension of the column space of any subset
of columns of these matrices, 
it follows that the span of the column vectors of 
\begin{equation}\label{eq_G_G_prime_leftover_from_Z_one}
\begin{bmatrix}
 G'_5 & G'_6 & G'_8 & G'_9 \\
G_5 & G_6 & G_8 & G_9 \\
\end{bmatrix}
\end{equation} 
is of dimension at most $(M+R'-5/3)F$.  In particular, the same bound holds 
for the span of the columns of 
\begin{equation}\label{eq_G_leftover_from_Z_one}
\begin{bmatrix}
G_5 & G_6 & G_8 & G_9 
\end{bmatrix} .
\end{equation} 

The same argument
with $Z_2$ replacing $Z_1$ shows that the column space of
$$
\begin{bmatrix}
G_1 & G_3 & G_7 & G_9 
\end{bmatrix} 
$$
has dimension at most $(M+R'-5/3)F$; 
using $Z_3$, the same holds for 
$$
\begin{bmatrix}
G_1 & G_2 & G_4 & G_5 
\end{bmatrix} 
$$
Since each column of $G$ appears once or twice in the above block matrices,
the entire column space of $G$ is at most 
$$
3 (M+R'-5/3)F.
$$
But the dimension of the column space of $G$ is the rank of $G$,
which equals $R'F$, by assumption;
and hence
$$
R' = {\rm Rank}(G)  \le 3 (M+R'-5/3)F
$$
It follows that $3M+2R'\ge 5$.
\end{proof}

Table~4 in \cite{tian_2018_computer_aided} implies that if 
$(2/3,4/3)$ is achievable, then the $Z_i$ must be as in
Theorem~\ref{th_tian_type_scheme}.

We remark that the above theorem does not analyze $X_{123}$
directly,
rather it draws conclusions based on the particular nature of the
$Z_i$ and the fact that $X_{123}$ and $Z_i$ imply certain information.
Similarly our discoordination lower bounds on $2R+3M$ do not directly
analyze the $X_{\mec d}$.

We also remark that Theorem~\ref{th_tian_type_scheme} is tight
for $(M,R)=(1,1)$ which is achievable.  The proof above gives a
little more: namely, the dimension formula implies that
$$
{\rm Rank}
\begin{bmatrix}
G_5 & G_6 & G_8 & G_9 
\end{bmatrix}
+
{\rm Rank}
\begin{bmatrix}
G_1 & G_3 & G_7 & G_9
\end{bmatrix}
$$
$$
=
{\rm Rank}
\begin{bmatrix}
G_1 & G_3 & G_7 & G_5 & G_6 & G_8 & G_9 
\end{bmatrix}
+
{\rm Rank}
\begin{bmatrix}
G_9
\end{bmatrix} .
$$
Applying the dimension theorem to the column spaces of
$$
\begin{bmatrix}
G_1 & G_3 & G_7 & G_5 & G_6 & G_8 & G_9 
\end{bmatrix}
,\quad
\begin{bmatrix}
G_1 & G_2 & G_4 & G_5 
\end{bmatrix} 
$$
whose intersection is the column space of $[G_1\ G_5]$, we can get a
more precise bound of
$$
3M+2R' \ge 5 + 
{\rm Rank}
\begin{bmatrix}
G_9
\end{bmatrix} 
+
{\rm Rank}
\begin{bmatrix}
G_1 & G_5
\end{bmatrix} 
.
$$
It follows if $3M+2R\ge 5$ holds with equality, then
the ranks of $G_9,G_5,G_1$ are zero.  In other words,
$X_{123}$ cannot involve nonzero coefficients in $\cA_1,\cB_2,\cC_3$.
Indeed, for $(M,R)=(1,1)$ it turns out that
we can take $Z_i=A_iB_iC_i$ and
$X_{123}=A_2\oplus B_1,A_3\oplus C_1,B_3\oplus C_2$,
which avoids $\cA_1,\cB_2,\cC_3$,
and similarly for other $X_{\mec d}$.
\section{Symmetrization and Averaging}
\label{se_Joel_symmetrization}

Let us review the well-known idea of averaging and symmetrization,
which simplify certain expressions
that arise in proving lower bounds (i.e., ``outer bounds'')
in coded caching.

\subsection{Symmetry and Averaging}

Consider either a classical or $\field$-linear $(N,K,F)$-coded caching scheme
\begin{equation}\label{eq_cS_a_coded_caching_scheme}
\cS = \Bigl( \{W_i\}_{i\in [N]},
\{ Z_j\}_{j\in [K]} , \{ X_{\mec d}\}_{\mec d\in [N]^K} \Bigr)  .
\end{equation} 
The symmetric group $S_K$ of permutations on $\{1,\ldots,K\}$ acts
on the $K$ users of a coded caching problem, and similarly
$S_N$ acts on the $N$ documents.  
Since these two actions are independent of each other (i.e., can
be performed in either order), 
this gives us an action of
$S_K\times S_N$ on all random variables
in the scheme, $\cS$, and therefore
the expressions involving the indices 
of $W_i,Z_i,X_{\mec d}$:
namely for $\kappa\in S_K$ and $\nu\in S_N$, we set
\begin{equation}\label{eq_kappa_nu_acts_on_W_Z}
(\kappa,\nu)W_i = W_{\nu i},
\ (\kappa,\nu) Z_i = Z_{\kappa i},
\end{equation} 
and
\begin{equation}\label{eq_kappa_nu_acts_on_X}
(\kappa,\nu) X_{\mec d} = X_{(\kappa,\nu)\mec d},
\quad\mbox{where}\quad
(\kappa,\nu)(d_1,\ldots,d_K) 
=
\bigl( \nu(d_{\kappa(1)}),\ldots , \nu(d_{\kappa(K)}) \bigr) 
\end{equation} 
(since each $d_i$ represents a value in $\{1,\ldots,N\}$ of a document
requested by user $i\in\{1,\ldots,K\}$).

\begin{definition}
Let $\cS$ be either a classical or $\field$-linear 
$(N,K,F)$-coded caching
scheme as in
\eqref{eq_cS_a_coded_caching_scheme}.
For $(\kappa,\nu)\in S_K\times S_N$,
we define {\em the action of $(\kappa,\nu)$ on $\cS$}
denoted $(\kappa,\nu)\cS$, to be the $(N,K,F)$-coded caching
scheme where $(\kappa,\nu)$ applied to the $W_i,Z_j,X_{\mec d}$
is given as in 
\eqref{eq_kappa_nu_acts_on_W_Z} and
\eqref{eq_kappa_nu_acts_on_X}.
We write $S_K\times S_N \cS$ for the concatenation 
of the $(\kappa,\nu)\cS$ ranging over all
$(\kappa,\nu)\in S_K\times S_N$ (whose document size is
therefore $K!\,N!\,F$), and refer to it as the
{\em symmetrization of $\cS$}.
\end{definition}

We remark that if $\cS$ achieves the memory-rate tradeoff
$(M,R)$, then so does the symmetrization of $\cS$.
It follows that for the sake of proving lower bounds, it suffices
to consider the case where the coded caching scheme is 
the symmetrized version of a smaller scheme.
We also easily see that if for $N=K=3$, the $Z_1,Z_2,Z_3$ are
separated, then the same holds for the symmetrization of this scheme. 

For schemes that are the symmetrization of some scheme,
the dimension of all expressions in the $W_i,Z_j,X_{\mec d}$ (involving
$\cap,+$ and parenthesis) are invariant under this
$S_K\times S_N$ action.
This will greatly simplify the proofs of the lower bounds we give in
this article.

\begin{definition}
Let $\cS$ be an $\field$-linear classical 
$(N,K,F)$-coded caching
scheme as in
\eqref{eq_cS_a_coded_caching_scheme}.
For any formula involving $+,\cap$, the variables
$W_i,Z_j,X_{\mec d}$ (and parenthesis), we use
$\dim^{\rm avg}$ to denote the average dimension of this formula
under the action of $S_K\times S_N$.
Similarly if $F_1,F_2$ are such formulas,
we define $\dim^{\cU/F_1,{\rm avg}}([F_2]_{F_1})$ to be
the average dimension of the action of $S_K\times S_N$
on the two expressions.
\end{definition}
We can similarly define $H^{\rm avg}$ of any join of
random variables of a classical coded caching scheme.
Clearly $\dim^{\rm avg}$ applied to a formula of random
variables of a scheme equals the dimension of the same
formula applied to the symmetrization of the scheme, divided by
$K!\,N!$.

For example, for $K=N=3$,
$$
\dim^{\rm avg}( Z_1,W_1,W_3,X_{122} ) 
= \frac{1}{K!\ N!} 
\sum_{(\kappa,\nu)\in S_K\times S_N}
\dim \bigl( Z_{\kappa(1)} W_{\nu(1)} W_{\nu(3)} 
X_{ \nu(\kappa(1,2,2))} \bigr).
$$
Hence this average dimension is also equal to that of
$$
(\kappa,\nu)(  Z_1,W_1,W_3,X_{122} )
$$
for any $\kappa\in S_K$ and $\nu\in S_N$, so that, for example,
\begin{equation}\label{eq_example_of_terms_of_equal_avg_dim}
\dim^{\rm avg}( Z_1,W_1,W_3,X_{122} ) 
=
\dim^{\rm avg}( Z_1,W_1,W_2,X_{133} )  .
\end{equation} 

This averaging technique is convenient in proving lower bounds
on 
{\mygreen achievable}
memory-rate pairs $(M,R)$ (i.e., ``outer bounds'');
see, for example,
equation~(27) of \cite{MA_avestimehr_2019_survey}, where $H^*$ denotes
the average value of $H$, and is used in a number of places
in this article thereafter.
[We do not know where this technique first arose in the literature.]
We will use averaging in our bounds, as well, for the same reasons as 
in \cite{MA_avestimehr_2019_survey}: namely to cancel
the difference terms related by a symmetry, and
therefore of the same average dimension
(such as the difference of the left-hand-side and right-hand-side 
of \eqref{eq_example_of_terms_of_equal_avg_dim}
above).

We also note that any lower bound for fixed $(N,K)$
of the form $\alpha M+\beta R\ge \gamma$
(for positive $\alpha,\beta,\gamma\in\reals$) also applies
the same lower bound with $M,R$ replaced by $\dim^{\rm avg}(Z_i)$ and
$\dim^{\rm avg}(X_{\mec d})$, by applying the lower bound
to the symmetrization of the scheme.
Hence, although $\alpha M+\beta R\ge \gamma$ is a priori
a lower bound on the maximum values of $\dim(Z_i)$ and $\dim(X_{\mec d})$,
the same bound must hold for the average values.

\subsection{Symmetric Coded Caching Schemes}

Tian (\cite{tian_2018_computer_aided}, equation~(16)) defines
a classical 
scheme to be {\em symmetric} if the entropy, $H$,
of the join of any subset of the variables
$W_i,Z_j,X_{\mec d}$ is invariant under the action of $S_K\times S_N$
of the scheme.  
In \cite{tian_2018_computer_aided}, Tian prefers to symmetrize
the coded caching schemes beforehand---which yield
symmetric schemes---in order to simplify
computations and proofs.
See Proposition~3 of Section~3.3 of \cite{tian_2018_computer_aided}.

We need a similar definition, although we require the invariance
of $\dim$ applied to the richer set of expressions in
linear information theory which involve $+,\cap$
(and therefore include invariants such as the discoordination
of any family of subspaces formed by such expressions.

\begin{definition}\label{de_symmetric_scheme}
We say that an $\field$-linear coded caching scheme $\cS$ as in
\eqref{eq_cS_a_coded_caching_scheme} in an $\field$-universe,
$\cU$ is {\em symmetric}
if for each $(\kappa,\nu)\in S_K\times S_N$ there is
an isomorphism $\iota=\iota_{\kappa,\nu}\from \cU\to\cU$
such that for all $i\in [N]$,
$$
(\kappa,\nu)W_i = \iota_{\kappa,\nu} W_i,
$$
and similarly for all the $Z_j$'s and $X_{\mec d}$'s.
\end{definition}
It follows that if an $\field$-linear scheme $\cS$ is symmetric,
then expressions involving the dimension of formulas with $+,\cap$
and the $W_i,Z_j,X_{\mec d}$ are invariant under
the $S_K\times S_N$ action.

We easily see that if $\cS$ is any $\field$-linear scheme
\eqref{eq_cS_a_coded_caching_scheme} in an $\field$-universe, $\cU$, then
$S_K\times S_N \cS$ is symmetric, via the natural
action of $S_K\times S_N$ on the universe that has one copy
of $\cU$ for each element of $S_K\times S_N$.

By contrast, a symmetric scheme need not arise as the symmetrization
of a smaller scheme: for example,
the $N=K=2$ and $M=1/2$
scheme of Maddah-Ali and Niesen in Example~\ref{ex_N_K_2_M_half}
is symmetric; 
{\mygreen more explicitly,}
the action on $S_K\times S_N=S_2\times S_2$
{\mygreen described above}
is given by: (1) the non-identity element of $S_K$
exchanges $A_1,B_1$ with, respectively, $A_2,B_2$,
and, (2) the non-identity element of $S_N$ exchanges $A_1,A_2$ with,
respectively, $B_1,B_2$.
However, in the symmetrization of a scheme,
the dimension of each $Z_j$ (and the other random variables)
must be divisible by $K!\,N!$, which here equals $4$,
and yet in this example $\dim(Z_j)=1$.

\subsection{A Lopsided Example: Average and Worst Case}

As a concrete illustration of the need to use
symmetrization, we give the following example of a 
``highly non-symmetric'' scheme with $N=K=3$ where $X_{123}$ can
be taken to be $0$.

We remark that we will later (Definition~\ref{de_Z_schemes})
refer to this scheme as an example of a {\em pure individual scheme},
although Tian's method (Theorem~\ref{th_tian_type_scheme})
does not apply since this scheme is not separated.

Consider the case $M=1$ where we set
$Z_i=W_i$ for all $i$.  In this case we can take $X_{123}=0$.
While Theorem~\ref{th_tian_type_scheme} shows that $X_{123}$
has dimension $R'F$ with $2R'+3M\ge 5$, i.e.,
$R'\ge (5-3M)/2$, the same cannot be
said of this particular scheme.  Of course, to prove $2R+3M\ge 5$, 
we need to prove that 
some $X_{ijk}$ must have dimension at least $(5-3M)/2$.
We remark that if we use symmetrization and we can prove that
some $X_{ijk}$ has this dimension, then we are actually proving something
stronger, namely that the average dimension of $X_{ijk}$ with
$i,j,k$ distinct is at least $X_{ijk}$.

It is instructive to compare the average versus worst case here:
we may take $X_{213}=W_1\oplus W_2$, so that $X_{213}$ can be
of dimension $F$, 
and similarly (the two other single transpositions) $X_{321}$
and $X_{132}$ can be taken to have dimension $F$.
However, we claim that $X_{312}$ must be of dimension at least
$2F$ under this scheme: indeed,
for $X_{312}$ and $Z_1$ to determine $W_3$, $X_{312}$ must contain
$W_3+\cL_3(W_1)$ for some linear map $\cL_3$, and similarly must 
contain $W_1+\cL_2(W_2)$ and $W_2+\cL_1(W_1)$.
Hence $X_{123}$ has the same row space as a matrix of the form
$$
\begin{bmatrix}
L_1 & 0 & I \\
I & L_2 & 0 \\
0 & I & L_3
\end{bmatrix} ;
$$
by dropping the first row and last column, 
we see that $X_{123}$ has rank at least that of
$$
\begin{bmatrix}
I & L_2 \\
0 & I \\
\end{bmatrix} ,
$$
from which we can eliminate the $L_2$ with row operations, leaving an
identity matrix of size $2F$ by $2F$.  Hence the dimension of
$X_{312}$ must be at least $2F$ (and this suffices, since we easily
verify that setting
$X_{312}$ to be $W_1\oplus W_2,W_1\oplus W_3$ satisfies the conditions
of each user).
A similar calculation holds for the other full-cycle permutation,
i.e., $X_{231}$.

Hence, under the lopsided scheme $Z_i=W_i$, the maximum
dimension of an $X_{ijk}$ is $2F$, and the average over all
$i,j,k$ distinct is $(3F + 2\cdot 2F)/6 = 7F/6$.

\section{The $Z$-Decomposition Lemma}
\label{se_Joel_scheme_decomp_lemma}

In this section, motivated by coded caching in the case $N=K=3$,
we consider for any linear
function $Z$ of a vector space $W$ that decomposes as $W_1,W_2,W_3$,
and show that we can decompose
$Z$ into some subspaces,
each of a particularly simple form with respect to the decomposition
$W_1,W_2,W_3$.

Our intention is to apply this theorem to study linear coded caching
schemes with $N=K=3$; however, this theorem is a really a statement
in linear algebra that holds in a fairly general setting.

\subsection{Definitions and Statement of the Decomposition Lemma}

\begin{definition}\label{de_Z_schemes}
Let $Z$ be linear subspace of an $\field$-universe, $\cU$, that
has a decomposition $W_1,W_2,W_3$.
\begin{enumerate}
\item
We say that
$Z$ is a {\em pure individual scheme} if $Z$ is spanned by
$A=Z\cap W_1$, $B=Z\cap W_2$, and $C=Z\cap W_3$, in which
case $Z=A+B+C$
(typically written $(A,B,C)$ or just $ABC$ in information theory;
see Notation~\ref{no_joint_random_variables_information_theory}).
\item 
We say that
$Z$ is a {\em pure Tian scheme} if there exist
$A\subset W_1$, $B\subset W_2$, and $C\subset W_3$ such that
$A,B,C$ are of the same dimension, $d$, and there are bases
$a_1,\ldots,a_d$ of $A$, $b_1,\ldots,b_d$ of $B$,
and $c_1,\ldots,c_d$ of $C$ such that $Z$ is spanned by
$a_i+b_i,b_i+c_i$ for $i=1,\ldots,d$; 
hence, in the notation of
Subsection~\ref{su_sum_direct_sum_oplus}, we have
$Z$ is the span of $A\oplus_{\nu_1}B$ and
$B\oplus_{\nu_2} C$ where $\nu_1$ is the isomorphism
$A\to B$ taking $a_i$ to $b_i$ for all $i$, and
similarly for $\nu_2\from B\to C$ taking $b_i$ to $c_i$.
\item 
We say that
$Z$ is a {\em pure $AB$ scheme} if there exist
$A\subset W_1$, $B\subset W_2$ such that $A,B$ are of the same 
dimension, and there are bases
$a_1,\ldots,a_d$ of $A$, $b_1,\ldots,b_d$ of $B$
such that $Z$ is spanned by $a_i\oplus b_i$ for $i=1,\ldots,d$;
hence $Z=A\oplus_{\nu}B$ where $\nu$ takes $a_i$ to $b_i$ for all $i$.
\item
We similarly define when $Z$ is a
{\em pure $AC$ scheme} and a {\em pure $BC$ scheme}.
\item
We say that $Z$ is a {\em pure symmetric two-way scheme}
when $Z$ decomposes as a sum of $AB$-, $AC$-, and $BC$-schemes,
each of the same dimension.
\item
We say that
$Z$ is a {\em pure triple sum scheme} if there exist
$A\subset W_1$, $B\subset W_2$, and $C\subset W_3$ such that
$A,B,C$ are of the same dimension, $d$, and there are bases
$a_1,\ldots,a_d$ of $A$, $b_1,\ldots,b_d$ of $B$,
and $c_1,\ldots,c_d$ of $C$ such that $Z$ is spanned by
$a_i+b_i+c_i$ for $i=1,\ldots,d$; in this case 
$Z=A\oplus_{\nu_1} B\oplus_{\nu_2} C$, where
$\nu_1,\nu_2$ are, respectively, the isomorphisms $A\to B$ and $A\to C$
taking $a_i$ to, respectively, $b_i$ and $c_i$.
\end{enumerate}
\end{definition}

\begin{lemma}\label{le_Z_scheme_decomposition}
Let $Z$ be linear subspace of an $\field$-universe, $\cU$, that
has a decomposition $A,B,C$.
Then there exist subspaces $A^j,B^j,C^j$ indexed on integers
$1\le j\le 5$ such that
\begin{enumerate}
\item 
$A^1,\ldots,A^5$ are linearly independent subspaces of $W_1$,
as are $B^1,\ldots,B^5\subset W_2$ and
$C^1,\ldots,C^5\subset W_3$;
\item 
$Z$ is spanned by:
\begin{enumerate}
\item $A^1+B^1+C^1$ (i.e., an individual scheme);
\item $A^2\oplus B^2,B^2\oplus C^2$ (i.e., a Tian scheme);
\item $A^3\oplus B^3$, $A^4\oplus C^3$, $B^4\oplus C^4$
(i.e., an $AB$-, $AC$-, and a $BC$-scheme); and
\item $A^5\oplus B^5\oplus C^5$ (i.e., a triple scheme).
\end{enumerate}
\end{enumerate}
\end{lemma}

Of course, our intended application is to caches $Z_i$.  This lemma
says that any cache is really some combination
of the schemes in Definition~\ref{de_Z_schemes}.
Of course, in a symmetrized scheme, the dimensions of 
all subspaces with superscripts $3$ and $4$ are of the same
dimension, which together comprise a pure symmetric two-way scheme.

\subsection{Proof of the Decomposition Lemma}

Our proof is quite straightforward, although a bit tedious.
The strategy is, roughly speaking to define, in the following stages:
the spaces $A^1,B^1,C^1$,
then $A^2,B^2,C^2$, then $A^3,B^3,A^4,C^3,B^4,C^4$, and then
$A^5,B^5,C^5$.
In each stage we make the necessary definitions and then show a number
of properties of these spaces.
Ultimately we need to show that the $A^1,\ldots,A^5$ are linearly
independent, and similarly for the $B^i$'s and $C^i$'s, and then
we need to decompose any $a+b+c\in Z$ with $a\in A$, $b\in B$, and
$c\in C$ as a sum of the above schemes in a unique way; the uniqueness
is immediate from the linear independence of these subspaces.

\begin{proof}[Proof of Lemma~\ref{le_Z_scheme_decomposition}]
Set $A^1=Z\cap A$, $B^1=Z\cap B$, and $C^1=Z\cap C$.

Say that an $u\in \cU$ is {\em $B$-pairable (with $b$)} 
if for some $b\in B$ we have $u+b\in Z$.
We easily see that set of $B$-pairable elements of $\cU$
are a subspace.
Let us show that 
$$
\mbox{if $u$ is $B$-pairable with $b$, then}
$$
\begin{equation}\label{eq_pairable_uniqueness}
\mbox{$u$ is $B$-pairable with $b'\in B$}
\quad\mbox{iff}\quad
\mbox{$b=b'+b^1$ for some $b^1\in B^1=Z\cap B$};
\end{equation}
``if'' follows from the fact that if $b'=b+z$ with $z\in B^1$,
then $b'$ lies in $B$, and since $u+b'=(u+b)+(b'-b)$,
we have that
$u+b'\in Z$.  The ``only if'' follows from the fact that if
$u+b$ and $u+b'$ both lie in $Z$, and then so does their difference
$b-b'$; hence $b-b'\in Z$; since $b,b'\in B$, we have $b-b'\in Z\cap B$.

We similarly define {\em $A$-pairable (with $a$)} and the analogs of
the remarks
in the previous paragraph hold for $A$-pairable elements of $\cU$;
similarly for {\em $C$-pairable (with $c$)}.

The set of $B$-pairable elements in $A$ (i.e., that also lie in $A$)
is therefore 
a subspace $A'\subset A$,
and it clearly contains all $a\in A^1$ (all paired with $b=0$).
Similarly the $C$-pairable elements of $A$ forms
a subspace $A''\subset A$ containing $A^1$.
Hence $A'\cap A''$ contains $A^1$; let $a_1^2,\ldots,a_d^2$ be a basis
of $A'\cap A''$ relative to $A^1$,
and for each $i=1,\ldots,d$, choose a $b_i^2\in B$ and a $c_i^2\in C$
such that $a_i^2+b_i^2$ and $a_i^2+c_i^2$ lie in $Z$.
Let us prove that the images of $b_1^2,\ldots,b_d^2$ in $B/B^1$
are linearly independent:
if not then some
{\mygreen non-trivial}
linear combination, $\sum_i \beta_i b_i^2$, 
of the $b_i^2$ lies
$B^1$, and hence the corresponding 
{\mygreen (non-trivial)}
linear combination,
$a=\sum_i\beta_i a_i^2$, of the $a_i^2$
satisfies $a=a+0\in Z$ in view of the fact that
$$
a = a + 0 =
\sum_i \beta_i(a_i^2+b_i^2)\in Z.
$$
But then $a\in Z$, and so $a\in A^1$,
contradicting the fact that $a_1^2,\ldots,a_d^2$ is a basis
relative to $A^1$.
Similarly the $c_1^2,\ldots,c_d^2$ are linearly independent in
$C/C^1$.
Let $A^2$ be the span of the $a_i^2$, and similarly for $B^2$ and 
$C^2$.  Then $B^2$ is linearly independent from $B^1$, by the above
argument, and similarly for $C^2$ and $C^1$;
by definition $A^2$ is linearly independent from $A^1$.

By the definition of $A',A''$ and of $A^1$, an element $a\in A$ is
both $B$- and $C$-pairable iff $a\in A'\cap A''=A^1+A^2$.  
Let us prove analogous statement holds with $A,B,C$ exchanged:
to start, let us prove that if $b\in B$ is both $A$- and $C$-pairable,
then $b\in B^1+B^2$: for any such $b$ there are $a\in A$ and
$c\in C$ such that $a+b,b+c\in Z$, and hence $a-c\in Z$; 
hence $a$ is both $B$- and $C$-pairable, and hence 
$a\in A'\cap A''=A_1+A_2$, so we may write $a=a^1+a^2$
with $a^i\in A^i$ for $i=1,2$.
Then there exists $b^2\in B^2$ such that $a^2$ is paired with $b^2\in B^2$
(by expressing $a^2$ as a linear combination of the $a_i^2$ and taking
the corresponding linear combination of the $b_i^2$);
hence $a+b^2=(a^2+b^2)+a^1\in Z$.  Hence $a$ is $B$-paired with
$b^2\in B$; since $a$ is also $B$-paired with $b$, by
\eqref{eq_pairable_uniqueness}
we have $b=b^2+b^1$ for some $b^1\in B^1$.
Hence $b\in B^1+B^2$.

Conversely, if $b\in B^1+B^2$, then $b=b^1+b^2$ with $b^i\in B^i$
and $b^1$ is both $A$- and $C$-pairable (paired with $0$ in both cases).
{\mygreen An argument similar to that in the previous paragraph}
(i.e., writing $b^2$ as a linear combination of $b_i^2$ and considering
the analogous combination of the $a_i^2$ and $c_i^2$) shows that
$b^2$ is both $A$- and $C$-pairable.
This establishes that $b\in B$ is both $A$- and $C$-pairable
iff $b\in B^1+B^2$.

Similarly, $c\in C$ is both $A$- and $B$-pairable iff
$c\in C^1+C^2$.

To summarize the above, we have shown the existence of $A^2,B^2,C^2$
with 
$A^1,A^2\subset A$ independent, as well as $B^1+B^2\subset B$ and
$C^1,C^2\subset C$, such that an element of $A$ is both $B$- and
$C$-pairable iff it lies in $A^1+A^2$, and similarly with $A,B,C$ exchanged.
Let us now construct $A^3,A^4,B^3,B^4,C^3,C^4$ with the desired 
properties.

Pick a basis, $a^3_1,\ldots,a^3_s$
of $A'$ relative to $A'\cap A''=A^1+A^2$, and let $A^3$ be the span of this
relative basis; similarly for 
$a^4_1,\ldots,a^4_t$, of $A''$ relative to $A'\cap A''$ and for $A^4$; by the
dimension formula $A^1+A^2$, $A^3$, $A^4$ are linearly independent.
For each $a^3_i$, choose a $b^3_i$ such that $a^3_i+b^3_i\in Z$,
and similarly for each $a^4_i$ and $c^3_i$.
We claim that the images of the 
$b^3_i$ in $B/(B^1+B^2)$ are linearly independent:
for otherwise some 
{\mygreen non-trivial}
linear combination of the $b^3_i$ would vanish,
and the corresponding linear combination of the
$a^3_i$, say $a$, would then satisfy
$a=a+0\in Z$; but this contradicts the definition of
the relative basis $a^3_1,\ldots,a^3_s$.
The symmetric argument shows that 
the images of the $c^3_i$ are linearly independent in $C/(C^1+C^2)$.
Let $B^3$ be the span of the $b^3_i$, and $C^3$ that of the $c^3_i$.

Next consider the subspace, $\tilde B$, of all $C$-pairable elements of $B$, 
which clearly contains
$B^1,B^2$; let $b^4_1,\ldots,b^4_p$ be a basis of $\tilde B$ relative
to $B^1+B^2$.  For each $i$, choose an element $c^4_i$ such that
$b^4_i+c^4_i\in Z$; let $B^4$ be the span of all $b^4_i$, and
$C^4$ those of the $c^4_i$.

Due to the asymmetry in our definition, we have a lot of knowledge
about
$A^1,\ldots,A^4$, namely:
\begin{enumerate}
\item
$A^1,\ldots,A^4\subset A$ are linearly independent;
\item 
the subspace of $A$ of elements that are $B$-pairable equals
$A^1+A^2+A^3$;
\item 
the subspace of $A$ of elements that are $C$-pairable equals
$A^1+A^2+A^4$.
\end{enumerate}
We now wish to prove the analogous claims about the $B^j$ and $B$,
and the $C^j$'s and $C$.

Let us start by giving proofs of analogous statements with the
$B^j$ and $B$.
\begin{enumerate}
\item
$B^1,B^2$ are linearly independent: shown above.
\item
$B^1,B^2,B^3$ are linearly independent:
if not, then we have $b^3=b^2+b^1$ for some nonzero $b^3\in B^3$ and 
$b^1\in B^1$, $b^2\in B^2$ (since $B^1,B^2$ are linearly independent).
But then we have $a^2\in A^2$ and $c^2\in C^2$ such that
$a^2+b^2,b^2+c^2\in Z$, and hence $a^2-c^2\in Z$;
also there exists a nonzero $a^3$ such that
$b^3+a^3\in Z$ (obtained by writing $b^3$ as a linear combination
of the $b^3_i$ and forming the analogous linear combination of the $a_i^3$).
But then $a=a^3$ is both $B$-pairable (with $b^3$), and we now
check that it is $C$-pairable with $-c^2$, since:
$$
a-c^2 = a-c^2 + (b^3 - b^2-b^1) = (a+b^3)-(c^2+b^2)-(b^1)
$$
and $a+b^3$, $c^2+b^2$, and $b^1$ all lie in $Z$.
But the fact that $a=a^3$ is nonzero, both $B$- and $C$-pairable,
but not in $A^1+A^2=A'\cap A''$ 
contradicts the definition of $A'$ and $A''$.
\item
$B^1,B^2,B^3,B^4$ are linearly independent:
if not, then we have $b^4=b^3+b^2+b^1$ for some $b^i\in B^i$ with
$b^4$ nonzero
(since $B^1,B^2,B^3$ are linearly independent).
Since the $b^4_i$ is a basis relative to $B^1,B^2$
of $B'$, we have $b^3\ne 0$.
To $b^4$ there is a corresponding linear combination, $c^4$, of the
$c_i^4$, such that $b^4+c^4\in Z$, and similarly for $b^3$ and $a^3$
with $b^3+a^3\in Z$; since $b^3\ne 0$ also $a^3\ne 0$;
similarly $b^2$ has a corresponding $a^2$ and $c^2$
such that both $b^2+a^2$ and $b^2+c^2$ lie in $Z$.
Since
$$
b^4+c^4,
\ b^3+a^3, 
\ b^2+a^2,
\ b^2+c^2,
\ b^1
$$
all lie in $Z$, we have that $b^3$ is a $B$ element that is
$A$-pairable (since $b^3+a^3\in Z$), and also
$$
b^3+(c^4-c^2)=b^4-b^2-b^1+c^4-c^2
=
(b^4+c^4)-(b^2+c^2)-b^1 \in Z.
$$
Hence if $a=a^3$, $b=b^3$, and $c=c^4-c^2$, then
$a+b\in Z$ and $b+c\in Z$ and hence $a-c\in Z$.
Hence $a$ is both $B$- and $C$-pairable.  But $a=a^3\notin A^1+A^2$,
since $a^3$ is a nonzero element of $A^3$,
which is a contradiction.
\item
If $b\in B$ is $A$-pairable, then
$b\in B^1+B^2+B^3$ (clearly the converse holds):
if $b$ is $A$-pairable with $a$, then $a$ is $B$-pairable and hence
$a=a^3+a^2+a^1$, and there exist $b^3,b^2$ in $B^3,B^2$ respectively
such that $a^i+b^i\in Z$ for $i=2,3$.
But then $a$ is $B$-pairable by $b^3+b^2$, and since it is $B$-pairable
by $b$ as well,
\eqref{eq_pairable_uniqueness} implies that $b$ equals some element
of $B^1$ plus $b^3+b^2$, and hence 
$b\in B^1+B^2+B^3$.
\item
If $b$ is $C$-pairable, then $b\in B^1+B^2+B^4$ (the converse clearly holds):
this follows from the definition of $B^4$.
\end{enumerate}

Next we address the same issues with the $C^j$ and $C$.
\begin{enumerate}
\item 
$C^1,C^2$ are linearly independent: proven above.
\item 
$C^1,C^2,C^3$ are linearly independent: one argues just as for $B^1,B^2,B^3$.
\item
$c^4_1,\ldots,c^4_p$ are linearly independent:
if not, then some nontrivial linear combination of them is zero,
and, $b$, the corresponding linear combination of $b^4_i$, has
$b+0\in Z$.  Then $B^4\cap B^1$ is nonzero, contracting the independence
of $B^1$ and $B^4$.
\item 
$C^1,C^2,C^3,C^4$ are linearly independent: 
any nonzero element
$c^4\in C^4$ has a corresponding nonzero $b^4\in B^4$ with $c^4+b^4\in Z$.
Since $c^4,c^2,c^1$ can be $B$-paired, so can
$c^3=c^4-c^2-c^1$; hence $c^3$ can be $B$-paired, say with $b'$, 
and by definition
any $c^3\in C^3$ can be $A$-paired, say with $a'$; hence 
$$
c^3+b',
\ c^3+a',
\ a'-b'
$$
all lie in $Z$, and hence $a'$ can be both $B$- and $C$-paired and
hence $a'=a^2+a^1$ for $a^1,a^2$ respectively in $A^1,A^2$; 
since $a^2$ can be $C$-paired with some $\tilde c^2\in C^2$, we
have $(a^2+\tilde c^2)+a_1\in Z$, and hence $c^3+\tilde c^2\in Z$, and so
$c^3+\tilde c^2\in Z\cap C=C^1$.  Since $C^1,C^2,C^3$ are linearly 
independent, it follows that $c^3=0$ (and $\tilde c^2=0$).
Hence $c^4=c^2+c^1$.  
To the nonzero linear combination of the $c^4_i$ that give $c^4$, there
corresponds a linear combination of the $b^4_i$, $b^4\in B^4$, such that
$b^4\ne 0$ (since $c^4\ne 0$)
and $c^4+b^4\in Z$; but since $c^2,c^1$ are both $A$- and
$B$-pairable, so is $c^4$, and hence for some $b'',a''$ we have
$c^4+b'',c^4+a'',b''-a''$ lie in $Z$, and hence---subtracting $b^4+c^4\in Z$,
also $-b^4+a''\in Z$.  Hence $b^4$ is $A$- and $C$-pairable,
and hence, $b^4\in B^1+B^2$ (shown in the last paragraphs).  
But this contradicts that $b^4\ne 0$
and $b^4\in B^4$.
\item 
If $c\in C$ is $A$-pairable, then $c\in C^1+C^2+C^3$ (the converse clearly
holds):
same proof as for $b\in B^1+B^2+B^3$.
\item 
If $c\in C$ is $B$-pairable, then $c\in C^1+C^2+C^4$ (the converse clearly
holds):
say that $b+c\in Z$ with $b\in B$.
Then $b$ is $C$-pairable, as shown above; hence $b=b^4+b^2+b^1$
with $b^i\in B^i$ for $i=1,2,4$,
and to $b^4$ and $b^2$ correspond $c^4$ and $c^2$
with $b^i+c^i\in Z$ for $i=2,4$.  Hence $b$ is pairable with $c^4+c^2$,
and applying
\eqref{eq_pairable_uniqueness} with all occurrences of $b,B$ replaced
with $c,C$, we see that $c$ equals $c^4+c^2$ plus some element of
$C^1$.  Hence $c\in C^1+C^2+C^4$.
\end{enumerate}

Finally we construct $A^5,B^5,C^5$:
to do so, consider the subset $\tilde A$ of $a\in A$ such that
$a+b+c\in Z$ for some $b\in B$ and $c\in C$.
Clearly $\tilde A$ is a subspace, and clearly it contains
$A^1,\ldots,A^4$;
let $a_1^5,\ldots,a_q^5$ be a basis of $\tilde A$ relative
to $A^1+A^2+A^3+A^4$, and for each $i=1,\ldots,q$ choose
$b^5_i$ and $c^5_i$ such that
$a^5_i+b^5_i+c^5_i\in Z$.
Let $A^5,B^5,C^5$, respectively, be the spans of the $a^5_i$, the $b^5_i$,
and the $c^5_i$.
We prove the following claims, all with ideas similar to the ideas above.
\begin{enumerate}
\item $A^1,\ldots,A^5$ are linearly independent:
immediate from the definition of $a^5_i$.
\item
If $a+b+c\in Z$ for some $a\in A$, $b\in B$, $c\in C$, then
$a\in A^1+\cdots+A^5$: clear from the definition of $A'$ and $A^5$ above.
\item
If $a+b+c\in Z$ for some $a\in A$, $b\in B$, $c\in C$,
and $a=a^1+\cdots+a^5$ with $a^i\in A^i$ and $a^5\ne 0$,
then $b\notin B^1+B^2+B^3+B^4$:
otherwise $b=b^1+\cdots+b^4$ with $b^i\in B^i$; in this case
we have $c^4\in C^4$ such that $b^4+c^4\in Z$
and $\tilde a^i\in A^i$ for $i=2,3$ with $\tilde a^i+b^i\in Z$.  
It follows that
$$
(a^1+\cdots+a^5+\tilde a^2+\tilde a^3)+(b^2+b^3+b^4)+c^4 \in Z,
$$
and hence
$$
(a^1+\cdots+a^5+\tilde a^2+\tilde a^3)+(b^2+b^3)\in Z,
$$
and hence 
$$
a^1+\cdots+a^5+\tilde a^2+\tilde a^3 
$$
is $B$-pairable and therefore lies in $A^1+A^2+A^3$.  But this
is impossible since $a^5\ne 0$ and $A^5$ is linearly independent
from $A^1,\ldots,A^4$.
\item
If $a+b+c\in Z$ for some $a\in A$, $b\in B$, $c\in C$,
and $a=a^1+\cdots+a^5$ with $a^i\in A^i$ and $a^5\ne 0$,
then $c\notin C^1+C^2+C^3+C^4$:
the same argument as above with $b,B$'s and $c,C$'s exchanged.
\item
The $b^5_i$ are linearly independent:
if not, then some linear combination of the $b^5_i$ equals zero, and
the corresponding linear combination, $a^5$, of the $a^5_i$, 
and $c^5$, that of the $c^5_i$, have
$a^5,c^5\ne 0$ and $a^5+c^5\in Z$.  
But then $a^5$ is $C$-pairable and must lie in $A^1+A^2+A^4$, which
contradicts the independence of $A^1,\ldots, A^5$ proven above.
\item
$B^1,\ldots,B^5$ are linearly independent: if they are dependent,
then we have $b^5=b^1+\cdots+b^4$ with $b^i\in B^i$ and not all
$b^i$ being zero; since $B^1,\ldots,B^4$ are
linearly independent, we have $b^5\ne 0$.
But then for all $i\ne 4$,
there are corresponding linear combinations $a^i\in A^i$ to
$b^i$ such that $a^5\ne 0$ and $a^i+b^i\in Z$; also
there is a $c^4\in C^4$ with $b^4$ with $c^4+b^4\in Z$.
Hence setting $a=a^5+a^3+a^2$ we have
$$
(a^5+a^3+a^2)+ (b^2+b^3+b^4)+c^4 = (a^2+b^2)+(a^3+b^3)+(b^4+c^4)\in Z.
$$
But since $a^5\ne 0$, this contradicts~(4) above,
since $c^4\in C^4\subset C^1+C^2+C^3+C^4$.
\item
The $c^5_i$ are linearly independent, and $C^1,\ldots,C^5$ are 
linearly independent:
the same argument as above with $b,B$'s and $c,C$'s exchanged.
\item
If $a+b+c\in Z$ for some $a\in A$, $b\in B$, $c\in C$,
then $b\in B^1+\cdots+B^5$:
we have $a\in A^1+\cdots+A^5$, by definition of the $A^i$, and hence
for some $a^i\in A^i$ we have $a=a^1+\cdots+a^5$.
It follows that $a^5+b^5+c^5\in Z$ for the corresponding $b^5,c^5$
to $a^5$, and $a^2+b^2\in Z$ for some $b^2\in B^2$, and similarly
for $a^3+b^3$ and $a^4+c^3$.  It follows that
$$
(-a-b-c)+
a^1+(a^2+b^2)+(a^3+b^3)+(a^4+c^3)+(a^5+b^5+c^5) \in Z,
$$
and hence
$$
(-b+b^3+b^5)+(-c+c^3+c^5) \in Z.
$$
It follows that $-b+b^3+b^5$ is $C$-pairable, and hence 
$-b+b^3+b^5\in B^1+B^2+B^4$.
Hence $b\in B^1+\cdots+B^5$.
\item
If $a+b+c\in Z$ for some $a\in A$, $b\in B$, $c\in C$,
then $c\in C^1+\cdots+C^5$:
the same argument as above with $b,B$'s and $c,C$'s exchanged.
\end{enumerate}

At this point we claim that $Z$ consists precisely of the
sums given in the above theorem.
Namely, say that $a+b+c\in Z$ with $a\in A$, $b\in B$, $c\in C$.
Then write $a$ as $a^1+\cdots+a^5$.
Corresponding to $a^5\in A^5$ there are $b^5\in B^5$ and
$c^5\in C^5$ with
$a^5+b^5+c^5\in Z$, $c^3$ corresponding to $a^4$ with $a^4+c^3\in Z$,
$b^3$ corresponding to $a^3$ with $a^3+b^3\in Z$,
and $a^2$ corresponding to $b^2$ with $a^2+b^2\in Z$.
Hence, setting $\tilde b=-b+b^5+b^3$ and $\tilde c=-c+c^5+c^3$ we have
$$
-a^1+\tilde b+\tilde c\in Z,
$$
and hence $\tilde b+\tilde c\in Z$.  It follows $\tilde b$ is $C$-pairable and
hence $\tilde b=\tilde b^1+\tilde b^2+\tilde b^4$ with
$\tilde b^i\in B^i$, and hence
$$
-a^1+(\tilde b^1+\tilde b^2+\tilde b^4)+\tilde c\in Z.
$$
Corresponding to $\tilde b^2,\tilde b^4$ there are
$\tilde c^2\in C^2$ and $\tilde c^4\in C^4$ with
$\tilde b^i+\tilde c^i\in Z$, and hence
$$
-a^1+\tilde b^1-\tilde c^2-\tilde c^4+\tilde c\in Z,
$$
and hence
$$
-\tilde c^2-\tilde c^4+\tilde c\in Z,
$$
and hence
$$
-\tilde c^2-\tilde c^4+\tilde c\in Z\cap C=C^1
$$
and hence
$$
-\tilde c^2-\tilde c^4+\tilde c = c^1
$$
for some $c^1\in C^1$.
Similarly
$$
-\tilde b^2-\tilde b^4+\tilde b = b^1
$$
for some $b^1\in B^1$.
%
%
Hence
\begin{align*}
a & =  a^1+\cdots+a^5, \\
b & =  b^5+b^3-\tilde b = b^5+b^3-\tilde b^2-\tilde b^4-b^1, \\
c & =  c^5+c^3-\tilde c=c^5+c^3-\tilde c^2-\tilde c^4-c^1.
\end{align*}
Hence any triple $(a,b,c)\in \cU$ with
$a+b+c\in Z$ satisfies $a\in\sum_i A^i$, and similarly for $b$ and $c$.
Since the $A^1,\ldots,A^5$ are linearly independent, and similarly for
the $B^i$'s and $C^i$'s, the decomposition of any such triple
$(a,b,c)$ as such is unique.
\end{proof}

\subsection{The Decomposition Lemma for Two and Four or More Subspaces}

Taking $C=0$ in the decomposition lemma, we see that if $Z$ is
a linear subspace of a universe that has a decomposition into
subspaces $A,B$, then
$Z$ is spanned by a sum of $A^1$, $B^1$, and $A^2\oplus B^2$
with $A^1,B^1,A^2,B^2$ independent.

To study linear coded caching schemes with $N\ge 4$ it could be useful to generalize the decomposition lemma
when the universe has a decomposition into $N\ge 4$ parts.
At present, it is not clear to us what is the correct statement
of such a lemma, even for $N=4$.

\section{A New Coded Caching Scheme with $N=K=3$: $(1/2,5/3)$ is achievable}
\label{se_Joel_new_point}

For the case of $N=K=3$, one can use pure symmetric two-way schemes
in Definition~\ref{de_Z_schemes} to
achieve the memory-rate trade-off $(1/2,5/3)$,
which is an $\field$-linear scheme with $\field=\integers/2\integers$.
In this section we describe the scheme: the $Z_1,Z_2,Z_3$ are
chosen to be pure symmetric two-way schemes,
separated, and otherwise as large as possibly; this determines
them.
We found it a bit more difficult to
find appropriate values of $X_{\mec d}$; we will describe how
we found their values.

We emphasize that for a given $N,K,F,\field$ 
and values of $Z_1,Z_2,Z_3$, we 
do not know of a good algorithm for finding the smallest possible
values of $\dim(X_{\mec d})$ for an arbitrary $\mec d$.

We let $F=6$, we take $W_1=\field^6$ with
$\field=\integers/2\integers$, and decompose $W_1$ into three two 
dimensional subspaces $A_1,A_2,A_3$;
we let
$a_j',a_j''$ be an arbitrary basis for $A_j$.
We do similarly for $W_2=B_1+B_2+B_3$ and $b_j',b_j''$,
and similarly for $W_3=C_1+C_2+C_3$ and $c_j',c_j''$.
Then for $j=1,2,3$ we set
$$
Z_1={\rm Span}(a_j'\oplus b_j',a_j''\oplus c_j',b_j''\oplus c_j'').
$$
Hence each $Z_i$ is of dimension $3=F/2=FM$ with $M=1/2$;

First we claim that there exists an $X_{123}$ of size $5F/3$ such that
for $j=1,2,3$, $X_{123}$ and $Z_j$ determines $W_j$.
To prove this, it is easiest to discuss our method for building $X_{123}$.
Since $Z_1$ along with
$X_{123}$ needs to learn $A_1,A_2,A_3$ for a total of $6(F/6)$
bits of information, and similarly for $Z_2$ and $Z_3$,
it is simplest to build $X_{123}$ by first adding all unencoded bits that
are useful to at least two of $Z_1,Z_2,Z_3$.
We easily see that each bit of 
\begin{equation}\label{eq_pure_X_one_two_three}
a_2',a_3'',b_1',b_3'',c_1',c_2''
\end{equation} 
is useful to two of $Z_1,Z_2,Z_3$, and after adding these bits
each of $Z_1,Z_2,Z_3$ has two bits of information it needs.
Next,
we know that $Z_1$ will need to make use of
$b_1''\oplus c_1''$, since if it does not then we can discard
$b_1''\oplus c_1''$ from $Z_1$, in which case it only needs 
$FM=F/3$ of its information, in which case $R\ge 2$ and hence we
cannot hope to achieve $R=5/3$; 
hence our
strategy is to add $b_1''\oplus c_1''$ to each missing part of
the information $Z_1$ needs in $W_1=A_1,A_2,A_3$, namely $a_2''$
and $a_3'$, each added to $b_1''\oplus c_1''$; this process suggests
that we add to $X_{123}$ the vectors
$$
a_2''\oplus b_1''\oplus c_1'', \quad a_3'\oplus b_1''\oplus c_1'',
$$
and similarly for $Z_2$ the vectors
$$
a_2''\oplus b_1''\oplus c_2', \quad a_2''\oplus b_3'\oplus c_2'
$$
and for $Z_3$
$$
a_3'\oplus b_3'\oplus c_1'', \quad a_3'\oplus b_3'\oplus c_2'.
$$
But the six vectors displayed above come in pairs whose
sum is the sum of the six bits not used in
\eqref{eq_pure_X_one_two_three}, namely
$$
a_2''\oplus a_3' \oplus b_1''\oplus b_3'\oplus c_1''\oplus c_2' .
$$
Hence the dimension of the span of these six vectors is $4$-dimensional.
Adding to this $4$-dimensional vector space the span
of the vectors in
\eqref{eq_pure_X_one_two_three} gives a $10$-dimensional subspace,
therefore giving an $X_{123}$ with dimension $10=(5/3)F$.

By symmetry, there is an $X_{d_1,d_2,d_3}$ of $5F/3$ bits
for any $d_1,d_2,d_3$ distinct that for each $j$ allows $Z_j$ to
infer $W_{d_j}$.

Furthermore, by symmetry it suffices to describe values for
$X_{111}$ and $X_{112}$ of at most $5F/3=10$ bits.
Of course, for $X_{111}$ it suffices to broadcast $W_1$, which
requires $F\le 5F/3$ bits.
Next we describe $X_{112}$:
we use only the fact that $Z_j$ knows $a_j'\oplus b_j'$,
and first put into $X_{112}$ all the $''$ bits, namely
$$
a_1'',a_2'',a_3'',b_1'',b_2'',b_3'' ,
$$
so the $Z_i$ only need infer the correct $'$ bits;
then add to $X_{112}$ the bits 
$$
b_1',b_2',a_3';
$$
at this point we see that
$Z_3$ can infer all of $B$, and that $Z_1$ has $a_1',a_3'$
and $Z_2$ has $a_2',a_3'$; hence we add the 10th bit 
$$
a_1'\oplus a_2'
$$
to $X_{112}$, which allows both $Z_1,Z_2$ to infer their last $'$ bit
of $W_1=A$ that each needs.

It follows that the $X_{\mec d}$ given above all have dimension
at most $5F/3$, which shows that the memory-rate trade-off
$(M,R)=(1/2,5/3)$ can be achieved.

\section{Coded Caching with $N=K=3$: Two Discoordination Bounds}
\label{se_Joel_N_K_3_discoord}

The point of this section is to prove some lower bounds 
(i.e., ``outer bounds'')
on linear coded
caching schemes that attain a memory-rate pair $(M,R)$ in the case
$N=K=3$.
We will get two bounds that are interesting and involve
the discoordination of certain random variables.
We will use the second bound to prove
$6M+5R\ge 11$ for separated linear schemes
in Section~\ref{se_Joel_hybrid}, and
general linear schemes 
in Section~\ref{se_Joel_hybrid_not_sep}.

\subsection{The Main Bounds for Linear Codes for $N=K=3$}

In this subsection we formally describe
the two main bounds involving linear schemes for
$N=K=3$ for coded caching,
plus a lemma that seems interesting in its own right.
In all cases, the inequalities can be turned into equalities by
going through the proofs and keeping certain
non-negative terms that we discard along the way.

\begin{theorem}\label{th_first_discoordination_bound}
Consider a coded caching scheme with notation as in
Subsection~\ref{su_coded_caching_notation} for $N=K=3$ 
and where $F$ is arbitrary.
Let
$$
P_1=X_{123}+Z_1, \quad
P_2=X_{213}+Z_2.
$$
Then
\begin{equation}\label{eq_first_discoordination_bound}
2R+3M \ge 5 + 
\dim^{\cU/W_1W_2}\bigl( [(P_1+P_2) \cap Z_3]_{W_1W_2} \bigr) 
+ \dim^\cU(W_2\cap Z_3) + s_1+s_2 - \delta,
\end{equation} 
where $\delta$ is the discoordination 
\begin{equation}\label{eq_delta}
\delta =  {\rm DisCoord}^{\cU/W_1}([P_1]_{W_1},[P_2]_{W_1},[Z_3]_{W_1})
\end{equation} 
and $s_1,s_2$ are terms that vanish under
symmetrization or for symmetric schemes, namely
\begin{equation}\label{eq_sym_cancel_terms}
s_1 = \dim(W_1W_3Z_3)-\dim(W_1W_2Z_3), \quad
s_2 = \dim(W_1\cap Z_3)-\dim(W_2\cap Z_3).
\end{equation} 
\end{theorem}
We warn the reader that the above term
$\dim^{\cU/W_1W_2}\bigl( (P_1+P_2) \cap Z_3 \bigr)$,
we first calculate the intersection 
$(P_1+P_2) \cap Z_3$ in $\cU$,
and then consider the image of this intersection in $\cU/W_1W_2$;
in the optimal scheme for $M=1/3$, this dimension is $0$,
whereas the dimension of $[P_1+P_2]\cap[Z_3]$ in $\cU/W_1W_2$
equals $1/3$.

To prove this bound, we will first prove Lemma~\ref{le_first_step}
below,
which seems interesting in its own right.

The second main theorem is a bound that involves the discoordination
of $W_1,W_2,Z_3$; since $W_1$ and $W_2$ are independent, we have
$W_1\cap W_2=0$, and hence
$$
S_2(W_1,W_2,Z_3) 
=W_1\cap Z_3+W_2\cap Z_3 
$$ 
and hence the discoordination of
$W_1,W_2,Z_3$ equals
$$
\delta' = {\rm DisCoord}^\cU(W_1,W_2,Z_3)
=\dim^{\cU/(W_1\cap Z+W_2\cap Z)}\bigl( (W_1+W_2)\cap Z_3 \bigr) 
$$
and since $S_2=W_1\cap Z+W_2\cap Z$ lies in both $W_1+W_2$ and $Z$,
this equals
$$
=\dim\bigl( (W_1+W_2)\cap Z_3 ) - \dim(W_1\cap Z_3+W_2\cap Z_3),
$$
$$
=\dim\bigl( (W_1+W_2)\cap Z_3 ) - \dim(W_1\cap Z_3)-\dim(W_2\cap Z_3)
+ \dim\bigl( (W_1\cap Z_3)\cap(W_2\cap Z_3) \bigr)
$$
and since $W_1\cap W_2=0$, the rightmost term equals zero.
Hence
$$
\delta' = {\rm DisCoord}^\cU(W_1,W_2,Z_3)
=\dim\bigl( (W_1+W_2)\cap Z_3 ) - \dim(W_1\cap Z_3)-\dim(W_2\cap Z_3).
$$

\begin{theorem}\label{th_second_discoordination_bound}
Consider a coded caching scheme with notation as in
Subsection~\ref{su_coded_caching_notation} for $N=K=3$ 
and where $F$ is arbitrary.
Then
\begin{equation}\label{eq_second_discoordination_bound}
2R+3M \ge 5 - s - \delta',
\end{equation} 
where $s$ is a term that vanishes after symmetrization, and $\delta'$ is
the average discoordination of $W_1,W_2,Z_3$ in $\cU$,
$$
\delta' = {\rm DisCoord}^{\cU,{\rm avg}}(W_1,W_2,Z_3).
$$
\end{theorem}

After we prove Theorem~\ref{th_second_discoordination_bound},
we remark that \eqref{eq_second_discoordination_bound}
remains valid if we add the following non-negative term to the
right-hand-side of
\eqref{eq_second_discoordination_bound}:
$$
\dim^{\cU}\bigl( Z_3\cap S_2 \bigr)
-
\dim^{\cU}\bigl( Z_3\cap W_1 \bigr) .
$$
We suspect that by better understanding this term we could
improve upon 
Theorem~\ref{th_second_discoordination_bound}.


\subsection{Proof of the First Discoordination Bound}

We organize the proof of the discoordination bounds
into a few lemmas.

\begin{lemma}\label{le_zeroth_step}
Consider a coded caching scheme with notation as in
Subsection~\ref{su_coded_caching_notation} for $N=K=3$ 
and where $F$ is arbitrary.
Then setting
$$
P_1=X_{123}+Z_1, \quad
P_2=X_{213}+Z_2
$$
we have
\begin{equation}\label{eq_zeroth_step}
2R+3M \ge \dim(P_1)+\dim(P_2) + \dim(Z_3).
\end{equation} 
\end{lemma}
\begin{proof}
We have
$$
2R+3M \ge \dim^{\cU}(X_{123})+\dim^{\cU}(Z_1) 
+ \dim^{\cU}(X_{213}) + \dim^{\cU}(Z_2)+ \dim^{\cU}(Z_3).
$$
By the dimension formula
$$
\dim^{\cU}(X_{123})+\dim(Z_1) \ge \dim^{\cU}(X_{123}+Z_1)
= \dim^{\cU}(P_1);
$$
similarly 
$$
\dim^{\cU}(X_{213})+\dim(Z_2) \ge \dim(P_2).
$$
Combining the three equations displayed above yields the lemma.
\end{proof}
We remark that \eqref{eq_zeroth_step} would hold with equality if
we add $\dim^{\cU}(X_{123}\cap Z_1)$ and 
$\dim^{\cU}(X_{213}\cap Z_2)$ to the right-hand-side.

\begin{lemma}\label{le_first_step}
Consider the hypothesis and notation of Lemma~\ref{le_zeroth_step}.
Then
\begin{equation}\label{eq_first_step}
2R+3M \ge 4 +  \dim^{\cU/W_1}([P_1]_{W_1}\cap [P_2]_{W_1}) +
\dim^{\cU}\bigl( (P_1+P_2) \cap Z_3 \bigr) .
\end{equation} 
\end{lemma}
\begin{proof}
By the dimension formula,
$$
\dim^\cU(P_1)+\dim^\cU(P_2) = \dim^\cU(P_1+P_2) +  \dim^\cU(P_1\cap P_2),
$$
and hence the right-hand-side of \eqref{eq_zeroth_step} can be written as
$$
2R+3M \ge 
\dim^{\cU}(P_1+P_2+Z_3) + \dim^{\cU}\bigl( (P_1+P_2) \cap Z_3 \bigr)
+  \dim^{\cU}(P_1\cap P_2).
$$
But $P_1+P_2+Z_3$ implies $X_{123},Z_1,Z_2,Z_3$ whose sum is all of
$\cU$.  Hence 
$$
2R+3M \ge 
3 + \dim^{\cU}\bigl( (P_1+P_2) \cap Z_3 \bigr)
+  \dim^{\cU}(P_1\cap P_2).
$$
Also $P_1$ and $P_2$ both imply $W_1$, and
\eqref{eq_first_step} follows.
\end{proof}

Next we study the first term on the right-hand-side of 
\eqref{eq_first_step}.

\begin{lemma}\label{le_first_term_in_first_step}
Consider the hypothesis and notation of Lemma~\ref{le_zeroth_step}.
Then
$$
\dim^{\cU/W_1}([P_1]_{W_1}\cap [P_2]_{W_1}) \ge
\dim^{\cU/W_1Z_3}([W_3]_{W_1Z_3}) + t_1 + t_2 - \delta,
$$
where
\begin{equation}\label{eq_delta_again}
\delta =  {\rm DisCoord}^{\cU/W_1}([P_1]_{W_1},[P_2]_{W_1},[Z_3]_{W_1})
\end{equation} 
and $t_1,t_2$ are the non-negative terms
$$
t_1 = \dim^{\cU/W_1}([P_1]_{W_1}\cap [P_2]_{W_1} \cap [Z_3]_{W_1}),
\quad
t_2 = 
\dim^{\cU/W_1W_3 Z_3}([P_1]_{W_1Z_3}\cap [P_2]_{W_1Z_3}) .
$$
In particular,
\begin{equation}\label{eq_first_term_first_step}
\dim^{\cU/W_1}([P_1]_{W_1}\cap [P_2]_{W_1}) 
\ge
\dim^{\cU/W_1Z_3}([W_3]_{W_1Z_3}) - \delta.
\end{equation}
\end{lemma}
Note that since $P_1$ and $P_2$ both imply $W_1$, 
$P_i$ already equals all of $[P_i]_{W_1}=P_i+W_1$.
\begin{proof}
Consider the universe $\cU/W_1$ and its three linear subspaces
$[P_1],[P_2],[Z_3]$;
Corollary~\ref{co_various_formulas_for_discoordination_of_three}
(with $C=[Z_3]$ there) implies
\begin{align*}
&\dim^{\cU/W_1}([P_1]_{W_1}\cap [P_2]_{W_1}) \\
&=\dim^{\cU/W_1Z_3}([P_1]_{W_1Z_3}\cap [P_2]_{W_1Z_3})
+
\dim^{\cU/W_1}([P_1]_{W_1}\cap [P_2]_{W_1} \cap [Z_3]_{W_1})
- \delta \\
&= \dim^{\cU/W_1Z_3}([P_1]_{W_1Z_3}\cap [P_2]_{W_1Z_3})
+
t_1
- \delta;
\end{align*}
since $P_i+Z_3$ implies $W_3$ (and $W_1$) we have
$$
\dim^{\cU/W_1Z_3}([P_1]_{W_1Z_3}\cap [P_2]_{W_1Z_3})
=
\dim^{\cU/W_1Z_3}([W_3]_{W_1Z_3}) + t_2.
$$
Hence
$$
\dim^{\cU/W_1}([P_1]_{W_1}\cap [P_2]_{W_1}) =
\dim^{\cU/W_1Z_3}([W_3]_{W_1Z_3}) + t_1 + t_2 - \delta.
$$
\end{proof}

\begin{proof}[Proof of Theorem~\ref{th_first_discoordination_bound}]
Consider the second term on the right-hand-side of
\eqref{eq_first_step}:
since $W_1,W_2$ are both implied by $P_1+P_2$ (since this contains
$X_{123},Z_1,Z_2$), we have
$$
\dim^{\cU}\bigl( (P_1+P_2) \cap Z_3 \bigr)
=
\dim^{\cU}\bigl( (W_1W_2) \cap Z_3 \bigr)
+
\dim^{\cU/W_1W_2}\bigl( [(P_1+P_2) \cap Z_3]_{W_1W_2} \bigr) .
$$
Using this formula and Lemma~\ref{le_first_term_in_first_step},
and applying these to the right-hand-side of \eqref{eq_first_step},
we get
\begin{equation}\label{eq_upshot_before_wwz_calc}
\begin{split}
2R+3M \ge& 
4 + 
\dim^{\cU/W_1Z_3}([W_3]_{W_1Z_3}) 
+
\dim^{\cU}\bigl( (W_1W_2) \cap Z_3 \bigr)\\
&+
\dim^{\cU/W_1W_2}\bigl( [(P_1+P_2) \cap Z_3]_{W_1W_2} \bigr) - \delta .    
\end{split}
\end{equation} 

Now we take two of the terms above and notice the following simplification
(modulo the term $s_1$, which drops out upon symmetrization):
$$
\dim^{\cU/W_1Z_3}([W_3]_{W_1Z_3}) = 
\dim^\cU(W_1Z_3W_3) - \dim^\cU(W_1Z_3),
$$
and
$$
\dim^{\cU}\bigl( (W_1W_2) \cap Z_3 \bigr)
= 
\dim^{\cU}( W_1W_2) + \dim^\cU (Z_3 ) - \dim(W_1W_2Z_3),
$$
and upon adding these equalities we get
\begin{align*}
&\dim^{\cU/W_1Z_3}([W_3]_{W_1Z_3}) 
+
\dim^{\cU}\bigl( (W_1W_2) \cap Z_3 \bigr) \\
&= s_1  - \dim^\cU(W_1Z_3) + \dim^{\cU}( W_1W_2) + \dim^\cU (Z_3 ) \\
&= 1+ s_1 + \dim^\cU(W_1\cap Z_3) 
= 1+ s_1 + s_2 + \dim^\cU(W_2\cap Z_3).
\end{align*}
Applying this to \eqref{eq_upshot_before_wwz_calc} yields
\eqref{eq_first_discoordination_bound}.
\end{proof}

\subsection{Proof of the Second Discoordination Bound}

\begin{proof}[Proof of Theorem~\ref{th_second_discoordination_bound}]
Our strategy is to use 
Theorem~\ref{th_first_discoordination_bound}
and a seemingly crude bound on the discoordination.
First, according to
Theorem~\ref{th_quotient_via_subspace_in_two}, since $W_1\subset P_1\cap P_2$
with notation as in Theorem~\ref{th_first_discoordination_bound},
we have
$$
\delta =  {\rm DisCoord}^{\cU/W_1}([P_1],[P_2],[Z_3])
=
{\rm DisCoord}^\cU(P_1,P_2,Z_3) ,
$$
which in turn equals
$$
\dim^{\cU/S_2}\Bigl( \bigl[ (P_1+P_2)\cap Z_3 \bigr]_{S_2} \Bigr),
$$
where
$$
S_2=S_2(P_1,P_2,Z_3) = P_1\cap P_2 + P_1\cap Z_3 + P_2\cap Z_3.
$$
Since $W_1\subset S_2$
we have
\begin{equation}\label{eq_perhaps_crude_discoord_bound}
\delta = 
\dim^{\cU/S_2}\Bigl( \bigl[ (P_1+P_2)\cap Z_3 \bigr]_{S_2} \Bigr)
\le      
\dim^{\cU/W_1}\Bigl( \bigl[ (P_1+P_2)\cap Z_3 \bigr]_{W_1} \Bigr).
\end{equation} 

Since $P_1+P_2$ implies $W_1,W_2$,
$$
\dim^{\cU/W_1}\bigl( [(P_1+P_2) \cap Z_3] \bigr) 
=
\dim^{\cU/W_1}\bigl( [(W_1W_2) \cap Z_3] \bigr) +
\dim^{\cU/W_1W_2}\bigl( [(P_1+P_2) \cap Z_3] \bigr) 
$$
Applying this to
\eqref{eq_perhaps_crude_discoord_bound} we have
$$
\delta \le 
\dim^{\cU/W_1}\bigl( [W_1W_2 \cap Z_3]_{W_1} \bigr) +
\dim^{\cU/W_1W_2}\bigl( [(P_1+P_2) \cap Z_3]_{W_1W_2} \bigr) ;
$$
equivalently
$$
\dim^{\cU/W_1W_2}\bigl( (P_1+P_2) \cap Z_3 \bigr) - \delta \ge 
- \dim^{\cU/W_1}\bigl( W_1W_2 \cap Z_3 \bigr)  .
$$
Putting this into
\eqref{eq_first_discoordination_bound} we have
$$
2R+3M \ge 5 + 
 \dim^\cU(W_2\cap Z_3) 
- \dim^{\cU/W_1}\bigl( [W_1W_2 \cap Z_3]_{W_1} \bigr)  + s,
$$
where $s$ is a term that vanishes upon symmetrization.
Since
$$
\dim^{\cU/W_1}\bigl( [W_1W_2 \cap Z_3]_{W_1} \bigr)
=
\dim^\cU\bigl( (W_1W_2) \cap Z_3 \bigr) - \dim^\cU(W_1\cap Z_3),
$$
we have 
$$
2R+3M \ge 5 + 
 \dim^\cU(W_2\cap Z_3) 
+ \dim^\cU(W_1\cap Z_3)
- \dim^{\cU}\bigl( (W_1W_2) \cap Z_3 \bigr)  + s
=5 - \delta' + s.
$$
Symmetrizing yields
\eqref{eq_second_discoordination_bound}.
\end{proof}

We remark that we can turn the inequality in
\eqref{eq_perhaps_crude_discoord_bound} into an equality by subtracting
from the right-hand-side the non-negative term
$$
\dim^{\cU}\bigl( (P_1+P_2)\cap Z_3\cap S_2 \bigr)
-
\dim^{\cU}\bigl( (P_1+P_2)\cap Z_3\cap W_1 \bigr) ,
$$
which, since $P_1+P_2$ contains both $W_1$ and $S_2$, simplifies to
$$
\dim^{\cU}\bigl( Z_3\cap S_2 \bigr)
-
\dim^{\cU}\bigl( Z_3\cap W_1 \bigr) .
$$

\section{A Hybrid Lower Bound Involving Tian's Method}
\label{se_Joel_hybrid}

The point of this section is to prove that for $N=K=3$, any separated
scheme must satisfy
$$
6M + 5R \ge 11.
$$
The proof given in this section will introduce some helpful concepts
to give the same bound for any general linear schemes
in Section~\ref{se_Joel_hybrid_not_sep}.
Along the way we will give some useful notation that we will use
in both sections.
We will also discuss our conjecture that $4M+3R\ge 7$ for all
linear schemes for $N=K=3$.
We begin with some useful notation.

\subsection{Notation for Linear Schemes for $N=3$}

In this subsection we introduce some useful notation for linear
schemes for $N=3$; our interest is the case $K=3$, but the
notation is valid for any $K$.

We will introduce the following notation for 
linear schemes with $K=3$; for simplicity we work with symmetric
schemes.
We recall that scheme becomes symmetric after taking the concatenation of
the $N!\,K!$ symmetric forms of the scheme
(see Section~\ref{se_Joel_symmetrization}), and if the original scheme
is linear or separated, then the same is true of the concatenation.

\begin{definition}\label{de_notation_for_N_equal_3_schemes}
Consider a symmetric linear scheme for $N=3$ and some value of $K$.
For all $j\in [K]$,
Lemma~\ref{le_Z_scheme_decomposition} implies that $Z_j$ decomposes
as a sum consisting of
\begin{enumerate}
\item 
a pure individual scheme  $A^1_j,B^1_j,C^1_j$;
\item
a pure Tian scheme, i.e.,
$A^2_j\oplus B^2_j,B^2_i\oplus C^2_j$;
\item
a pure {\em symmetric two-way scheme}, meaning
an $AB$-scheme $A^{'3}_j\oplus B^{'3}_j$,
an $AC$-scheme $A^{''3}_j\oplus C^{'3}_j$, and
a $BC$-scheme $B^{''3}_j \oplus C^{''3}_j$, where
all the 
$A^{'3}_j,A^{''3}_j,B^{'3}_j,B^{''3}_j,C^{'3}_j,C^{''3}_j$ have
the same dimension; and
\item
a three-way-scheme
$A^4_j\oplus B^4_j\oplus C^4_j$.
\end{enumerate}
By symmetry, the dimensions of 
$$
A^{'3}_j,\ A^{''3}_j,
\ B^{'3}_j,\ B^{''3}_j,
\ C^{'3}_j,\ C^{''3}_j
$$
are independent of $j$ and are all of the same dimension.
We set $A^3_j=A^{'3}_j+A^{''3}_j$, 
$$
A_j = A^1_j+\cdots+ A^4_j,
$$
and similarly for $B_j$ and $C_j$.
We then let $r_1,\ldots,r_4$ be given by
by
$$
F r_i=\dim(A^i_j) = \dim(B^i_j) = \dim(C^i_j),
$$
and we let $r_5$ be given by
$$
3 F r_5 = F - \dim(A_1+A_2+A_3)
$$
{\mygreen (which therefore represents the ``information''
of $W_1$ that
is ``unused'' by the caches).}
we call $(r_1,r_2,r_3,r_4,r_5)$ the {\em ratios} of the scheme.
It will also be convenient to set for $1\le i\le 4$
$$
W^i = \sum_{j=1}^K (A^i_j+B^i_j+C^i_j),
$$
which therefore gives linearly independent subspaces $W^1,\ldots,W^4$
of $W=W_1+W_2+W_3$.
\end{definition}

Next we make a few remarks about symmetric linear schemes 
in the above notation.
First, for all $j\in[K]$ we have
$$
\dim(A_j) = \dim(B_j) = \dim(C_j) = F(r_1+r_2+r_3+r_4)
$$
since for fixed $j$, the subspaces $A_j^i$ with $1\le i\le 4$ are
independent in Lemma~\ref{le_Z_scheme_decomposition}, and
similarly for $B$ or $C$ replacing $A$.
Second, we have 
\begin{equation}\label{eq_memory}
M = 3 r_1 + 2 r_2 + (3/2) r_3 + r_4.
\end{equation} 
Third, 
the second discoordination bound \eqref{eq_second_discoordination_bound}
implies that
\begin{equation}\label{eq_from_second_discoordination}
2R+3M \ge 5 - r_2 - r_3/2 .
\end{equation} 
Fourth,
if $K=3$ and the scheme is separated
(Definition~\ref{de_separated}), then $A_1,A_2,A_3$ are linearly
independent, and hence
$$
3 F  r_5 = F - 3 F(r_1+r_2+r_3+r_4),
$$
and hence
\begin{equation}\label{eq_sum_of_r_i_s}
r_1+r_2+r_3+r_4+r_5 = 1/3.
\end{equation} 

\subsection{Lower Bounds on Pure Schemes}

By a {\em pure scheme} we mean a separated, linear scheme for $N=K=3$
with some $r_i=0$ for all but one value of $i$, and hence, 
by \eqref{eq_sum_of_r_i_s}, equivalently $r_i=1/3$ for a unique $i$.
In this section we discuss pure schemes and their implication for
linear schemes in the (currently open) range of $1/2<M<1$.

We remark that for the pure schemes we have the following bounds
(recall results from Subsection~\ref{su_review_N_K_both_3}):
\begin{enumerate}
\item 
If $r_i=1/3$ for $i=1,3,4$, then \eqref{eq_from_second_discoordination}
gives us tight bounds:
the value $r_1=1/3$ corresponds to the caching scheme that
achieves $(M,R)=(1,1)$ (due to \cite{MA_niesen_2014_seminal}),
the $r_3=1$ to our new point
$(M,R)=(1/2,3/5)$ of Section~\ref{se_Joel_new_point},
and $r_4=1/3$ the scheme of
\cite{chen_journal} that achieves $(1/3,2)$.
\item
If $r_5=1/3$, then we have $M=0$ and then $R=3$ is clearly achievable
and is optimal due to $3R+M\ge 3$ of 
\cite{MA_niesen_2014_seminal}.
\item
If $r_2=1/3$, then $M=2/3$ and 
Theorem~\ref{th_tian_type_scheme} applies here and
gives $3M+2R\ge 5$ and gives
$R\ge 3/2$, which is {\em worse} then a convex combination of $(1/2,5/3)$
of $(1,1)$ which achieves $(2/3,13/9)$.
At present we do not know of the best lower bound for $R$ for this
particular scheme, i.e., the pure Tian scheme;
however, it may be useful to get some bounds to see if one can show
$4M+3R\ge 7$, since the main weakness of
\eqref{eq_second_discoordination_bound} 
is its $r_2$ term.
\end{enumerate}
Hence, assuming that we use one of the above pure schemes,
we have lower bounds that match what is achievable,
except for $r_2=1/3$, the pure Tian scheme, where the lower bound
for $R$ rules out this scheme as optimal.

Notice that when we create a scheme from a convex combination
of the two schemes that achieve the $(M,R)$ values of $(1/2,5/3)$ and $(1,1)$,
so that the caches involve only nonzero $\cW^3$ and $\cW^1$ parts,
then $X_{123}$ can be written as a sum of two subspaces, one that 
involves only $\cW^3$, the other only $\cW^1$.
It follows that if a scheme achieves an $(M,R)$ value with $1/2<M<1$
and $4M+3R<7$, i.e., below the convex hull of $(1/2,5/3)$ and $(1,1)$, 
then the corresponding $X_{123}$ cannot decompose in this way,
i.e., the $X_{123}$ does not factor through the decomposition
of $W^1,\ldots,W^4$ of $W^1+\cdots+W^4$.
(This observation may be useful in trying to find a scheme with
$1/2<M<1$ and $4M+3R<7$ or in refuting its existence.)

\subsection{Tian's Method for Separated Schemes without Other Considerations}

In this subsection we motivate our hybrid approach to proving
Theorem~\ref{th_separated_N_K_three} and sketch our hybrid
approach in rough terms;
this subsection provides intuition but
is not essential to the rest of this paper.

Our hybrid approach 
applies Tian's method to $X_{123}$ after we prove that we
can separate $X_{123}$ into a part that deals with
information it needs to due an $r_5>0$ and an $r_4>0$ part of the
$Z_i$'s, and a
part that deals with the $r_1,r_2,r_3>0$ parts.
We deal with any $r_5,r_4$ parts by a ``direct'' linear algebra argument,
and then apply Tian's method to the remaining $r_1,r_2,r_3>0$ parts.

In more detail, 
consider applying the proof of Theorem~\ref{th_tian_type_scheme} to a
separated linear scheme with notation as in
Definition~\ref{de_notation_for_N_equal_3_schemes}.  
In this case
$Z_1,X_{123}$ implies $W_1$ and
$$
B_1^1,\ C_1^1,\ B_1^2,\ C_1^2,B_1^{'3},C_1^{'3},B_1^{''3}\oplus C_1^{''3},
B_1^4\oplus C_1^4,
$$
and hence, with notation as in 
\eqref{eq_G_1_thru_G_9_Tian_scheme} and
\eqref{eq_G_leftover_from_Z_one},
it follows that
\begin{equation}\label{eq_direct_tian_method_no_r_5}
{\rm Rank}
\begin{bmatrix}
G_5 & G_6 & G_8 & G_9 
\end{bmatrix} 
\le
F(M+R''-1 - 2r_1-2r_2-(3/2)r_3-r_4).
\end{equation} 
Therefore the same reasoning applied to $Z_2,X_{123}$ and
$Z_3,X_{123}$ yields, with the same reasoning there, the bound
\begin{equation}\label{eq_Tian_separated_without_hybrid} 
3M+2R'' \ge 
3(1 - 2r_1-2r_2-(3/2)r_3-r_4).
\end{equation} 
combining this with \eqref{eq_sum_of_r_i_s}
we may write
\eqref{eq_Tian_separated_without_hybrid} as
\begin{equation}\label{eq_Tian_separated_without_hybrid_2}
3M+2R \ge 5 - (3/2)r_3-6r_4-9 r_5.
\end{equation} 
Notice that $r_5$ is conspicuously absent from
\eqref{eq_direct_tian_method_no_r_5} and
\eqref{eq_Tian_separated_without_hybrid}, although if $r_5>0$
we certainly expect that this places a further condition on
$X_{123}$ that should be reflected in these equations.
In fact,
if $r_5>0$ it is easy to see by ``direct linear algebra'' that $X_{123}$
must contain $Fr_5$ bits of information 
for $Z_1$ to reconstruct the information about 
$W_1$ that is missing from $A_1+\cdots+A_4$, and similarly
for $W_2$ and the $B_i$'s, and $W_3$ and the $C_i$'s;
furthermore this information is independent (in the sense of
independent subspaces) from the part of $X_{123}$ needed to
deal with the $r_1,\ldots,r_4$ parts of the scheme.
Hence we are led to consider a hybrid approach:
before applying Tian's method, we first prove that
$X_{123}$ is a sum of the above $3Fr_5$ bits of
missing information, plus
information leftover for the part of the scheme
represented by any positive values of $r_1,\ldots,r_4$.
One can take this approach further, reasoning ``directly'' about the
part of the scheme represented by $r_4$; this gives an
improved result for separated schemes; however, we do not know how to
reason directly about the $r_4$ part
for general linear schemes, due to the possibly complicated way that the
$r_4$ parts of the $Z_i$ may intersect
with the $r_1,r_2,r_3$ parts of the $Z_i$ 
(without the assumption
of separability).
Hence in Section~\ref{se_Joel_hybrid_not_sep} we apply direct linear
algebra only to the $r_5$ part, and Tian's method to the
$r_1,\ldots,r_4$ parts;
the resulting bound is weaker, but still suffices---when combined
with the second discoordination bound---to prove $6M+5R \ge 11$.

\ignore{

If we consider a scheme where $r_i=0$ for all $i$ except $i=4$, i.e.,
$r_4=1/3$, i.e.,
a pure three-way-scheme, we conclude
$$
3M+2R \ge 2;
$$
this bound is not optimal, since the three-way scheme has $M=1/3$,
for which we know $R\ge 2$. 
Even worse is the case where $r_i=0$ for
all $i$ except $i=5$, i.e., $r_5=1/3$, i.e., $M=0$, where
\eqref{eq_Tian_separated_without_hybrid_2} reads
$$
3M+2R\ge 2,
$$
whereas for $M=0$ we know $R\ge 3$.

Hence, to get better bounds for general separated schemes,
we will follow a hybrid approach: first we reason directly
about the part of $X_{123}$ due to the $r_5>0$
part of the scheme, i.e., what we can infer if
$A_1+A_2+A_3$ is not all of $W_1$.  
If we immediately apply Tian's method to what remains,
i.e., the four schemes
in the proportions $r_1,r_2,r_3,r_4$, we get the improved bound
$$
3M+2R \ge 5 - (3/2) r_3 -6r_4 + 3r_5,
$$
which is tight when $r_5=1/3$ and the rest are $0$.
However, if we directly reason
about the part of $X_{123}$ due to the $r_4>0$ part,
and then apply Tian's method 
i.e., we seek an improved rank lower bound due to what
$X_{123}$ must contain due to the $A^4_j,B^4_j,C^4_j$ parts,
we get the further improvement
$$
3M+2R \ge 5 - (3/2) r_3 + 3r_5.
$$
We can also directly reason about the $r_3>0$ part of the
scheme, but this will yield the same inequality.
}

\subsection{The Result $6M+5R\ge 11$ for Separated Linear
Schemes}

The main computation we do in this section is the following.
\begin{theorem}\label{th_separated_N_K_three}
In any linear separated scheme as above for coded caching with 
$N=K=3$, we have
\begin{equation}\label{eq_hybrid_tian_type_bound}
2R+3M \ge 5 - (3/2)r_3 + 3 r_5.
\end{equation} 
\end{theorem}

Notice that the lower bound in the above theorem is not tight for
$r_3=1/3$, which is the $(1/2,5/3)$ achievable point described in
Section~\ref{se_Joel_new_point}.
The above theorem easily gives the following corollary.

\begin{corollary}\label{co_separated_6M_5M_bound}
In any linear separated scheme above for coded caching as in
Theorem~\ref{th_separated_N_K_three}, we have
$$
6M+5R \ge 11.
$$
\end{corollary}
\begin{proof}
Adding the following equalities and inequalities:
\begin{enumerate}
\item 
$1$ times \eqref{eq_hybrid_tian_type_bound},
\item 
$-9/2$ times \eqref{eq_sum_of_r_i_s}
\item 
$-3/2$ times \eqref{eq_memory}, and
\item 
$3/2$ times \eqref{eq_from_second_discoordination},
\end{enumerate}
yields
\begin{equation}\label{eq_hybrid_separable_with_rs}
6M+5R \ge 11 + 3r_4+
{\mygreen (15/2)r_5.}
\end{equation} 
\end{proof}
Note that this shows that $5R+6M>11$ unless $r_4=r_5=0$.

The bound in this corollary gives a slight improvement to Tian's bound 
$R+M\ge 2$, as both
pass through the achievable point $(M,R)=(1,1)$; however, the bound
in the corollary requires the assumption that the scheme
is linear separated schemes; in the next section we remove the
separability assumption, but obtain the weaker inequality
$$
6M+5R \ge 11 + 5r_5 ,
$$
which still gives $6M+5R\ge 11$;
by contrast, Tian's bound $M+R\ge 2$ is valid
for any scheme, including non-linear schemes.

[We remark that Tian's bound $R+M\ge 2$ has a short proof:
for linear schemes, Tian's bound $R+M\ge 2$ follows
from the fact that $(2R+2M)F \ge \dim(X_{123}+Z_1)+\dim(X_{312}+Z_2)$,
which by the dimension formula equals
$$
\dim( X_{123}+Z_1+X_{312}+Z_2) + 
\dim\bigl( (X_{123}+Z_1)\cap(X_{312}+Z_2) \bigr),
$$
and the first dimension above equals $3F$, and the second dimension is
at least that of $W_1$, namely $F$.  For non-linear schemes this
proof still holds, since the above lower bound on $2R+2M$ becomes
$$
H(X_{123},Z_1,X_{312},Z_2) + 
I\bigl( (X_{123}+Z_1);(X_{312}+Z_2) \bigr),
$$
which again are bounded above by $3F+F$, using the fact that the
two-way mutual information $I(X;Y)$ of random variables $X$ and $Y$
is bounded from below by $H(Z)$ for any $Z$ that is implied by both $X$
and $Y$.]

Next we make some conjectures and further remarks. 

First, we conjecture
that one can improve
the bound in Theorem~\ref{th_separated_N_K_three} to
$$
2R+3M \ge 5 - (1/2)r_3 + 3 
{\mygreen r_5.}
$$
If so then adding to this inequality the following
\begin{enumerate}
\item 
$-3/2$ times \eqref{eq_sum_of_r_i_s}
\item 
$-1/2$ times \eqref{eq_memory}, and
\item 
$1/2$ times \eqref{eq_from_second_discoordination},
\end{enumerate}
we get $4M+3R\ge 7$.
This would then imply
that no separated linear scheme can improve upon a convex
combination of $(1/2,5/3)$ and $(1,1)$.

Second, we conjecture that any optimal linear scheme is separated.
Our difficulty in attacking either conjecture is the possible way
in which the $X_{ijk}$ can involve XOR's of the bits $A^j_i,B^j_i,C^j_i$
over different values of $j$; furthermore, if a scheme is not
separated, the relationships between the $Z_1,Z_2,Z_3$ could
conceivably be quite complicated.

Third, the line connecting $(1/2,5/3)$ and $(1,1)$ is
$4M+3R=7$, and we conjecture than
$$
4M+3R \ge 7
$$
holds for all $1/2 \le M\le 1$ and any linear scheme;
if this holds, then $4M+3R\ge 7$ holds for all schemes unless there is
a non-linear scheme which improves upon this
(we do not particularly conjecture one way or another on the existence
of such a non-linear scheme).

\subsection{Proof of Theorem~\ref{th_separated_N_K_three}}

\ignore{
THIS FIRST PARAGRAPH IS EXPLAINED BEFORE.  OMIT ONE EXPLANATION.

Our approach will be similar to Tian's method and the proof
of Theorem~\ref{th_tian_type_scheme} here.
However we need to do some ``preprocessing'' of the analog of
$G$ in our proof of Theorem~\ref{th_tian_type_scheme} for the following
reason:
consider a pure three-way scheme, so $r_4=1/3$ and $r_j=0$ for all other $j$.
Mimicking the proof of Theorem~\ref{th_tian_type_scheme}, we see that
$X_{123}+Z_1$ contain the subspaces $A_1,A_2,A_3$ and 
$A_1\oplus A_2\oplus A_3$, and hence 
the ranks of the matrices in
\eqref{eq_G_G_prime_leftover_from_Z_one} and
\eqref{eq_G_leftover_from_Z_one} are at most of rank
$$
(M+R'-4/3)F,
$$
and one concludes that
$$
R' \le 3(M+R'-4/3)F
$$
which gives only $2R'+3M\ge 4/3$; furthermore if $r_5=1/3$, so
each $Z_i=0$, this type of argument would show
$$
R' \le 3(M+R'-3/3)F
$$
and hence $2R'+3M\ge 3/3$.
Similarly, 
for general $r_1,\ldots,r_5$, this method---as is---yields the lower bound
$$
2R+3M \ge 5 - (3/2) r_3 - 3r_4 - 6 r_5.
$$
A better approach will be a ``hybrid'' approach as follows: first we
directly reason about the matrix $G$ whose row space equals
$\iota(X_{123})$, namely about its parts corresponding to $r_4$ and 
$r_5$.  We then apply Tian's method to the remaining parts of $G$.
}

We follow the hybrid approach of first making some
``direct linear algebra'' remarks regarding the information
that $X_{123}$ must contain due to parts of the scheme with
$r_5>0$ and $r_4>0$; then we apply Tian's method to the rest.

\begin{proof}[Proof of Theorem~\ref{th_separated_N_K_three}]
As in the proof of Theorem~\ref{th_tian_type_scheme}, let us
specify a basis for $W=W_1+W_2+W_3$.
Consider the basis for $W$ consisting of five parts:
\begin{enumerate}
\item 
for each $i=1,2$ and $j=1,2,3$, let
$\cA^i_j$ be a basis $A^i_j$, and similarly for
$\cB^i_j$ and $\cC^i_j$.
Let $\cA^i$ be the union of $\cA^i_1,\cA^i_2,\cA^i_3$, and similarly
for $\cB^i,\cC^i$,
$\cW^i$ be the union of $\cA^i,\cB^i,\cC^i$.
\item
Similarly, let $\cA^{'3}_j$ be an arbitrary basis of
$A^{'3}_j$, and similarly for
$\cA^{''3}_j$, and let $\cA^3$ be the union of $\cA^{'3}$ and $\cA^{''3}$.
Similarly for $\cB$ or $\cC$ replacing everywhere $\cA$,
and let
$\cW^3$ be the union of $\cA^3,\cB^3,\cC^3$.
\item
For $i=4$ we take a different approach: 
for $j=1,2,3$, let
$\tilde\cC^4_j$ be an arbitrary basis
for $A^4_j\oplus B^4_j \oplus C^4_j$
(recall the mildly abusive meaning of
$A^4_j\oplus B^4_j \oplus C^4_j$
in Definition~\ref{de_Z_schemes}, which is first introduced in
Subsection~\ref{su_sum_direct_sum_oplus}),
let
$\cA^4_j,\cB^4_j$ respectively be arbitrary bases for 
$A^4_j,B^4_j$;
let $\cA^4$ be the union of the $\cA^4_j$, and similarly for
$\cB^4$ and $\tilde\cC^4$,
and let $\cW^4$ be the union of these sets.
\item
For $1\le j\le 3$,
let $A_j$ be the span of $\cA^i_j$ over all $1\le i\le 4$.
Let $\cA^5$ be a basis of $W_1$ relative to $A_1+A_2+A_3$;
hence $|\cA^5|=3Fr_5$ which represents the amount of information
in $W_1$ that does not occur in the $A^i_j$ ranging over all $1\le i\le 4$
and $1\le j\le 3$.
Introduce similar notation for $\cB^5$ and $\cC^5$, and
$\cW^5=\cA^5\cup \cB^5\cup\cC^5$.
\end{enumerate}
Finally, let $\cW$ be the union of the $\cW^i$, which for block purposes
we arrange in the order $\cW^1,\ldots,\cW^5$.
As in the proof of Theorem~\ref{th_tian_type_scheme}, the basis
$\cW$ of $W$ gives an isomorphism $\iota=\iota_{\cW}\from W\to\field^{3F}$
with $\field=\integers/2\integers$.

It will be crucial to note that for each $j=1,2,3$,
the vectors in $\iota(Z_j)$ have zeros in all their components corresponding
to the basis elements in all $\cW^5$, and of those in $\cA^4$ and $\cB^4$.
[We actually know more: for example, $\iota(Z_1)$ has zeros in its
components corresponding to the $\cA^i_j$ with $i\le 3$ and $j\ne 1$,
but we won't need such observations here.]

Now we will describe a set of vectors in $X_{123}$ that are linearly
independent; their span will be a subspace of $X_{123}$, namely
$X'_{123}$.  Our goal is to describe vectors so that $\iota(X'_{123})$
has a convenient form to which we can employ a hybrid strategy, first
making direct observations about part of $\iota(X'_{123})$, and
afterward we will apply Tian's method
(in Theorem~\ref{th_tian_type_scheme}) to a matrix involving of
the remaining upper left part of $\iota(X'_{123})$.

Consider any basis vector $a\in\cA^5$.  Since user $1$ must be
able to infer $a$ from $X_{123}$ and $Z_1$, we have
$a=x+z$ where $x\in X_{123}$ and $z\in Z_1$.  It follows that
$-x=a-z$, and hence 
$\iota(-x)$---viewed as a block vector that breaks the basis $\cW$
into its $\cW^1,\ldots,\cW^5$ blocks---is of the form:
$$
\begin{bmatrix}
\ell^1 & \ell^2 & \ell^3 & \ell^4 & e_r
\end{bmatrix} ,
$$
where $e_r$ is one of the standard basis vectors 
in the $\cW^5$ block (in particular, in the $\cA^5$ part of $\cW^5$),
and where 
the $-\ell^i$ corresponds to the part of $\iota(z)$ in
the $\cW^i$ block of the basis $\cW$;
furthermore, as remarked above, $\ell^4$ has zeros in the components 
corresponding to vectors in $\cA^4$ and $\cB^4$.

Doing this for each basis vector in $\cA^5$, and similarly for the
rest of the basis vectors in $\cW^5$
we get
a set of vectors in $X_{123}$ whose image under $\iota$, when arranged
as row vectors, yields a block matrix of the form:
\begin{equation}\label{eq_iota_X_fifth_part}
\begin{bmatrix}
L^1 & L^2 & L^3 & L^4 & I 
\end{bmatrix},
\end{equation} 
where $I$ is a $9r_5F \times 9r_5F$ identity matrix, and for $i=1,2,3,4$,
$L^i$ is a block matrix with $9r_5F$ rows.
These rows are linearly independent because of the $I$ in the block form
above.

Next, user 1 can infer each element, $a\in\cA^4_j$ with $j=1,2,3$
from $Z_1$ and $X_{123}$, and again $a=z+x$ with $z\in Z_1$ and
$x\in X_{123}$, so $\iota(x)=\iota(a)-\iota(z)$ gives us vectors
of the form
$$
\begin{bmatrix}
\ell^1 & \ell^2 & \ell^3 & e_r+\ell^4 & 0
\end{bmatrix} ,
$$
where $e_r$ is the standard basis vector corresponding to $a\in \cA^4$,
and the $\ell^i$ result from $-\iota(z)$, and we observe that
$\ell^4$ has zero components in the positions corresponding to
$\cA^4$ and $\cB^4$, only possibly nonzero in those components
corresponding to $\tilde\cC^4$.
Doing the same for all $b\in \cB^4_j$ with $j=1,2,3$, we get
elements of $X_{123}$ such that $\iota$ of these elements, 
arranged as row vectors, is of the form
\begin{equation}\label{eq_iota_X_forth_part}
\begin{bmatrix}
 P^1 & P^2 & P^3 & [Q\ I] & 0
\end{bmatrix}
\end{equation}
where the $I$ in $[Q\ I]$ is an $6r_4F$ by $6r_4F$ identity matrix,
and $Q$ is the matrix of components corresponding to 
elements of $\tilde\cC^4$.

Now we observe that all the rows of the matrices in
\eqref{eq_iota_X_forth_part} and \eqref{eq_iota_X_fifth_part} are
linearly independent with the following argument:
when we combine these matrices we get a
matrix
\begin{equation}\label{eq_iota_X_forth_and_fifth_part}
\begin{bmatrix}
 P^1 & P^2 & P^3 & [Q\ I] & 0 \\
L^1 & L^2 & L^3 & [Q'\ 0] & I 
\end{bmatrix}
\end{equation}
(here $L^4$ becomes $[Q'\ 0]$ since the $Z_i$ have zero components
in elements of $\cA^4$ and $\cB^4$),
and separating the $\tilde\cC^4$ block part from the $\cA^4,\cB^4$ part
we get a block matrix
\begin{equation}\label{eq_iota_X_six_parts}
\begin{bmatrix}
 P^1 & P^2 & P^3 & Q & I & 0 \\
L^1 & L^2 & L^3 & Q'& 0 & I 
\end{bmatrix},
\end{equation} 
whose two right-most columns give a $6r_4F+9r_5F$ identity matrix.

At this point we have identified a subspace
$X'_{123}$ of $X_{123}$, and a basis of $X'_{123}$, whose image
under $\iota$, viewed as row vectors, equals the block matrix in
\eqref{eq_iota_X_six_parts}.
Now list all of the vectors in $\cW^1\cup \cW^2\cup \cW^3$ as a sequence
in any order
$$
v_1,\ldots,v_m 
$$
(note that here the subscript does not refer to which scheme or user
is involved).
Each $v_i=x_i+z_i$ for some $x_i\in X_{123}$ and $z_i\in Z_1\cup Z_2\cup Z_3$.
We let 
$$
\hat X_{123} = X'_{123} + {\rm Span}(x_1,\ldots,x_m).
$$
Now we create a matrix whose rowspace equals $\iota(\hat X_{123})$ as 
follows: we begin with the matrix in
\eqref{eq_iota_X_six_parts} and for $i=1,\ldots,m$ we add a row
for each $x_i$ such that
$$
x_i \notin X'_{123} + {\rm Span}(x_1,\ldots,x_{i-1})
$$
using the same idea as before: since $-x_i=v_i-z_i$
we add the row $\iota(-x_i)=\iota(v_i)-\iota(z_i)$ which has the form
$$
\begin{bmatrix}
\ell^1 & \ell^2 & \ell^3 & \ell^4 & 0_{\cA^4,\cB^4} & 0_{\cW^5}
\end{bmatrix} ,
$$
where $\ell^4$ corresponds to the $\tilde\cC^4$ part, and the
subscripts on the $0$ indicates the corresponding components.
Adding all such vectors $x_i$ to obtain $\hat X_{123}$ we have that
$\iota(\hat X_{123})$, viewed as row vectors, is the row space of a matrix
\begin{equation}\label{eq_iota_X_three_block_rows}
G =
\begin{bmatrix}
G^{''1} & G^{''2} & G^{''3} & Q'' & 0 & 0 \\
 P^1 & P^2 & P^3 & Q & I & 0 \\
L^1 & L^2 & L^3 & Q'& 0 & I 
\end{bmatrix}.
\end{equation} 
Setting
$$
G'' = 
\begin{bmatrix}
G^{''1} & G^{''2} & G^{''3}
\end{bmatrix},
$$
we have
$$
FR \ge \dim(X_{123})\ge \dim(\hat X_{123})
\ge {\rm Rank}(G) 
$$
$$
=
{\rm Rank}( [\ G'' \ Q''] )
+ 6r_4F + 9 r_5 F
\ge
{\rm Rank}(G'') + 6r_4F + 9 r_5 F,
$$
and hence
\begin{equation}\label{eq_FR_versus_rank_G_double_prime}
FR \ge R''F + 6r_4F + 9 r_5 F,
\quad\mbox{where}\quad R'' = {\rm Rank}(G'')/F
\end{equation} 

Now let's reason about $G''$.  First, we prove that
$$
\iota\Bigl( {\rm Span}(x_1,\ldots,x_m) \Bigr)
$$
lies entirely in the rowspace of
$[\ G''\ Q'' \ 0 \ 0]$: indeed, surely each $x_i$ with
$$
x_i \notin X'_{123} + {\rm Span}(x_1,\ldots,x_{i-1})
$$
has $\iota(x_i)$ as
one of the rows of $[\ G''\ Q''\ 0 \ 0]$, by our construction above.
However, if 
$$
x_i \in  X'_{123} + {\rm Span}(x_1,\ldots,x_{i-1})
$$
then $\iota(x_i)$ lies in some combination of the row space of 
$G$ in \eqref{eq_iota_X_three_block_rows}.  But since
$x_i=v_i-z_i$, then $x_i$ has zero components in positions corresponding
to $\cA^4,\cB^4$ and $\cW^5$; but since the two last columns of
$G$ are 
$$
\begin{bmatrix}
0 & 0 \\
I & 0 \\
0 & I 
\end{bmatrix},
$$
and $v_i$ corresponds to a vector in one of $\cW^1,\cW^2,\cW^3$, we have
that $\iota(x_i)=\iota(v_i)-\iota(z_i)$ has zero component in
the positions corresponding to $\cA^4,\cB^4$ and $\cW^5$; hence
$\iota(x_i)$, which is a linear combination of rows in $G$, cannot
involve the bottom two rows blocks, which correspond to $\iota(X'_{123})$.

Hence we know
$$
\iota\Bigl( {\rm Span}(x_1,\ldots,x_m) \Bigr) =
{\rm Rowspace}\Bigl([\ G''\ Q'' \ 0 \ 0] \Bigr) .
$$
Next consider the special case where all the vectors in
$\tilde\cC^4=0$, i.e., $A^4_j\oplus B^4_j \oplus C^4_j=0$ for all $j$;
in this special case,
user $1$ can reconstruct $A^i_j$ for all $1\le i\le 3$
and $1\le j\le 3$, since each vector in $\cA^i_j$ occurs in the
sequence $v_1,\ldots,v_m$, and we may set $Q''=0$ and compute
the same values of $v_1,\ldots,v_m$.
Hence we may replace $x_1,\ldots,x_m$ with the vectors 
$x''_1,\ldots,x''_m$
obtained by
discarding the $\tilde\cC^4$ components of $x_1,\ldots,x_m$
giving an $X''_{123}={\rm Span}(x''_1,\ldots,x''_m)$
that allows users to reconstruct their $\cW^1,\cW^2,\cW^3$ parts
of their files.  
So the total memory user $j$ needs to do this is $M''F$ where
$$
M'' = M - r_4,
$$
and the dimension of
$$
X''_{123}= {\rm Span}(x''_1,\ldots,x''_m)
$$
equals $R''F={\rm Rank}(G'')$.

Now we apply Tian's argument to show that in view of what user $1$
knows with $Z''_1$ and $X''_{123}$
(see proof of Theorem~\ref{th_tian_type_scheme}),
the columns of $G''$
corresponding to $\cB^i_2,\cB^i_3,\cC^i_2,\cC^i_3$ ranging over
all $1\le i\le 3$ has dimension at most 
\begin{equation}\label{eq_stuff_that_bounds_part_of_col_sp_G_double_prime}
M''F + R'' F - \bigl(5r_1+5r_2+(9/2)r_3  \bigr) F.
\end{equation} 
Proceeding similarly for users $2$ and $3$ we have
that the 
{\mygreen dimension of the} entire column space of $G''$,
{\mygreen which equals $R''F$
(see \eqref{eq_FR_versus_rank_G_double_prime})}, is bounded above by
{\mygreen three times 
\eqref{eq_stuff_that_bounds_part_of_col_sp_G_double_prime},
and hence}
$$
R''F\le 3 \Bigl( M''F + R'' F - \bigl(5r_1+5r_2+(9/2)r_3  \bigr) F \Bigr).
$$
and hence
$$
3M''+2R'' \ge 15 r_1 + 15 r_2 + (27/2) r_3.
$$
Using \eqref{eq_FR_versus_rank_G_double_prime} and the fact
that $M=M''+r_4$ we get
$$
3M+2R \ge 3(M''+r_4)+2(R'' + 6r_4+ 9 r_5) 
$$
$$
\ge 
15 r_1 + 15 r_2 + (27/2) r_3 + 15 r_4 + 18 r_5 
= 5 - (3/2) r_3 + 3 r_5
$$
using $\sum_i r_i = 1/3$.
\end{proof}
\section{A Hybrid Bound Without Assuming Separability}
\label{se_Joel_hybrid_not_sep}

The point of this section is to prove a slightly weaker hybrid bound
without the assumption of separability.
It is based on a weaker form of 
Theorem~\ref{th_separated_N_K_three}.


\begin{theorem}\label{th_general_N_K_three}
Consider a symmetric linear coded caching for $N=K=3$, and let
notation be as in
Definition~\ref{de_notation_for_N_equal_3_schemes}
{\mygreen (which defines $A^i_j,B^i_j,C^i_j$ for $1\le i\le 4$ and
$1\le j \le 3$, and the ratios $(r_1,\ldots,r_5)$).}
Then we have (without the assumption of separability)
\begin{equation}\label{eq_hybrid_tian_type_bound_wo_separable}
2R+3M \ge 5 - (3/2)r_3 - 3 r_4 + 3 r_5,
\end{equation}
and
\begin{equation}\label{eq_sum_of_rs_inequality_unsep}
0 \le r_1+r_2+r_3+r_4+r_5 - 1/3 .
\end{equation} 
Furthermore both these inequalities
are strict unless the scheme is separated.
\end{theorem}

Notice that the above theorem has a $-3r_4$ term 
in \eqref{eq_hybrid_tian_type_bound_wo_separable} that
Theorem~\ref{th_separated_N_K_three} does not.
Fortunately, we still get the same $6M+5R\ge 11$ bound:
the $-3r_4$ here means that
\eqref{eq_hybrid_separable_with_rs} becomes
\eqref{eq_hybrid_wo_separable_with_rs} below,
which is worse by $-3r_4$, but still implies $6M+5R\ge 11$.

\begin{corollary}
Consider a symmetric linear coded caching for $N=K=3$, and let
notation be as in
Definition~\ref{de_notation_for_N_equal_3_schemes}.
Then we have (without the assumption of separability)
$$
6M + 5R \ge 11.
$$
\end{corollary}
\begin{proof}
(Similar to 
the proof of Corollary~\ref{co_separated_6M_5M_bound},)
we add 
\begin{enumerate}
\item
$1$ times \eqref{eq_hybrid_tian_type_bound_wo_separable},
\item
$-9/2$ times \eqref{eq_sum_of_rs_inequality_unsep}
\item
$-3/2$ times \eqref{eq_memory}, and
\item
$3/2$ times \eqref{eq_from_second_discoordination},
\end{enumerate}
which yields
\begin{equation}\label{eq_hybrid_wo_separable_with_rs}
6M+5R \ge 11  
{\mygreen + (15/2)r_5 .}
\end{equation}
\end{proof}

\subsection{Proof of Theorem~\ref{th_general_N_K_three}}

The proof below attempts to keep most of the same notation
as in the proof of Theorem~\ref{th_separated_N_K_three}.

\begin{proof}[Proof of Theorem~\ref{th_general_N_K_three}]
Let notation 
be as in Definition~\ref{th_general_N_K_three},
{\mygreen (which defines $A^i_j,B^i_j,C^i_j$ for $1\le i\le 4$ and
$1\le j \le 3$, and the ratios $(r_1,\ldots,r_5)$.)}
(We will use a mostly different choice of basis elements of $W$ than
{\mygreen we did in Section~\ref{se_Joel_hybrid}).}
For $j=1,2,3$, let
$$
A_j = A^1_j+\cdots+A^4_j.
$$
Hence $A_1,A_2,A_3\subset W_1$, but we do not assume that they
are independent.
Let $\cA^5$ be a basis of $W_1$ relative to $A_1+A_2+A_3$.
Similarly define $B_i,C_i$ for $i=1,2,3$, and $\cB^5,\cC^5$.

{\mygreen Before defining the rest of the bases we use for $W$, let us
describe} 
the rough idea behind our proof: {\mygreen the idea}
is to use our hybrid bound
by first considering the $\cA^5,\cB^5,\cC^5$ part that $X_{123}$
must contain, and 
then applying a form
of Tian's method to what is left.  
It follows that we use Tian's method on the three-way part of the
scheme, which means that
\eqref{eq_hybrid_tian_type_bound_wo_separable} will
have a $-3r_4$ coefficient that is not present in
\eqref{eq_hybrid_tian_type_bound}.
Our approach to dealing with the fact that $A_1,A_2,A_3\subset W_1$
may not be independent is to 
apply Theorem~\ref{th_main_three_subspaces_decomp} and to
use a discoordination minimizer, $X$, of $A_1,A_2,A_3$ to 
us to write down a convenient basis of $A_1+A_2+A_3$ to apply
Tian's method.
What then happens, roughly speaking, is that
any dependence between $A_1,A_2,A_3$ will force 
the dimension of $A_1+A_2+A_3$ to be smaller than the sum of
$\dim(A_i)$, and hence force the values of
$r_1,\ldots,r_5$ to be larger (which can only improve our lower bound)
than what would be if $A_1,A_2,A_3$ are
independent.
(And similarly with $A$ replaced everywhere by $B$, or by $C$.)

{\mygreen The above rough ideas motivate our}
need to introduce notation and a different type
of basis {\mygreen than in Section~\ref{se_Joel_hybrid}:
namely, the bases we introduce here} express how
any dependence between $A_1,A_2,A_3$ arises, 
and for that we
apply Theorem~\ref{th_main_three_subspaces_decomp}
to $A_1,A_2,A_3$
(and similarly for $B_1,B_2,B_3$ and $C_1,C_2,C_3$).
{\mygreen Here are the precise bases; we start with $A_1,A_2,A_3$.}

So 
apply to Theorem~\ref{th_main_three_subspaces_decomp}
with $A=A_2$, $B=A_3$, and $C=A_1$ (it is important to take $C=A_1$)
and $\cU=A_1+A_2+A_3$; this yields a decomposition of $\cU\subset W_1$ 
into subspaces
$\cU_1,\cU_2$ with the properties stated in the theorem.
In particular, there exists a basis $\cA$ of $\cU_1$ that
coordinates $A_i\cap \cU_1$, and there exists a basis
$y_1,\ldots,y_m$ of $A_2\cap\cU_2$ and
$y_1',\ldots,y_m'$ of $A_3\cap\cU_2$ such that
the $y_i+y_i'$ are a basis of $A_1\cap \cU_2$
(where $m$ is the discoordination of $A_1,A_2,A_3$).

Given these bases, introduce the following notation: 
\begin{enumerate}
\item $\cA'=\{y_1,\ldots,y_m\}\cup\{y_1',\ldots,y_m'\}$
(hence $\cA\cup\cA'$ is a basis of $A_1+A_2+A_3$, and
$\cA\cup\cA'\cup \cA^5$ a basis of $W_1$);
\item 
$\cA_{123} = \cA\cap A_1\cap A_2\cap A_3$;
\item 
$\cA_{12{\rm \,only}}= (\cA\cap A_1\cap A_2) \setminus \cA_{123}$,
and similarly for $\cA_{13{\rm \,only}}$ and 
$\cA_{23{\rm \,only}}$;
\item 
$\cA_{1{\rm \,only}}= (\cA\cap A_1)\setminus
(\cA_{12{\rm \,only}}\cup \cA_{13{\rm \,only}})$, and similarly for
$\cA_{2{\rm \,only}}$ and 
$\cA_{3{\rm \,only}}$.
\end{enumerate} 
Since $\cA$ coordinates the $A_i\cap\cU_1$, and $\cA$ spans only
their sum, it follows (from Theorem~\ref{th_easy_minimizer_facts})
that each element of $\cA$ lies in at least one of the $A_i\cap\cU_1$.
Hence the sets
\begin{equation}\label{eq_cA_partition}
\cA_{123}, \ \cA_{12{\rm \,only}}, \ \cA_{13{\rm \,only}}, 
\ \cA_{23{\rm \,only}}, \ \cA_{1{\rm \,only}},
\ \cA_{2{\rm \,only}}, \ \cA_{3{\rm \,only}}
\end{equation} 
form a partition of $\cA$.
{\mygreen $\cA'$ can be partitioned into its subsets}
\begin{equation}\label{eq_cA_prime_partition}
\cA'_2=\{y_1,\ldots,y_m\}, \quad
\cA'_3=\{y_1',\ldots,y_m'\} .
\end{equation}
Since $\cA^5$ is a basis of $W_1$ relative to $\cU=A_1+A_2+A_3$,
$\cA^5$ and the union of
\eqref{eq_cA_prime_partition} and
\eqref{eq_cA_partition} form a basis for $W_1$.

Now form a similar basis for $W_2$ with $B$'s replacing the $A$'s,
making sure that $\cB'=\cB'_1\cup\cB'_3$
(the subscripts $1$ and $3$ are important), so on the discoordinated
part of $B_1+B_2+B_3$ we take $\cB'$ to consist of basis elements
of $B_1$ and $B_3$.
Then form a similarly basis for $W_3$ with $C$'s, similarly with
$\cC'=\cC'_1\cup\cC'_2$.

We will
exploit any dependence between $A_1,A_2,A_3$ in the following equation:
\begin{equation}\label{eq_cA_five_dimesion}
|\cA^5| = 3F r_5 = F - \dim(A_1+A_2+A_3).
\end{equation} 
Note that 
$$
\dim(A_1)=m+|\cA_{1{\rm \,only}}|+
|\cA_{12{\rm \,only}}|+
|\cA_{13{\rm \,only}}|+ |\cA_{123}|,
$$
and similarly for $\dim(A_2)$ and $\dim(A_3)$; since
$\dim(A_1+A_2+A_3)$ is the sum of the size of the sets
in 
\eqref{eq_cA_partition} and \eqref{eq_cA_prime_partition},
it follows that
$$
\dim(A_1+A_2+A_3)=
\dim(A_1)+\dim(A_2)+\dim(A_3) 
- m - \sum_{i<j}|\cA_{ij{\rm \,only}}| 
- 2 |\cA_{123}|
$$
$$
= 3F(r_1+r_2+r_3+r_4)
- m - \sum_{i<j}|\cA_{ij{\rm \,only}}| 
- 2 |\cA_{123}|
$$
Applying \eqref{eq_cA_five_dimesion}
we get
$$
3F(r_1+r_2+r_3+r_4+r_5) - F
=  
 m + \sum_{i<j}|\cA_{ij{\rm \,only}}| 
+ 2 |\cA_{123}| .
$$
{\mygreen
By the symmetry of the scheme, we have $|\cA_{ij{\rm \,only}}|=
|\cA_{23{\rm \,only}}|$ for any $1\le i<j\le 3$, and we may rewrite
the above equation as}
\begin{equation}\label{eq_extra_rs_offset_by_dependence_of_As}
3F(r_1+r_2+r_3+r_4+r_5-1/3)
=  
 m + \sum_{i<j}|\cA_{ij{\rm \,only}}| 
+ 2 |\cA_{123}| = m + 3|\cA_{23{\rm \,only}}| + 2 |\cA_{123}|.
\end{equation} 
So if $A_1,A_2,A_3$ are not independent,
then some of the quantities on the right-hand-side above must be nonzero,
which forces $r_1+\cdots+r_5$ to be larger than $1/3$.
Since clearly
$$
|\cA_{23{\rm \,only}}| + |\cA_{123}| 
\le 
\bigl( m + 3|\cA_{23{\rm \,only}}| + 2 |\cA_{123}| \bigr) /2 ,
$$
considering the right-hand-side of 
\eqref{eq_extra_rs_offset_by_dependence_of_As} yields
\begin{equation}\label{eq_bound_on_what_will_be_cE}
|\cA_{23{\rm \,only}}| + |\cA_{123}|
\le
3F(r_1+r_2+r_3+r_4+r_5-1/3)/2,
\end{equation} 
which is weaker than \eqref{eq_extra_rs_offset_by_dependence_of_As}
but sufficient for our needs.

Now we proceed similarly to the 
proof of Theorem~\ref{th_separated_N_K_three}.
So let $\iota\from W\to \field^{3F}$ be the isomorphism 
expressing an
element of $W$ in terms of its coefficients in $\cA\cup\cB\cup\cC$.
We similarly let $\cW^5=\cA^5\cup\cB^5\cup \cC^5$, which allows
us to infer that $X_{123}$ contains vectors whose image under $\iota$
whose span is the row space of a block matrix $[L\ I]$,
where $I$ is an identity matrix of size $\dim(\cW^5)$;
here $L$ plays the role of $[L_1\ \cdots\ L_4]$ in
\eqref{eq_iota_X_fifth_part}.

Now we go straight to the Tian style argument: we 
take the vectors in $X_{123}$ corresponding to $[L\ I]$ above,
and add independent vectors of $X_{123}$ each of which is needed by
some $Z_i$ to infer $W_i$ from $X_{123}$.  This gives us a basis
for a subspace $X'_{123}$ such that $X'_{123}$ and $Z_i$ implies $W_i$,
and $\iota$ of the basis vectors forms a matrix of the form
$$
G = 
\begin{bmatrix}
G'' & 0 \\
L & I
\end{bmatrix}
$$
with $L,I$ as above (the only crucial observation is that there is
a $0$ above the $I$, which occurs since each new row corresponding
to a $X'_{123}$ vector is needed by some $Z_i$ and therefore cannot
include any $\cW^5$ component.
In particular, similar to 
\eqref{eq_FR_versus_rank_G_double_prime} (but looking only at
$\cW^5$ rather than both $\cW^5$ and $\cW^4$) we have
$$
FR \ge R''F + 9 r_5 F,
\quad\mbox{where}\quad R'' = {\rm Rank}(G'')/F.
$$

Now we consider $Z_1,X'_{123}$.  First we note that $Z_1$ contains no $\cW^5$
component.  Second we note that the $Z_1$ contains all vectors in the
following parts:
$$
\cB_{123}, \ \cB_{12{\rm \,only}}, \ \cB_{13{\rm \,only}}, 
\ \cB_{1{\rm \,only}},
\ \cB_1'
$$
and the case with $\cC$ replacing $\cB$ everywhere.
Since $Z_1$ and $X'_{123}$ allows us to deduce all vectors in $\cW^1$,
Tian's argument shows that setting
$$
\cD_1 = \{  
\cB_{2{\rm \,only}}, 
\ \cB_{3{\rm \,only}}, 
\ \cB_{23{\rm \,only}}, 
\ \cB'_3,
\ \cC_{2{\rm \,only}}, 
\ \cC_{3{\rm \,only}}, 
\ \cC_{23{\rm \,only}}, 
\ \cC'_2
\},
$$
we have that
the columns of $G''$ corresponding to the columns in $\cD_1$ have
the dimension of their column space to be at most
$$
{\rm Rank}\bigr(G''|_{\cD_1} \bigr) \le M+R - F(1+2r_1+2 r_2+ (3/2) r_3 + r_4 ).
$$
Considering $Z_2$ and $X'_{123}$,
the same remark holds with $\cD_1$ replaced by
$$
\cD_2 = \{  
\cA_{1{\rm \,only}}, 
\ \cA_{3{\rm \,only}}, 
\ \cA_{13{\rm \,only}}, 
\ \cA'_3,
\ \cC_{1{\rm \,only}}, 
\ \cC_{3{\rm \,only}}, 
\ \cC_{13{\rm \,only}}, 
\ \cC'_1
\}
$$
and similarly with
$$
\cD_3 = \{  
\cA_{1{\rm \,only}}, 
\ \cA_{2{\rm \,only}}, 
\ \cA_{12{\rm \,only}}, 
\ \cA'_2,
\ \cB_{1{\rm \,only}}, 
\ \cB_{2{\rm \,only}}, 
\ \cB_{12{\rm \,only}}, 
\ \cB'_1
\}
$$
Setting $\cD=\cD_1\cup\cD_2\cup \cD_3$, we have
$$
{\rm Rank}\bigr(G''|_{\cD} \bigr) 
\le 
3  \Bigl( M+R - F(1+2r_1+2 r_2+ (3/2) r_3 + r_4 ) \Bigr).
$$
Now we notice that the only basis elements of $\cW\setminus\cW^5$ that
do not appear in $\cD$ are
$$
\cE = \{
\cA_{23{\rm \,only}}, \cA_{123},
\cB_{13{\rm \,only}}, \cB_{123},
\cC_{12{\rm \,only}}, \cC_{123}
\},
$$
and note that, by symmetry
$$
|\cE| = 3|\cA_{23{\rm \,only}}| + 3|\cA_{123}|,
$$
and hence 
\eqref{eq_bound_on_what_will_be_cE} implies that
\begin{equation}\label{eq_bound_on_cE}
|\cE| \le 
(9/2) F(r_1+r_2+r_3+r_4+r_5-1/3).
\end{equation} 
It follows that
$$
FR'' = {\rm Rank}(G'') 
\le
{\rm Rank}\Bigl( G''|_\cD \Bigr) + 
{\rm Rank}\Bigl( G''|_\cE \Bigr) 
$$
$$
\le
3  \Bigl( MF+R''F - F(1+2r_1+2 r_2+ (3/2) r_3 + r_4 ) \Bigr) + |\cE| .
$$
Hence
$$
3M+2R'' \ge 3 \Bigl(
F \bigl( 1+2r_1+2 r_2+ (3/2) r_3 + r_4 \bigr) 
-|\cA_{23{\rm \,only}}|-
|\cA_{123}| \Bigr) .
$$
Since $R\ge R''+ 9r_5$ we get, using
\eqref{eq_bound_on_what_will_be_cE} we get
\begin{equation}\label{eq_almost_there_for_unsep}
3M+2R \ge 3 \bigl( 1+2r_1+2 r_2+ (3/2) r_3 + r_4 \bigr) + 9r_5
- (9/2)  
(r_1+r_2+r_3+r_4+r_5-1/3).
\end{equation} 
In view of 
\eqref{eq_bound_on_cE} or 
\eqref{eq_bound_on_what_will_be_cE} we have
\begin{equation}\label{eq_sum_of_rs_at_least_a_third}
r_1+r_2+r_3+r_4+r_5 - 1/3 \ge 0,
\end{equation}
we may replace the factor of $9/2$ in 
\eqref{eq_almost_there_for_unsep} by anything larger,
and setting it to $6$ we get
$$
3M+2R \ge 3 \bigl( 1+2r_1+2 r_2+ (3/2) r_3 + r_4 \bigr) + 9r_5
- 6  
(r_1+r_2+r_3+r_4+r_5-1/3)
$$
$$
= 5 - (3/2)r_3 - 3r_4 + 3r_5.
$$

We remark that in view of 
\eqref{eq_bound_on_what_will_be_cE}, the inequality 
\eqref{eq_sum_of_rs_at_least_a_third} is strict unless
$m,|\cA_{23{\rm \,only}}|,|\cA_{123}|$ are all $0$.
This only happens if $A_j=\cA_{j{\rm \,only}}$ for $j=1,2,3$,
which implies that $A_1,A_2,A_3$ are independent.
\end{proof}

\section{Concluding Remarks}
\label{se_Joel_conclusion}

Let us indicate some directions for future research.

The first direction involves whether or not one can generalize our
main theorem, Theorem~\ref{th_main_three_subspaces_decomp},
to four or more subspaces, $A_1,\ldots,A_m$ with $m\ge 4$,
to obtain some sort of decomposition of the ambient $\field$-universe,
$\cU$, into a part where the subspaces are coordinated and other
parts that have a sort of canonical form.
Perhaps there are also non-trivial relationships between the
discoordination of different collections of subspaces, either all
in the
original universe or some in quotient universes.
Let us state a related question for $m\ge 4$ subspaces
that seems very interesting.

Given a set subspaces $\cA=\{A_1,\ldots,A_m\}$ of an $\field$-universe,
define the {\em closure} of $\cA$ to be the set of all subspaces
that can be expressed by a formula involving $+,\cap$ and 
elements of $\cA$ (and parenthesis).
We remark that if $A_1,\ldots,A_m\subset \field^n$ are coordinate 
subspaces, so $A_j=\field^{I_j}$ with $I_j\subset[n]$, then the
closure of these subspaces are all subspaces of the form $\field^I$
where $I\subset[n]$ and $I$ can be expressed as a formula involving
$\cap,\cup$ and the 
$I_1,\ldots,I_m$ (and parenthesis).  Considering the Venn diagram
of the $I_1,\ldots,I_m$, we see that the size of this closure is
bounded as a function of $m$.
Hence the same holds if $A_1,\ldots,A_m$ are coordinated subspaces
of some $\field$-universe.
Theorem~\ref{th_main_three_subspaces_decomp} implies
that the closure of a set $\{ A_1,A_2,A_3\}$ is also bounded by
universal constant, since applying $+,\cap$ to
$$
e_1\otimes\field^m, \ e_2\otimes\field^m, \ (e_1+e_2)\otimes\field^m
$$
yields either $0$ or all of $\field^2\otimes\field^m$.
Hence, we wonder if one can give a bound on the closure of
$\cA=\{A_1,\ldots,A_m\}$ for $m\ge 4$ that depends only on $m$;
at present we do not even know if this closure is necessarily finite
for $m=4$.

As mentioned before, another linear algebra question would be to
generalize the decomposition lemma, Lemma~\ref{le_Z_scheme_decomposition},
when the universe has a decomposition into $N\ge 4$ subspaces.
We remark that we discovered the achievability of $(M,R)=(1/2,5/3)$
after proving Lemma~\ref{le_Z_scheme_decomposition}, so we believe
that a generalization of this lemma may give new achievable
memory-rate pairs.
We also mention that our proof of
Lemma~\ref{le_Z_scheme_decomposition} seems long and tedious,
so we hope that future work, either for $N=3$ or $N\ge 4$, 
would eventually come with simpler proofs based on some new concepts.

As mentioned in Section~\ref{se_Joel_new_point}, even if we
take the $Z_i$ to be a fairly simple scheme, such as a separated
scheme consisting entirely of one of the
pure schemes in Definition~\ref{de_Z_schemes},
we don't know of any reasonable algorithm to determine the
corresponding $X_{\mec d}$ of minimum dimension.

Regarding coded caching for $N=K=3$, our new achievable point
$(M,R)=(1/2,5/3)$ shows that Tian's bound $2M+R \ge 8/3$ is
tight for $1/3\le M\le 1/2$, but leaves open $1/2 < M < 1$.
We wonder if one can add discoordination bounds to
Tian's type of computer-aided search and get improved results.

We are, of course, interested to know if the equation
$$
I(A;B;C) = \dim(A\cap B\cap C) - {\rm DisCoord}(A,B,C)
$$
and the many other equalities involving the ${\rm DisCoord}(A,B,C)$
could have new applications in information theory under the assumption
that the random variables involved are linear.

Finally we wonder if there are analogs of the above formula for
$I(A;B;C)$ when $A,B,C$ are not assumed to be linear, and
of the mutual information of more than three random variables.

\providecommand{\bysame}{\leavevmode\hbox to3em{\hrulefill}\thinspace}
\providecommand{\MR}{\relax\ifhmode\unskip\space\fi MR }
\providecommand{\MRhref}[2]{%
  \href{http://www.ams.org/mathscinet-getitem?mr=#1}{#2}
}
\providecommand{\href}[2]{#2}

\end{document}